\documentclass[12pt,a4paper]{article}
\usepackage[top=1in,bottom=1in,left=1in,right=1in]{geometry}
\usepackage{epsfig}
\usepackage{hyperref}
\usepackage{algorithmic}
\usepackage{authblk}
\usepackage{algorithm}
\usepackage{setspace}
\usepackage{cite}
\usepackage{graphicx}
\usepackage{subfigure}
\usepackage{xcolor}
\usepackage{caption}
\usepackage{appendix}
\usepackage{mathrsfs}
\usepackage{amssymb,amsmath,amsthm}
\newtheorem{theorem}{Theorem}

\newtheorem{example}{Example}

\usepackage{mathtools} 

\usepackage{bm}
\usepackage{multicol}
\usepackage{titlesec}
\usepackage{array}

\usepackage[utf8x]{inputenc}
\usepackage[mathscr]{euscript}
 \usepackage{multirow}
 \usepackage [english]{babel}
 \usepackage [autostyle, english = american]{csquotes}
\title{Logic-dependent emergence of multistability, hysteresis, and biphasic dynamics in a ``minimal'' positive feedback network with an autoloop} 

\author[1,2]{Akriti Srivastava\thanks{akriti.iitp@gmail.com}}
\author[1]{Mubasher Rashid\thanks{\textbf{Corresponding Author: Mubasher Rashid (mubasherrashid@gmail.com or mubasher@iitk.ac.in)}}}

\affil[1]{Department of Mathematics and Statistics, Indian Institute of Technology Kanpur, Kanpur-208016, India}
\affil[2]{Department of Mathematics, Dayananda Sagar College of Engineering, Bangalore-560078, India}

\date{}
\begin{document}
	\maketitle
	\begin{abstract}
\noindent Cellular decision-making (CDM) is a dynamic phenomenon crucial for development and diseases. CDM is often, if not invariably, controlled by regulatory networks defining interactions between genes and transcription factor proteins. Traditional studies have focussed on molecular switches such as positive feedback circuits that upon multimerization exhibit at most bistability. However, higher-order dynamics such as tristability is also prominent in many biological processes. It is thus imperative to identify a "minimal" circuit that can alone explain mono, bi, and tristable dynamics. In this work, we consider a two-component positive feedback network with an autoloop and explore these regimes of stability for different degrees of multimerization and the choice of Boolean logic functions. We report that this network can exhibit numerous dynamical scenarios such as bi-and tristability, hysteresis, and biphasic kinetics, explaining the possibilities of abrupt cell state transitions and the $smooth \ state \ swap$ without a step-like switch. 
Specifically, while with monomeric regulation and competitive OR logic, the circuit exhibits mono-and bistability and biphasic dynamics, with non-competitive AND and $\mathcal{OR}$ logics only monostability can be achieved. To obtain bistability in the latter cases, we show that the autoloop must have (at least) dimeric regulation. In pursuit of higher-order stability, we show that tristability occurs with higher degrees of multimerization and with non-competitive $\mathcal{OR}$ logic only. Our results, backed by rigorous analytical calculations and numerical examples, thus explain the association between multistability, multimerization, and logic in this "minimal" circuit. Since this circuit underlies various biological processes, including epithelial-mesenchymal transition which often drives carcinoma metastasis, these results can thus offer crucial inputs to control cell state transition by manipulating dimerization and the logic of regulation in cells.

	\end{abstract} 
 

\section{Introduction}
Cellular decision-making is a cell non-autonomous process where at each cell-lineage branch point, a cell drives into one of the alternative distinct cell types \cite{cell-fate-roy-soc-inter}. Investigating dynamical principles of regulatory networks that govern cellular decision-making is essential to understanding and controlling stepwise lineage decisions of cells \cite{und-gene_cir_cell-fate_branch_pts, Syn-gencir_cellDecMak, CellDecMak_JJ-collins}. A key aspect of this decision-making process is the ability of cells to manifest multiple stable states or phenotypes in response to varying internal and external cues without altering their genetic makeup \cite{multiStab_dec-mak_in_diff, cell-fate-roy-soc-inter}. This phenomenon, known as multi-stability, is pivotal to cellular differentiation \cite{multiStab_dec-mak_in_diff, differ-ruiqi}, reprogramming \cite{und-gene_cir_cell-fate_branch_pts, cell-fate-roy-soc-inter}, and its manifestation through epithelial-mesenchymal plasticity  \cite{Phys-of-cell-dec-mak-in-EMT}, a cellular program enabling bidirectional transition among epithelial, mesenchymal, and one or more hybrid epithelial-mesenchymal cells, enables intra-tumor heterogeneity \cite{EMT-heterog_tripathi, Heter&Plast}, induces cancer progression and metastasis of carcinomas \cite{phenotypicHetero_Science_adv, myPlosCompaper}. Since cellular decision-making is largely controlled by regulatory networks defining “molecular switches” \cite{why-are-cell-switch-boolean, necess-cond-kauffman}, unraveling the dynamical behaviors and logic of multi-stable switches has profound implications in synthetic biology and regenerative medicine.

A necessary condition for nonlinear control systems and networks to exhibit multistability is the presence of positive feedback loops \cite{Angeli&Sontag2004, sontag-1, sontag-2, graphic-req-multi, necess-cond-kauffman}. Interestingly, this requirement also extends to biological systems where a commonly observed network structure underlying multistability is the toggle switch \cite{Toggle-Switch-JJCollins}, comprising mutual inhibition of two opposite fate-determining transcription factors thus forming a positive feedback loop. This mutual repression enables cells to adopt different states, driving an 'either-or' binary choice between alternative cell fates from a common progenitor \cite{und-gene_cir_cell-fate_branch_pts}. However, it is rational to posit that cellular decision-making doesn't rely solely on binary outcomes. The intricate structure of gene regulatory networks enables them to exhibit multiple alternative stable states. Examples of, for instance, bi-and tristable solutions are found across biological contexts including development \cite{develop}, differentiation \cite{multistable-switches-differen, diff-1, diff-2, tristability-mouse}, epigenetic processes \cite{epigene}, and metastasis of carcinomas \cite{Tristability-in-miRNA-TF-Switch, microRNA-based, CombCoop-ov}. Often, a two-component system operates at the core of the underlying bistable decision-making circuits \cite{und-gene_cir_cell-fate_branch_pts, chick, genetic-bistable-switch}. It is imperative to identify the "minimal" (in terms of the number of variables and the regulatory interactions) two-component genetic circuits capable of demonstrating bistability. Two such networks are (i) circuits featuring multimeric regulation and lacking autoloop, and (ii) circuits employing monomeric regulators along with a singular monomeric autoloop \cite{Cherry&Adler2000, why-are-cell-switch-boolean}. Jules et al., \cite{Jules_BMB} and Zhang et al., \cite{Zhang_mono&BiTrans-specific-case} analytically proved that such systems can exhibit at most bistability. This begs a question: are low-dimensional networks such as two-component circuits capable of only bistable behavior or higher-order stability such as tristability (three states) can also be possible in these circuits? It may be noted that compared to bistable dynamics, networks exhibiting tristable dynamics are not much studied due to the complexity of analytical calculations as well as numerical simulations. 

Recent theoretical and experimental studies found that miR200-ZEB positive feedback circuit with one multimeric autoloop, which underlies epithelial-mesenchymal transition, exhibits tristability \cite{CombCoop-ov, microRNA-based, Tristability-in-miRNA-TF-Switch, mir200-Zeb_motor}. It allows reversible transition among epithelial, mesenchymal, and hybrid epithelial-mesenchymal cells and is a crucial driver of invasion and metastasis of carcinomas \cite{mir200-zeb_genetic-dissection}. However, these studies are purely based on a numerical approach and lack mathematical rigor. Also, they don't explore any association of multimerization with tristability. Motivated by this, we asked a general question: what can be the minimal "multimeric" requirements for a two-component feedback network with one autoloop (Fig. \ref{schematic}a) to exhibit mono-and bistability and under what conditions can it exhibit tristability? Further, node $a$ is co-regulated by node $b$ and node $a$ itself through autoloop. When a gene is regulated by more than one transcription factor, the combinatorial rule (Boolean logic) must also be considered. We considered three types of regulatory logic: (a) Competitive OR logic, (b) non-competitive AND logic, and (c) non-competitive OR logic, and explored their effect on the order of stability. For simplicity, throughout the text, we denote these three logics as Type-I OR, AND, and Type-II $\mathcal{OR}$, respectively.

In this work, we thus analyze the dynamics of such a generic, minimal network (depicted in Fig. \ref{schematic}a) by systematically finding the analytical conditions for mon, bi, and tristable dynamical regimes. We model the network using a system of coupled ordinary differential equations for which each stable equilibrium solution corresponds to a possible stable state (or phenotype or attractor). We explore how these solutions are affected by the degree of multimerization as well as the type of regulatory logic. We integrated the autoloop with the core positive feedback network using three regulatory logics: AND gate, Type-I OR gate, and Type-II $\mathcal{OR}$ gate, each having a different mathematical formalism and qualitative behavior (Fig. \ref{schematic}b). Starting with monomeric regulation, we systematically tested higher degrees of multimerization in each model. We report that multistability in this network is logic-dependent. In monomeric regulation with monomeric autoloop, while type-I OR logic exhibits at most bistability, AND and type-II $\mathcal{OR}$ logic can exhibit only monostability. To obtain bistability with AND and Type-II $\mathcal{OR}$ logics, we show that a dimeric autoloop is required. Upon testing different combinations of multimerization, we show that tristability can occur in higher degrees of multimerization, though bistability is persistently prevalent irrespective of the logic of regulation. Our results thus comprehensively provide an association between multistability, multimerization, and logic in this circuit. Since this circuit underlies various biological processes including EMT in cancers and the multistability can be mapped to phenotypic heterogeneity, these results can thus offer crucial inputs to control phenotypic heterogeneity by manipulating dimerization and synthetically (re)engineering the logic of regulation in cells.
\begin{figure}[H]
 \centering
\subfigure[]{\includegraphics[scale=0.7]{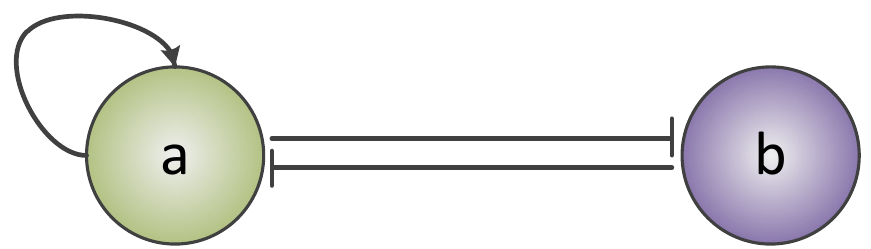}}
\end{figure}
\begin{figure}[H]
 \centering
 \subfigure[]{\includegraphics[scale=1]{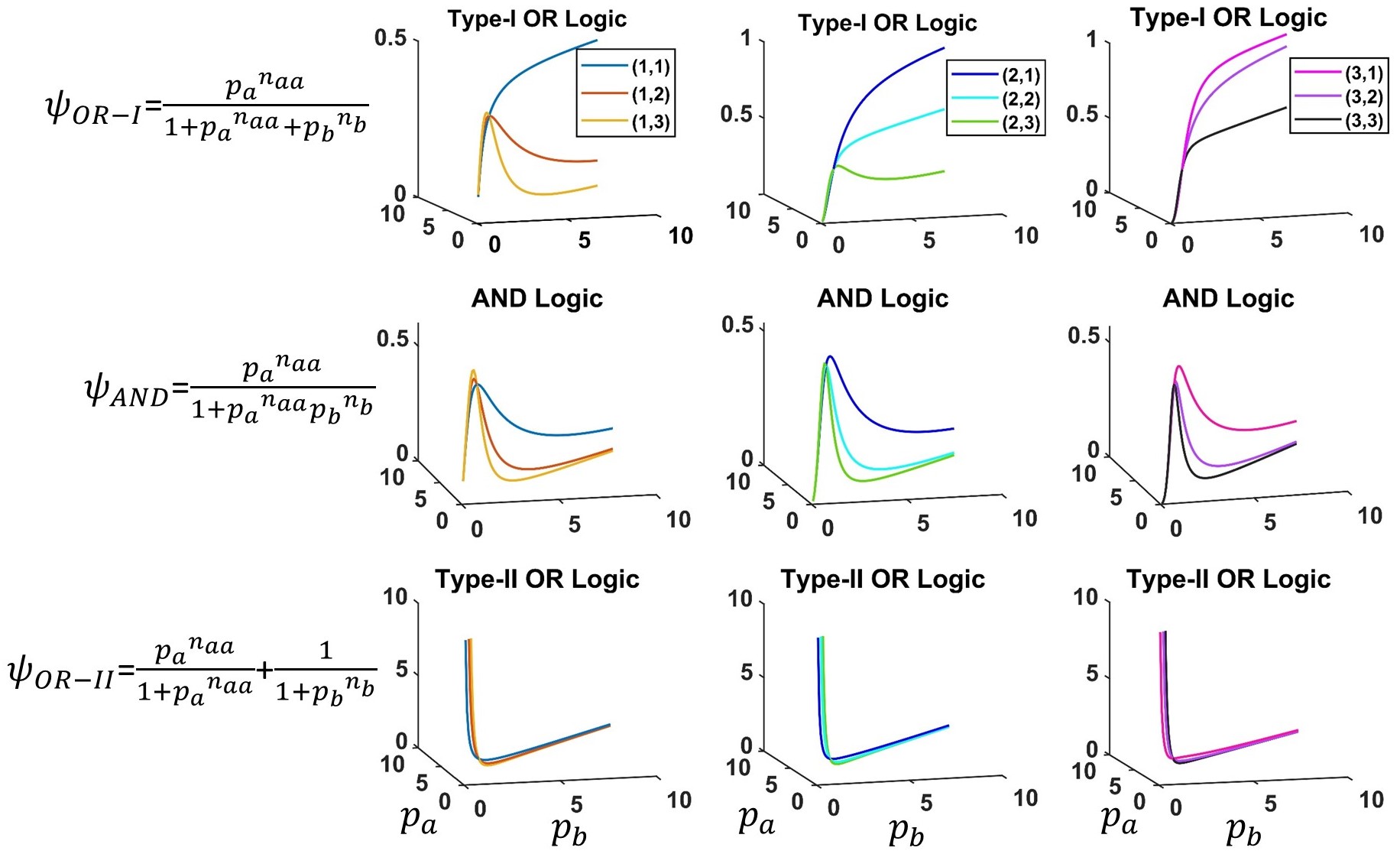}}
  \caption{Illustration of auto-regulated positive feedback network and gate function response curves. (a) Genes $a$ and $b$ mutually inhibit each other to form a positive feedback circuit. Also, gene $a$ can promote its expression (directly or indirectly), which is represented by autoloop. An example of a biological network corresponding to this schematic diagram is the miR200-ZEB mutual inhibition circuit underlying EMT and driving metastasis of carcinomas \cite{CombCoop-ov, microRNA-based}. (b) Qualitative behavior of three types of Boolean logic functions modeling co-regulation of gene $a$ (activation by $a$ and inhibition by $b$) for different values of ($n_{aa}$, $n_b$). Type-I OR logic shows mixed increasing/decreasing responses for different multimerization conditions (i.e., Hill coefficients; $n_{aa}$, $n_b$). The AND and type-II $\mathcal{OR}$ logic functions show a unified response independent of multimerization. Also, it is evident that type-I OR and type-II $\mathcal{OR}$ logic functions differ qualitatively. The former in most cases is increasing and saturating while the latter is always monotonically decreasing.}
  \label{schematic}
 \end{figure}
 
\section{Organization of the paper}
This paper contains three main sections based on the three different types of Boolean logic used to model the co-regulation of gene $a$ by two transcription factors. In section 3, this co-regulation is modeled using competitive OR logic. In section 4, the non-competitive AND logic is employed. In section 5, a non-competitive $\mathcal{OR}$ logic is considered to model the co-regulation. For brevity, throughout the text, competitive OR logic, non-competitive AND logic, and non-competitive OR logic are referred to as Type-I OR logic, AND logic, and Type-II $\mathcal{OR}$ logic, respectively. In each section, we explore the minimal multimerization conditions (manifested through Hill coefficient values) for the occurrence of mono, bi-and multistability. We have applied the nullcline approach to show the existence of steady states. For each model, we analyze the associated characteristic equations and obtain the conditions for the stability of the steady states and the existence of saddle-node bifurcations. Numerical simulations are then performed to illustrate the results. Through bifurcations, we also present possible ways of state transitions such as step-like abrupt state transition (a manifestation of saddle-node bifurcation) and biphasic $smooth \ state \ swaps$ (manifested through sigmoidal curves). In Section 6, we discuss the impact of basal production rate (leakage) on each model's dynamics. Finally, section 7 presents a detailed discussion and concluding remarks based on our findings.

\section{Mathematical Model [Type-I OR Logic]} \label{MM_OR}
The model under investigation in this paper is depicted 
in Fig. \ref{schematic}. The complete nonlinear mathematical model uses four state variables for the concentrations of mRNAs and transcription factor proteins (TF). The concentration of mRNA produced by genes $a$ and $b$ is denoted by $r_a$ and $r_b,$ respectively, while the corresponding TF protein concentrations are denoted by $p_a$ and $p_b.$  The inhibition of gene $a$ and $b$ by TF $p_b$ and $p_a$, respectively, is modelled by the Hill function. Also, the activation of gene $a$ by its own TF, $p_a$, is modeled by the Hill function with a different Hill coefficient, $n_{aa}$. The co-regulation of gene $a$ by TF $p_b$ and also by its own TF $p_a$ is modeled using competitive OR logic \cite{logics}. The transcription of genes $a$ and $b$ into mRNA $r_a$ and $r_b$ is modeled using the non-linear Hill function while the translation of these mRNA's into TF proteins $p_a$ and $p_b$ is modeled using linear mass-action kinetics.

Using the above formalism, the following system of coupled nonlinear ordinary differential equations describe the dynamics of the network in Fig. \ref{schematic}. \\
		\begin{equation}\label{1}
			\begin{aligned}
				&\frac{dr_a}{dt}   =  m_a\frac{(\frac{p_a}{\theta_{aa}})^{n_{aa}}}{1+(\frac{p_a}{\theta_{aa}})^{n_{aa}}+(\frac{p_b}{\theta_b})^{n_{b}}}-\gamma_a r_a+A_1, \\
		&\frac{dr_b}{dt}  =m_b\frac{1}{1+(\frac{p_a}{\theta_a})^{n_{a}}}-\gamma_b r_b+B_1, \\
				&\frac{dp_a}{dt} =k_ar_a-\delta_ap_a,\\
					&\frac{dp_b}{dt} =k_br_b-\delta_bp_b,
			\end{aligned}
		\end{equation}
		with $r_a(0)\geq 0,\,r_b(0)\geq 0,\,p_a(0)\geq 0,\,p_b\geq 0,$
  where $m_i$ are the maximum transcription rates, $\gamma_i$ are the mRNA degradation rates, $k_i$ are the translation rates, and $\delta_i$ are the protein degradation rates for genes $i=a,\,b$. $A_1$ and $B_1$ represent the basal production rates (in the absence of activation and inhibition) of genes $a$ and $b,$ respectively. The threshold parameters $\theta_a, \, \theta_b, \, \theta_{aa}$ denote the thresholds of $p_a$ and $p_b$ to induce a significant response of $r_a$ and $r_b$. The integer parameters $n_a,\,n_b$  and $n_{aa}$ are Hill coefficients that determine the steepness of Hill curves.

\subsection{Mathematical Analysis}\label{2.1}
 This section presents the existence and stability analysis of the steady states for the model system (\ref{1}). We use the Routh-Hurwitz criterion \cite{perko_2013} to prove the local asymptotic stability of the steady states, while Sotomayor's theorem \cite{perko_2013} is used to derive the transversality conditions of saddle-node bifurcation. 
 
 We obtained the steady state $S^*=(r_a^*,r_b^*,p_a^*,p_b^*)$ by setting the right hand side of the model (\ref{1}) to zero.
	From third and fourth equations of model (\ref{1}), we obtained ${r_a}^*=\frac{\delta_a {p_a}^*}{k_a}$ and ${r_b}^*=\frac{\delta_b {p_b}^*}{k_b},$ respectively, as $\frac{dp_a}{dt}=0$ and $\frac{dp_b}{dt}=0.$ Using  $\frac{dr_b}{dt}=0$ and then putting the value of $r_b,$ we get ${p_b}^*=\frac{1}{k_2}\bigg\{B_1+m_b\frac{\theta_a^{n_a}}{\theta_a^{n_a}+{{p_a}^*}^{n_a}}\bigg\}.$ Finally, using the value of $p_b$ and $r_a$ in $\frac{dr_a}{dt}=0,$ we get a polynomial in $p_a$ as
	\begin{equation}\label{2}
   \begin{split}
\bigg(m_a p_a^{n_{aa}}+(A_1-k_1p_a)(\theta_{aa}^{n_{aa}}+p_{a}^{n_{aa}})\bigg)\theta_b^{n_b}k_b^{n_b}(\theta_a^{n_a}+p_a^{n_a})^{n_b}+(A_1-k_1p_a)&\theta_{aa}^{n_{aa}}\bigg(B_1(\theta_a^{n_a}+p_a^{n_a})\\+m_b\theta_a^{n_a}\bigg)^{n_b}=0,
 \end{split}
	\end{equation}
where $k_1=\frac{\gamma_a \delta_a}{k_a}$ and $k_2=\frac{\gamma_b \delta_b}{k_b}.$ Here, $p_a^*$ is a positive real root of (\ref{2}) which we discuss in different cases explicitly. \\
\textbf{Case 1:} When $n_{aa}=n_a=n_b=1,$ then from (\ref{1}), we get a cubic polynomial in $p_a$
	\begin{equation}\label{3}
f(p_a)=Ap_a^3+Bp_a^2+Cp_a+D,
\end{equation}
where $A=k_{1}\,k_{2}\,\theta _{b},\,B=B_{1}\,k_{1}\,\theta _{aa}+k_{1}\,k_{2}\,\theta _{a}\,\theta _{b}-k_{2}\,\theta _{b}\,\left(A_1+\mathrm{m_a}-k_{1}\,\theta _{aa}\right)\\
C=k_{1}\,\theta_{aa}\,\theta_a\,\left(B_{1}+\mathrm{m_b}\right)-k_{2}\,\left(A_1+\mathrm{m_a}-k_{1}\,\theta _{aa}\right)\theta _{a}\theta_b-A_1\,B_{1}\,\theta _{aa}-A_1\,k_{2}\,\theta _{b}\,\theta _{aa}$ and $D=-A_1\,\theta_{aa}\,\theta_a\,\left(B_{1}+\mathrm{m_b}+k_2\,\theta_b\right).$
The nature of the curve $f$ can be observe from (\ref{3}) as 
\begin{enumerate}
	\item[(i)] $f(0)<0$ because $D<0,$
	\item[(ii)] $f\to \infty\, (-\infty)$ as $t\to \infty\,(-\infty)$ because $A>0,$
		
	\item[(iii)] The number of positive real roots of equation (\ref{3}) can be seen using Descarte's rule of sign, mentioned in below Table (\ref{T1}).
	\begin{table}[H]
		\centering
		\begin{tabular}{|cccc|c|}
			\hline
			\multicolumn{4}{|c|}{\textbf{Coefficients}}                                                                       & \multirow{2}{*}{\textbf{\begin{tabular}[c]{@{}c@{}}Number of possible positive\\ real roots of $f(p_a)=0$\end{tabular}}} \\ \cline{1-4}
			\multicolumn{1}{|c|}{\textbf{A}} & \multicolumn{1}{c|}{\textbf{B}} & \multicolumn{1}{c|}{\textbf{C}} & \textbf{D} &                                                                                                                         \\ \hline
			\multicolumn{1}{|c|}{+}          & \multicolumn{1}{c|}{-}          & \multicolumn{1}{c|}{-}          & -          & 1                                                                                                                       \\ \hline
			\multicolumn{1}{|c|}{+}          & \multicolumn{1}{c|}{+}          & \multicolumn{1}{c|}{-}          & -          & 1                                                                                                                       \\ \hline
			\multicolumn{1}{|c|}{+}          & \multicolumn{1}{c|}{+}          & \multicolumn{1}{c|}{+}          & -          & 1                                                                                                                       \\ \hline
			\multicolumn{1}{|c|}{+}          & \multicolumn{1}{c|}{-}          & \multicolumn{1}{c|}{+}          & -          & 3,1                                                                                                                     \\ \hline
		\end{tabular}
	\caption{Number of possible positive real roots of $f(p_a)=0.$}\label{T1}
	\end{table}
\end{enumerate}
Therefore, $f(p_a)=0$ will have either unique or three positive real roots. \\
\textbf{Case 2:} When $n_a=n_b=1,\,n_{aa}=2,$ then from (\ref{1}), we get a quartic polynomial in $p_a$
\begin{equation}\label{4}
	f(p_a)=Ap_a^4+Bp_a^3+Cp_a^2+Dp_a+E,
\end{equation}
where $A=k_{1}\,k_{2}^{n_{b}}\,\theta_{b}^{n_{b}},\,B=-(m_a+A_1-k_1\theta_a^{n_a})k_2^{n_b}\theta_b^{n_b},\,
C=k_1B_1\theta_{aa}^{n_{aa}}-k_2^{n_b}\theta_b^{n_b}(\theta_a^{n_a}(m_a+A_1)-\theta_{aa}^{n_{aa}}k_1),\,D=-k_2^{n_b}\theta_b^{n_b}\theta_{aa}^{n_{aa}}(A_1-k_1\theta_a^{n_a})-\theta_{aa}^{n_{aa}}(A_1B_1-k_1(B_1+m_b)\theta_a^{n_a}),$ and $E=-A_1\theta_{aa}^{n_{aa}}\theta_a^{n_a}(k_2^{n_b}\theta_b^{n_b}+B_1+m_b).$\\
The number of positive real roots of equation (\ref{4}) can be seen using Descarte's rule of sign, mentioned in below Table (\ref{T2}).
\begin{table}[H]
	\centering
	\begin{tabular}{|ccccc|c|}
		\hline
		\multicolumn{5}{|c|}{\textbf{Coefficients}}                                                                                                                              & \multirow{2}{*}{\textbf{\begin{tabular}[c]{@{}c@{}}Number of possible positive\\ real roots of $f(p_a)=0$\end{tabular}}} \\ \cline{1-5}
		\multicolumn{1}{|c|}{\textbf{A}} & \multicolumn{1}{c|}{\textbf{B}} & \multicolumn{1}{c|}{\textbf{C}} & \multicolumn{1}{c|}{\textbf{D}} & \multicolumn{1}{l|}{\textbf{E}} &                                                                                                                          \\ \hline
		\multicolumn{1}{|c|}{+}          & \multicolumn{1}{c|}{-}          & \multicolumn{1}{c|}{+}          & \multicolumn{1}{c|}{- or +}     & -                               & 3, 1                                                                                                                     \\ \hline
		\multicolumn{1}{|c|}{+}          & \multicolumn{1}{c|}{+ or -}     & \multicolumn{1}{c|}{-}          & \multicolumn{1}{c|}{+}          & -                               & 3, 1                                                                                                                     \\ \hline
		\multicolumn{1}{|c|}{+}          & \multicolumn{1}{c|}{+ or -}     & \multicolumn{1}{c|}{-}          & \multicolumn{1}{c|}{-}          & -                               & 1                                                                                                                        \\ \hline
		\multicolumn{1}{|c|}{+}          & \multicolumn{1}{c|}{+}          & \multicolumn{1}{c|}{+}          & \multicolumn{1}{c|}{- or +}     & -                               & 1                                                                                                                        \\ \hline
	\end{tabular}\caption{Number of possible positive real roots of $f(p_a)=0.$}\label{T2}
\end{table}
\textbf{Case 3:} When $n_a=2,\,n_b=1,\,n_{aa}=2$ then from (\ref{1}), we get a quintic polynomial in $p_a$
\begin{equation}\label{5}
	f(p_a)=Ap_a^5+Bp_a^4+Cp_a^3+Dp_a^2+Ep_a+F,
\end{equation}
where $A=k_{1}\,{k_{2}}^{n_{b}}\,{\theta _{b}}^{n_{b}},\,B=-{k_{2}}^{n_{b}}\,{\theta _{b}}^{n_{b}}\,\left(A_1+\mathrm{m}_a\right),\,
C=B_{1}\,k_{1}\,{\theta _{aa}}^{n_{aa}}+k_{1}\,{k_{2}}^{n_{b}}\,{\theta _{b}}^{n_{b}}\,\\ \left({\theta _{a}}^{n_{a}}+{\theta _{aa}}^{n_{aa}}\right),\,D=-A_1\,B_{1}\,{\theta _{aa}}^{n_{aa}}-{k_{2}}^{n_{b}}\,{\theta _{b}}^{n_{b}}\,\left({\theta _{a}}^{n_{a}}\,\left(A_1+\mathrm{m}_a\right)+A_1\,{\theta _{aa}}^{n_{aa}}\right)
,\\E=k_{a}\,{\theta _{a}}^{n_{a}}\,{\theta _{aa}}^{n_{aa}}\,\left(B_{1}+\mathrm{m}_b+{k_{2}}^{n_{b}}\,{\theta _{b}}^{n_{b}}\right)$ and $F=-A_1\,{\theta _{a}}^{n_{a}}\,{\theta _{aa}}^{n_{aa}}\,\left(B_{1}+\mathrm{m}_b+{k_{2}}^{n_{b}}\,{\theta _{b}}^{n_{b}}\right).$\\
The nature of curve $f$ can be observe from (\ref{3}) as 
\begin{enumerate}
	\item[(i)] $f(0)<0$ because $F<0,$
	\item[(ii)] $f\to \infty\, (-\infty)$ as $t\to \infty\,(-\infty)$ because $A>0,$
	\item[(iii)] $f(p_a)=0$ will have either unique or three or five positive real roots.
	\end{enumerate}
	 To deal with the fifth-degree polynomial with parametric coefficients is quite difficult, so we use the nullcline method to further analyze the existence of possible positive real roots. From the system (\ref{1}), we have
\begin{equation}\label{6}
	\begin{aligned}
    		&F_1(p_a,p_b):=\frac{m_{a}\ p_a^{2}\theta_{b}}{\theta_{aa}^{2}\theta_{b}+p_a^{2}\theta_{b}+p_b\theta_{aa}^{2}}+A_1-\frac{\gamma_{a}\delta_{a}p_a}{k_{a}}=0 \\
		&F_2(p_a,p_b):=\frac{m_{b}\theta_{a}^{2}}{\theta_{a}^{2}+p_a^{2}}+B_{1}-\frac{\gamma_{b}\delta_{b}p_b}{k_{b}}=0,
	\end{aligned}
\end{equation}
First to observe the nature of curve $F_2(p_a,p_b)=0,$ we reduce it in the form of
\begin{equation}\label{7}
	p_b=g_1(p_a):=\frac{k_b}{\gamma_b\delta_b}\Big\{\frac{m_{b}\theta_{a}^{2}}{(\theta_{a}^{2}+p_a^{2})}+B_{1}\Big\}\tag{$\mathscr{A}$}
\end{equation}
The following conclusion can be made from (\ref{7})
\begin{enumerate}
	\item[$\mathscr{A}$(i)] $g_1(0)=\frac{k_b}{\gamma_b\delta_b}(m_b+B_{1})>\frac{k_b}{\gamma_b\delta_b}B_1>0$ 
	\item[$\mathscr{A}$(ii)] $g_1(p_a)\to \frac{k_b}{\gamma_b\delta_b}B_1$ as $p_a\to \pm\infty$ 
	\item[$\mathscr{A}$(iii)] $g'(p_a)_1<0$, i.e, $g_1(p_a) $ is a decreasing function of $p_a.$
\end{enumerate}
The graph of $F_2(p_a,p_b)=0$ (or $p_b=g_1(p_a)$) is shown in Fig. \ref{F1} (a).\\
Now, to observe the nature of curve $F_1(p_a,p_b)=0,$ we reduce it in the form of
\begin{equation}\label{8}
 p_b=g_2(p_a):=\frac{\left(-\gamma_{a}\delta_{a}p_a^{3}+\left(m_{a}k_{a}+A_1k_{a}\right)p_a^{2}-\gamma_{a}\delta_{a}\theta_{aa}^{2}p_a+A_1k_{a}\theta_{aa}^{2}\right)\theta_{b}}{\theta_{aa}^{2}\left(p_a-\frac{A_1k_{a}}{\delta_{a}\gamma_{a}}\right)\delta_{a}\gamma_{a}}.\tag{$\mathscr{B}$}
\end{equation}
The following conclusion can be made from (\ref{8})
\begin{enumerate}
	\item[$\mathscr{B}$(i)]  $g_2(p_a)$ is a rational function and it has an asymptote at $(p_a-\frac{A_1'k_{a}}{\delta_{a}\gamma_{a}})$. 
	\item[$\mathscr{B}$(ii)]The positive $p_a-$intercept can be find with $g_2(p_a)=0,$ which implies that the positive real roots of $-\gamma_{a}\delta_{a}p_a^{3}+\left(m_{a}k_{a}+A_1k_{a}\right)p_a^{2}-\gamma_{a}\delta_{a}\theta_{aa}^{2}p_a+A_1k_{a}\theta_{aa}^{2}=0$ will provide the positive $p_a-$intercept of $g_1(p_a).$ 
	\item[$\mathscr{B}$(iii)]The $p_b$-intercept of $g_2(p_a)$ is negative.
	\item[$\mathscr{B}$(iv)]We note that for $p_a\in (0,\frac{A_1k_{a}}{\delta_{a}\gamma_{a}})$, we have $g_2(p_a)<0$ so if positive real roots of $g_2(p_a)=0$ exits then it always lie on the right-hand side of asymptote.\\
	The possibility of unique positive $p_a$-intercept of $g_2(p_a)$ is shown in Fig. \ref{F1} (b) and Fig. \ref{F1} (c). With unique positive $p_a$-intercept, there can be unique and three positive steady states, as shown in Fig. \ref{F2} (a) and Fig. \ref{F2} (b), respectively. In Fig. \ref{F1} (d), we have shown the three positive $p_a$-intercept of $g_2(p_a)$, and for this case, we obtain three and unique positive steady states, as shown in Fig. \ref{F2} (c) and Fig. \ref{F2} (d), respectively. Therefore, from the above discussion, we discard the case of five positive steady states.
\end{enumerate}
\vspace{-0.8em}
\begin{figure}[H]
 \centering
	\subfigure[]{\includegraphics[scale=0.37]{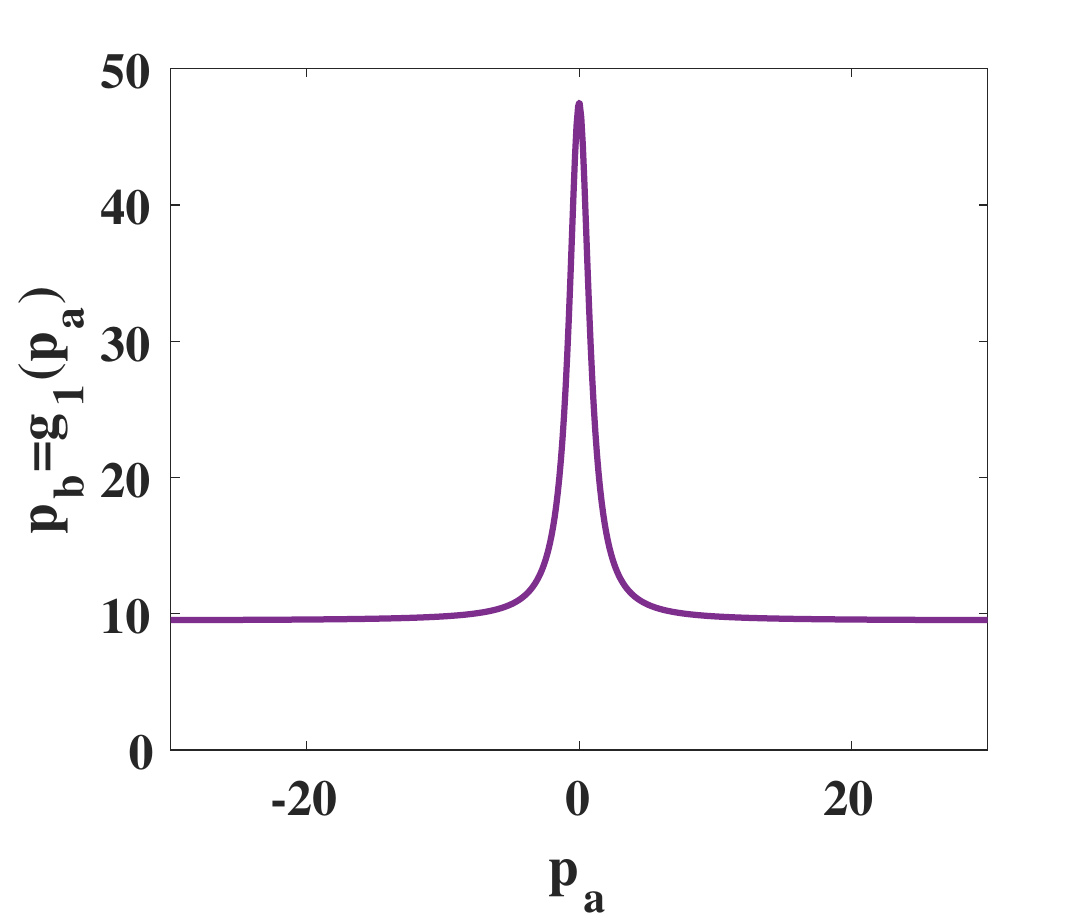}}
	\hfil
	\subfigure[]{\includegraphics[scale=0.37]{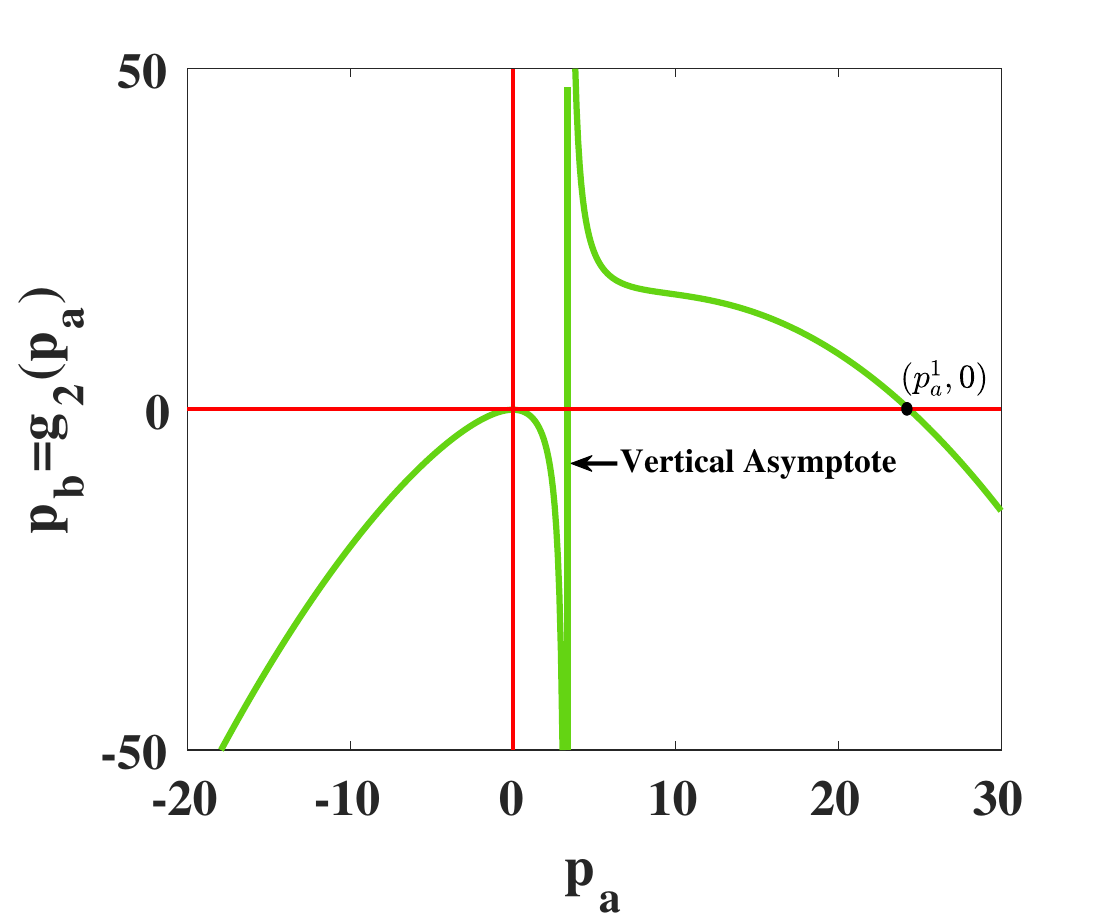}}
	\hfil
\subfigure[]{\includegraphics[scale=0.37]{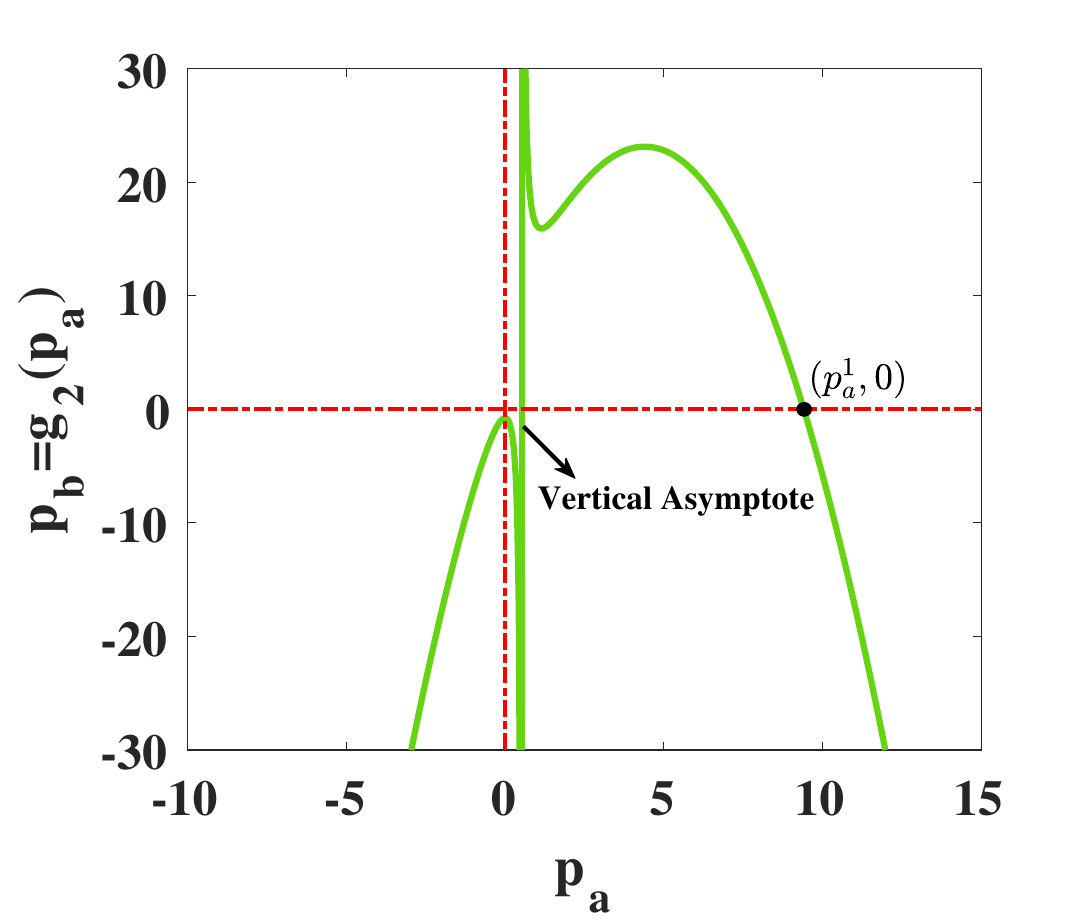}}
	\hfil
			\subfigure[]{\includegraphics[scale=0.37]{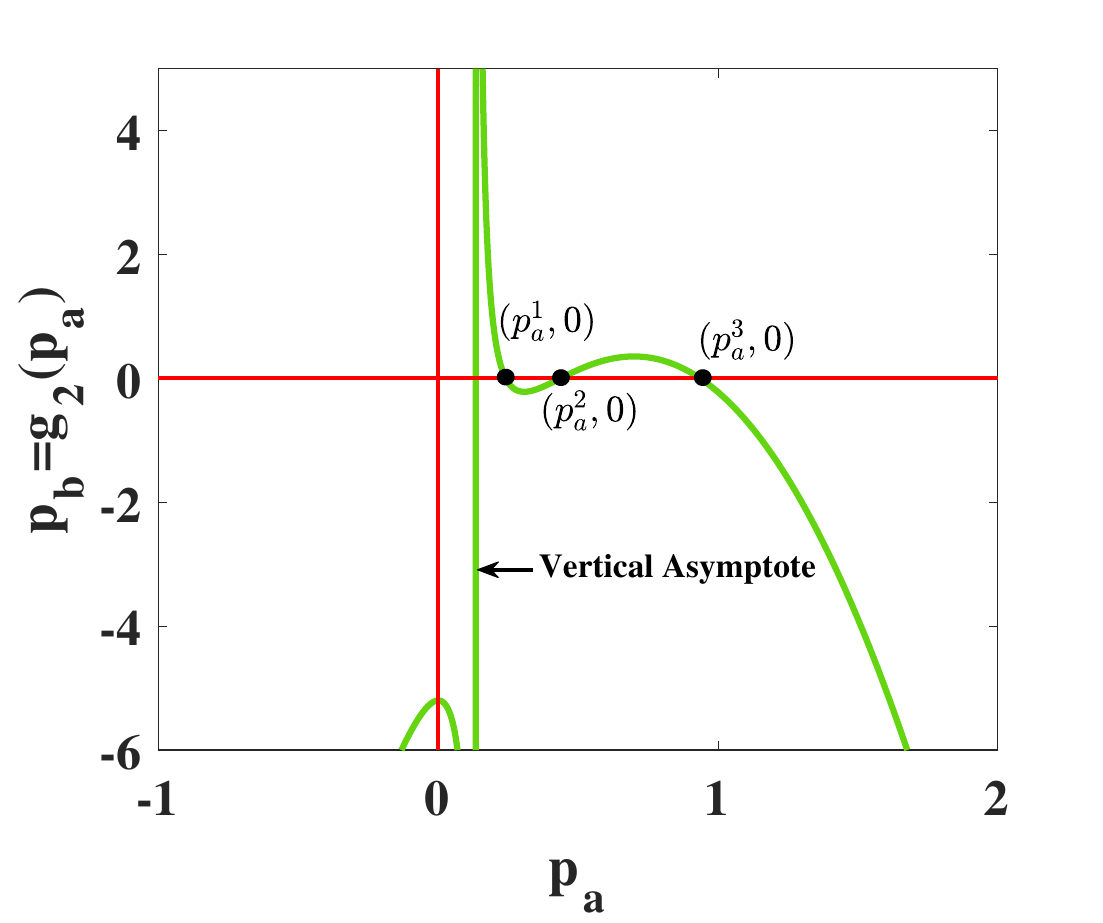}}
		\caption{For $ n_a=n_{aa}=1=2,\,n_b=1;$ (a) Plot of $g_1(p_a)$ as a decreasing function of $p_a$ (b) The graph depicting unique positive $p_a$-intercept of $g_2(p_a)$ (c) The graph depicting another possibility of unique positive $p_a$-intercept of $g_2(p_a)$ (d) The graph depicting three positive $p_a$-intercept of $g_2(p_a).$ }\label{F1}	
		\end{figure}	
\begin{figure}[H]
        \subfigure[]{\includegraphics[scale=0.38]{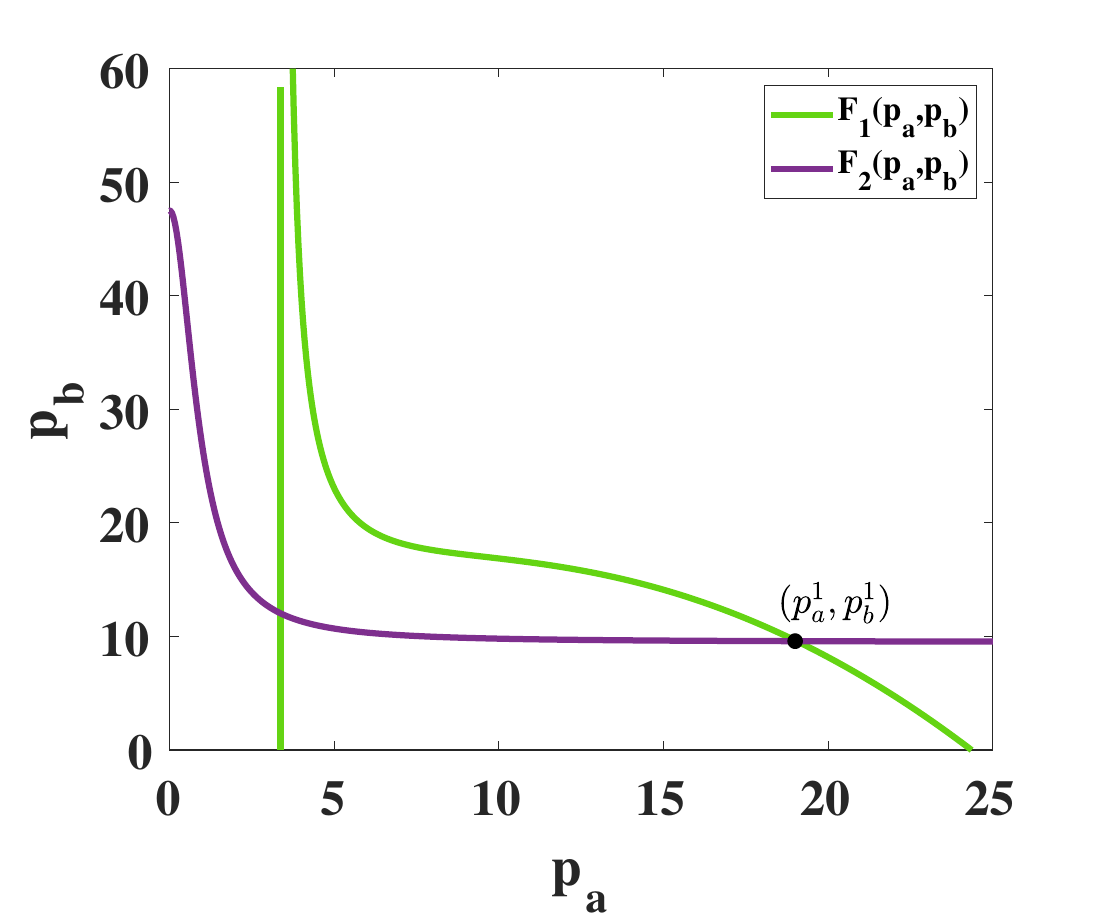}}
		\hfill
		\subfigure[]{\includegraphics[scale=0.38]{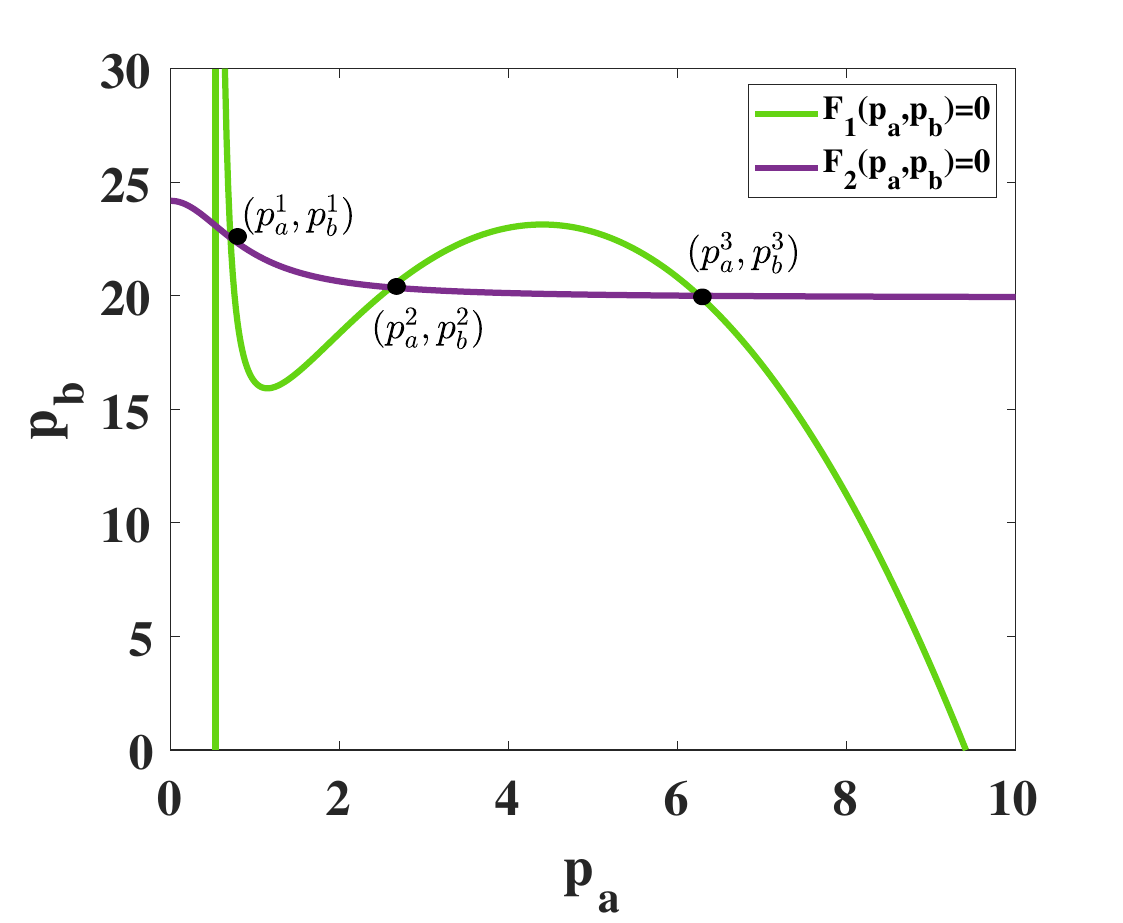}}
		\hfill
		\subfigure[]{\includegraphics[scale=0.38]{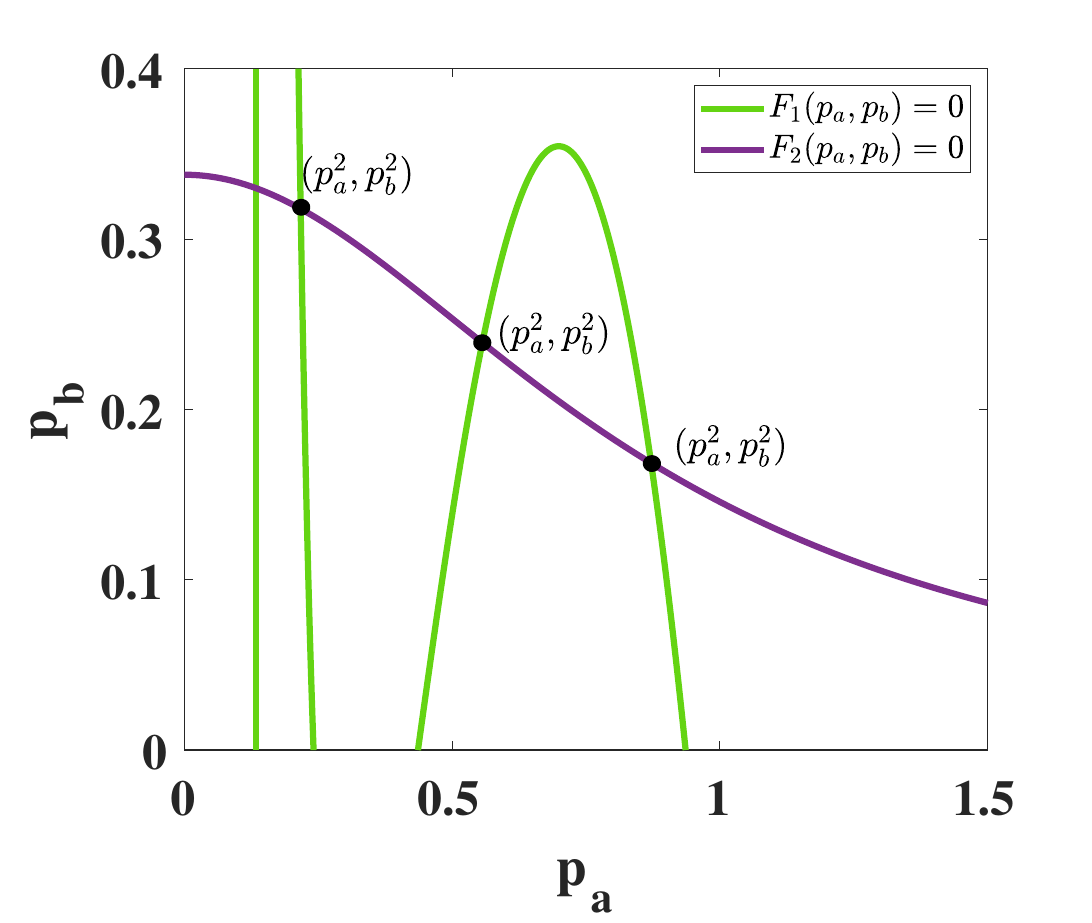}}
		\hfill
		\subfigure[]{\includegraphics[scale=0.38]{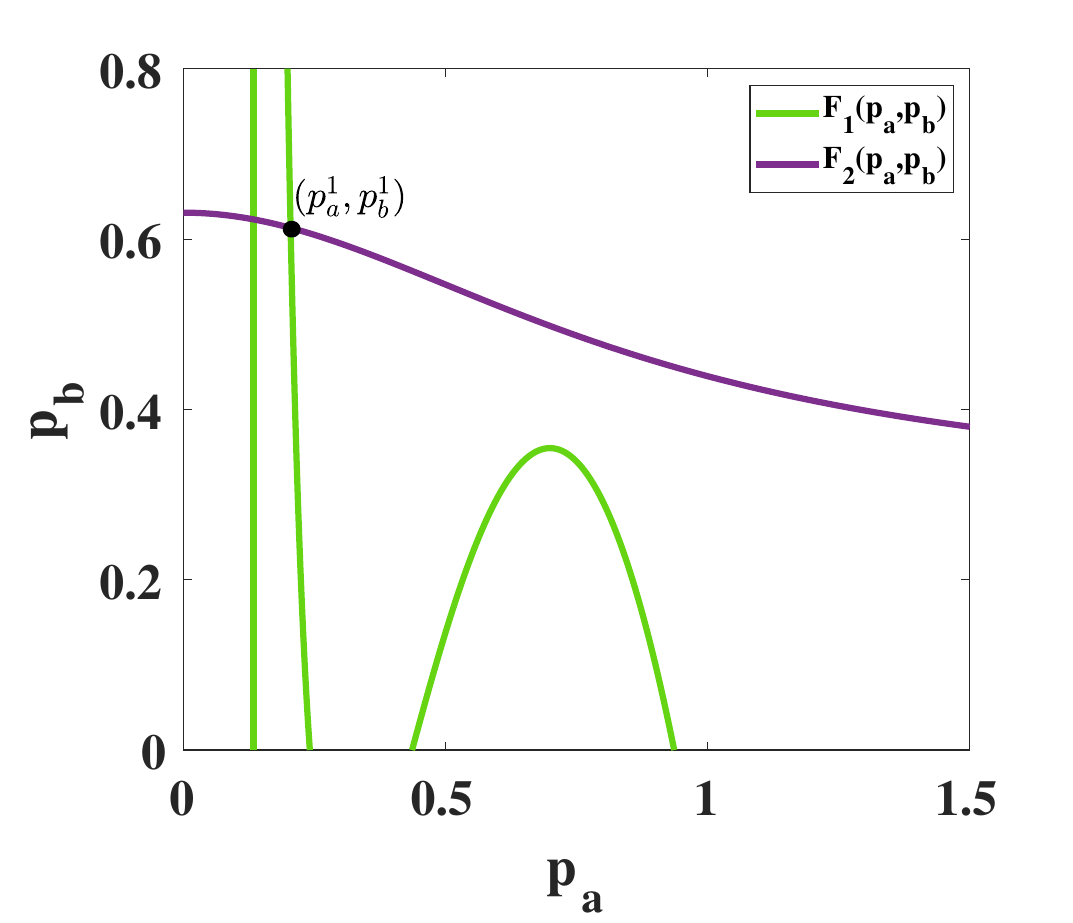}}
		\hfill
	\caption{For $ n_a=n_{aa}=1=2,\,n_b=1;$ (a) Existence of unique positive steady state corresponding to Fig. \ref{F1} (b). (b) Existence of three positive steady states corresponding to Fig. \ref{F1} (c) Existence of three positive steady states corresponding to Fig. \ref{F1} (d)  Existence of unique positive steady state corresponding to Fig. \ref{F1} (d).}
	\label{F2}
\end{figure}

 \subsubsection{Local stability of steady state}
 
\begin{theorem}\label{existence}
	The steady state $S_1^{*}$ is locally asymptotically stable if $\epsilon_1>0,\,\epsilon_3>0,\,\epsilon_4>0,$ and $\epsilon_1\epsilon_2\epsilon_3-\epsilon_1^2\epsilon_4 - \epsilon_3^2 >0,$ where $\epsilon_i,\,i=1,..,4$ are provided in the proof.
\end{theorem}
\begin{proof}At $S_1^{*},$ the Jacobian matrix of the model system (\ref{1}) is given by
	\begin{center}
	$J_{S_1^*}=
	\left[\begin{array}{cccc} -\gamma_a & 0 & X&Y\\ 0 & -\gamma_b & Z&0\\ k_a &0  & -\delta_a&0\\0&-k_b&0&-\delta_b \end{array}\right].$
\end{center}
where $X=\frac{\mathrm{m}_a\,n_{aa}\,{\mathrm{p_a}}^{n_{aa}-1}\,{\theta _{b}}^{n_{b}}\,{\theta _{aa}}^{n_{aa}}\,\left({\mathrm{p_b}}^{n_{b}}+{\theta _{b}}^{n_{b}}\right)}{{\left({\mathrm{p_a}}^{n_{aa}}\,{\theta _{b}}^{n_{b}}+{\mathrm{p_b}}^{n_{b}}\,{\theta _{aa}}^{n_{aa}}+{\theta _{b}}^{n_{b}}\,{\theta _{aa}}^{n_{aa}}\right)}^2}
,\,Y=-\frac{\mathrm{m}_a\,n_{b}\,{\mathrm{p_a}}^{n_{aa}}\,{\mathrm{p_b}}^{n_{b}-1}\,{\theta _{b}}^{n_{b}}\,{\theta _{aa}}^{n_{aa}}}{{\left({\mathrm{p_a}}^{n_{aa}}\,{\theta _{b}}^{n_{b}}+{\mathrm{p_b}}^{n_{b}}\,{\theta _{aa}}^{n_{aa}}+{\theta _{b}}^{n_{b}}\,{\theta _{aa}}^{n_{aa}}\right)}^2},$ and $Z=-\frac{\mathrm{m}_b\,n_{a}\,{\mathrm{p_a}}^{n_{a}-1}\,{\theta _{a}}^{n_{a}}}{{\left({\mathrm{p_a}}^{n_{a}}+{\theta _{a}}^{n_{a}}\right)}^2}.$
	The characteristic equation of $J_{S_1^*}$ is given as
	\begin{equation}\label{9}	\lambda^4+\epsilon_1\lambda^3+\epsilon_2\lambda^2+\epsilon_3\lambda+\epsilon_4=0.
	\end{equation}
	Here, $\epsilon_1=\gamma_a+\gamma_b+\delta_a+\delta_b,\epsilon_2=\gamma_a\gamma_b+(\gamma_a+\gamma_b)(\delta_a+\delta_b)+\delta_a\delta_b-k_aX,\,\epsilon_3=\delta_a\delta_b(\gamma_a+\gamma_b)+(\delta_a+\delta_b)(\gamma_a\gamma_b)-k_a(\delta_b+\gamma_b)X,$ and  $\epsilon_4=\gamma_a\gamma_b\delta_a\delta_b-k_ak_bYZ-k_a\gamma_b\delta_bX.$
	By the Routh-Hurwitz criteria, all the roots of (\ref{9}) have negative real parts provided $\epsilon_1>0,\,\epsilon_3>0,\,\epsilon_4>0,$ and $\epsilon_1\epsilon_2\epsilon_3-\epsilon_1^2\epsilon_4 - \epsilon_3^2 >0.$
\end{proof}
\subsubsection{Saddle Node Bifurcation}\label{sn_OR}
We derive the transversality condition for saddle-node bifurcation considering $\theta_b$ as a bifurcation parameter for  system (\ref{1}). Let $\theta_b^{SN}$ is the threshold value of $\theta_b$ for the occurrence of saddle-node bifurcation.The system (\ref{1}) has two steady states ($S_2^*,\,S_3^*$) for $\theta_b>\theta_b^{SN} (\theta_b<\theta_b^{SN}),$ collide ($S_2^*=S_3^*$) at $\theta_b=\theta_b^{SN}$  and then disappear when $\theta_b<\theta_b^{SN} (\theta_b>\theta_b^{SN})$. The point of collision is represented as $S_{SN}^*=(r_a^{{SN}^*},r_b^{{SN}^*},p_a^{{SN}^*},p_b^{{SN}^*}).$ Define $f=(f_1,f_2,f_3,f_4)^T,$ where $f_1,\,f_2 ,\,f_3,$ and $f_4$ are defined as
	\begin{equation*}
	\begin{aligned}
		&f_1  =  m_a\frac{(\frac{p_a}{\theta_{aa}})^{n_{aa}}}{1+(\frac{p_a}{\theta_{aa}})^{n_{aa}}+(\frac{p_b}{\theta_b})^{n_{b}}}-\gamma_a r_a+A_1, \\
		&f_2 =m_b\frac{1}{1+(\frac{p_a}{\theta_a})^{n_{a}}}-\gamma_b r_b+B_1, \\
		&f_3 =k_ar_a-\delta_ap_a,\\
		&f_4 =k_br_b-\delta_bp_b,
	\end{aligned}
\end{equation*}
Let the Jacobian matrix $J=Df(S_{SN}^*,\theta_b^{SN})$ has a simple eigenvalue $\lambda=0$ with
eigenvector $v=(v_1,v_2,v_3,v_4)^T$, and the transpose of the Jacobian matrix $J^T$ has
an eigenvector $w=(w_1,w_2,w_3,w_4)^T$ to the eigenvalue $\lambda=0.$
Then the model system (\ref{1})
experiences a saddle-node bifurcation at $(S_{SN}^*,\theta_b^{SN})$, if the following transversality conditions \cite{perko_2013} are satisfied:
\begin{equation*}
	w^Tf_{\theta_b}(S_{SN}^*,\theta_b^{SN})=w_1\frac{m_a n_b \theta_{aa}^{n_{aa}}\theta_b^{n_b-1}{p_b^{n_b}}^{{SN}^*}{p_a^{n_{aa}}}^{{SN}^*}}{(\mathcal{X})^2}\neq0. \tag{$\mathcal{A}_1$}\end{equation*}
and \begin{align*}
	w^T[D^2f(S_{SN}^*,\theta_b^{SN})(v,v)]=&\frac{-w_1m_a \theta_{aa}^{n_{aa}} \theta_b^{n_b}}{\mathcal{X}^3}\bigg\{2v_3^2 n_{aa}\theta_b^{n_b}{p_a^{n_{aa}-1}}^{{SN}^*}(\theta_b^{n_b}+{p_b^{n_b}}^{{SN}^*})+v_4^2n_b{p_a^{n_{aa}}}^{{SN}^*}\\&\Big(\theta_b^{n_b}(\theta_{aa}^{n_{aa}}+{p_{a}^{n_{aa}}}^{{SN}^*}){p_b^{n_b-2}}^{{SN}^*}(n_b-1)-\theta_{aa}^{n_{aa}}{p_b^{2n_b-2}}^{{SN}^*}(n_b+1)\Big)\\&-2v_3v_4n_b{p_b^{n_b-1}}^{{SN}^*}\Big(\theta_b^{n_b}{p_a^{n_{aa}}}^{{SN}^*}-\big(\theta_b^{n_b}+{p_b^{n_b}}^{{SN}^*}\big)\theta_{aa}^{n_{aa}}\Big)\bigg\}
\\&	-\frac{w_2v_3^2m_b n_a \theta_a^{n_a}}{(\theta_a^{n_a}+{p_a^{n_a}}^{{SN}^*})^3} \tag{$\mathcal{A}_2$}\bigg\{{p_a^{n_a-2}}^{{SN}^*}\theta_a^{n_a}(n_a-1)-(n_a+1){p_a^{2n_a-2}}^{{SN}^*}\bigg\}	\\&\neq 0,
\end{align*} 
 where $\mathcal{X}=\theta_{aa}^{n_{aa}}\theta_b^{n_b}+{p_a^{n_{aa}}}^{{SN}^*}\theta_b^{n_b}+\theta_{aa}^{n_{aa}}{p_b^{n_b}}^{{SN}^*}.$
\subsubsection{Numerical simulation and discussion corresponding to model (\ref{1})}
The purpose of this section is to provide validation for the theoretical results that were covered in Subsection \ref{2.1}. We have talked about the system's behavior in terms of monostability and bi-stability with fictitious parameters. 
\begin{example} $n_{aa}=n_a=n_b=1$ (case 1) and other parameters are chosen as 
$\delta_{a} = 0.9,\, \delta_{b} = 0.001,\,  \gamma_{a} = 0.99,\,  \gamma_{b} = 0.046,\,  k_{a} = 4.74,\,  k_{b} = 0.017,\,  m_b = 4.6,\,  m_a = 18.6,\,  \theta_{a} = 1.8,\,  A_1 = 1.1,\,  B_1 = 0.01,\, \theta_{aa} = 2.6,\,  
\theta_b = (0.1,6.14).$ In this range of $\theta_b,$ there is unique ($S_1^*$) steady state in $\theta_b^1$, three ($S_1^*,\,S_2^*,\,S_3^*$) in $\theta_b^2$ and unique ($S_3^*$) in $\theta_b^3,$ such that $\theta_b=\theta_b^1\cup\theta_b^2\cup\theta_b^2$ where $\theta_b^1=(0.1,{\theta_b^{SN}}^1),\,\theta_b^2=[{\theta_b^{SN}}^1,{\theta_b^{SN}}^2]$ and $\theta_b^2=({\theta_b^{SN}}^2,6.14).$ The steady states $S_1^*,\,S_3^*$ are stable and ($S_2^*$) is unstable. The corresponding bifurcation diagram is shown in Fig. \ref{case1traj_new} (a).\\	
To ensure the stability of $S_1^*,$ we chosen $\theta_b=1.364\in\theta_b^1$ and the coefficients of characteristic equation at $S_1^*=(1.2963,21.0814,6.8273,358.39)$ are  calculated  as $\epsilon_1=1.937>0,\,\epsilon_2=0.854>0,\,\epsilon_3=0.0360>0,\,\epsilon_4= 0.00003>0,$ and $\epsilon_1\epsilon_2\epsilon_3-\epsilon_1^2\epsilon_4 - \epsilon_3^2=0.0582>0.$ Therefore, from Theorem \ref{existence}, we conclude that $S_1^*$ is locally asymptotically stable. We draw the phase portrait for the considered parameters, shown in Fig. \ref{case1traj_new} (b). We can see that all trajectories starting from any initial points are converging to $S_1^*.$\\
Further, we chosen $\theta_b=3.051\in\theta_b^2$ and the coefficients of characteristic equation at $S_1^*=(1.9213,15.3197,10.1187,260.4357)$ are  calculated  as $\epsilon_1=1.9370>0,\,\epsilon_2=0.6204>0,\,\epsilon_3=0.0251>0,\,\epsilon_4=0.0000108>0,$ and $\epsilon_1\epsilon_2\epsilon_3-\epsilon_1^2\epsilon_4 - \epsilon_3^2=0.0295>0.$ The coefficients of characteristic equation at $S_2^*=(4.5514,7.2020,23.9709,122.4344)$ are  calculated  as $\epsilon_1=1.937>0,\,\epsilon_2=0.4297>0,\,\epsilon_3=0.0161>0,\,\epsilon_4=-0.0000066<0,$ and $\epsilon_1\epsilon_2\epsilon_3-\epsilon_1^2\epsilon_4 - \epsilon_3^2=0.0132>0.$ The coefficients of characteristic equation at $S_3^*=(11.9929,2.9882,63.1625,50.7998)$ are  calculated  as $\epsilon_1=1.937>0,\,\epsilon_2=0.6397>0,\,\epsilon_3= 0.0260>0,\,\epsilon_4=0.000012>0,$ and $\epsilon_1\epsilon_2\epsilon_3-\epsilon_1^2\epsilon_4 - \epsilon_3^2=0.0315>0.$  Therefore, from Theorem \ref{existence}, we conclude that $S_1^*,\,S_3^*$ is locally asymptotically stable and $S_2^*$ is unstable.  We draw the phase portrait for the considered parameters, shown in Fig. \ref{case1traj_new} (c). We can see that all trajectories are converging to $S_1^*,\,S_3^*,$ and nearby trajectories of $S_2^*$ are moving away from $S_2^*.$ This shows the bistable behaviour of the system.\\
We chosen $\theta_b=5.795\in\theta_b^3,$
the coefficients of characteristic equation at $S_3^*=(16.3872,$\\$2.2604,86.3061, 38.4265)$ are  calculated  as $\epsilon_1=1.937>0,\,\epsilon_2=0.8246>0,\,\epsilon_3= 0.0347>0,\,\epsilon_4=0.00002835>0,$ and $\epsilon_1\epsilon_2\epsilon_3-\epsilon_1^2\epsilon_4 - \epsilon_3^2=0.0541>0.$ Therefore, from Theorem \ref{existence}, we conclude that $S_3^*$ is locally asymptotically stable. We draw the phase portrait for the chosen parameters, shown in Fig. \ref{case1traj_new} (d). We can see that all trajectories starting from any initial points are converging to $S_3^*.$\\
For this case, again we choose the parameters as $\delta_a = 0.7,\,\delta_b = 0.1,\, \gamma_a = 0.67,\, \gamma_b = 0.265,\, k_a = 0.9,\, k_b = 0.06,\, m_b = 0.9,\, m_a = 9,\, \theta_a = 0.86,\, A_1 = 0.8,\, B_1 = 0.9,\, \theta_{aa} = 0.86,\,\theta_b = (0.00000001,0.4).$ For this set of parameters, we obtain a unique stable steady state as discussed in Case 1. The corresponding bifurcation plot is shown in Fig. \ref{case1traj_new} (e). To ensure the stability of $S_1^*,$  we choose $\theta_b=0.1$ from the range of $\theta_b.$  The coefficients of characteristic equation at $S_1^*$ are  calculated as $\epsilon_1=1.735>0,\,\epsilon_2=0.738>0,\,\epsilon_3=0.11349>0,\,\epsilon_4=0.0049>0,$ and $\epsilon_1\epsilon_2\epsilon_3-\epsilon_1^2\epsilon_4 - \epsilon_3^2=0.1178>0.$ Therefore, from Theorem \ref{existence}, we conclude that $S_1^*$ is locally asymptotically stable. We draw the phase portrait for the considered $\theta_b$, shown in Fig. \ref{case1traj_new} (f). We can see that all trajectories starting from any initial points are converging to $S_1^*=(3.61984,3.92592, 4.65408,  2.35555).$\\
In this example, we also discuss the case of saddle-node bifurcation. We have seen that the system (\ref{1}) has two steady states $S_2^*,\,S_3^*$ in $\theta_b^2$ such that as the value of $\theta_b$ decreases, the two steady states collide at $\theta_b^{{{SN}^1}}=2.714726411506651$ and denoted as $S_{{SN}^1}^*=(8.3509,4.1491,43.9812,70.5353).$ The Jacobian matrix $J=Df(S_{{SN}^1}^*,\theta_b^{{{SN}^1}})$ has a simple eigenvalue $\lambda=0$ and  $v=(-0.1066,0.0482,-0.5612,0.8194)^T,$\\$w=(-0.0156,0.3466,-0.0033,0.9379)^T$, are the eigenvectors of $J$ and $J^T,$ respectively. Both the tranversality conditions $w^Tf_{\theta_b}(S_{{SN}^1}^*,\theta_b^{{{SN}^1}})=-0.0244\neq 0$ and \\ $	w^T[D^2f(S_{{SN}^1}^*,\theta_b^{{{SN}^1}})(v,v)]=0.0000124\neq 0$ are satisfied. Hence, the system (1) experiences saddle-node bifurcation at $\theta_b=\theta_b^{{{SN}^1}}.$ Also, the system (\ref{1}) has two steady states $S_1^*,\,S_2^*$ in $\theta_b^2$ such that as the value of $\theta_b$ increases, the two steady states collide at $\theta_b^{{{SN}^2}}=3.345369897729079$ and denoted as $S_{{SN}^2}^*=(2.7005,11.4515,14.2227,194.6749).$ The Jacobian matrix $J=Df(S_{{SN}^2}^*,\theta_b^{{{SN}^2}})$ has a simple eigenvalue $\lambda=0$ and  $v=(-0.0158,0.0585,-0.0835,0.9947)^T,\,w=(-0.1278,0.3437,-0.0267,0.93)^T$, are the eigenvectors of $J$ and $J^T,$ respectively. Both the tranversality conditions $w^Tf_{\theta_b}(S_{{SN}^2}^*,\theta_b^{{{SN}^2}})=-0.0541\neq 0$ and  $	w^T[D^2f(S_{{SN}^2}^*,\theta_b^{{{SN}^2}})(v,v)]=-0.000006817\neq 0$ are satisfied. Hence, the system (\ref{1}) experiences saddle-node bifurcation at $\theta_b=\theta_b^{{{SN}^2}}.$
\begin{figure}[H]
	\subfigure[]{\includegraphics[scale=0.27]{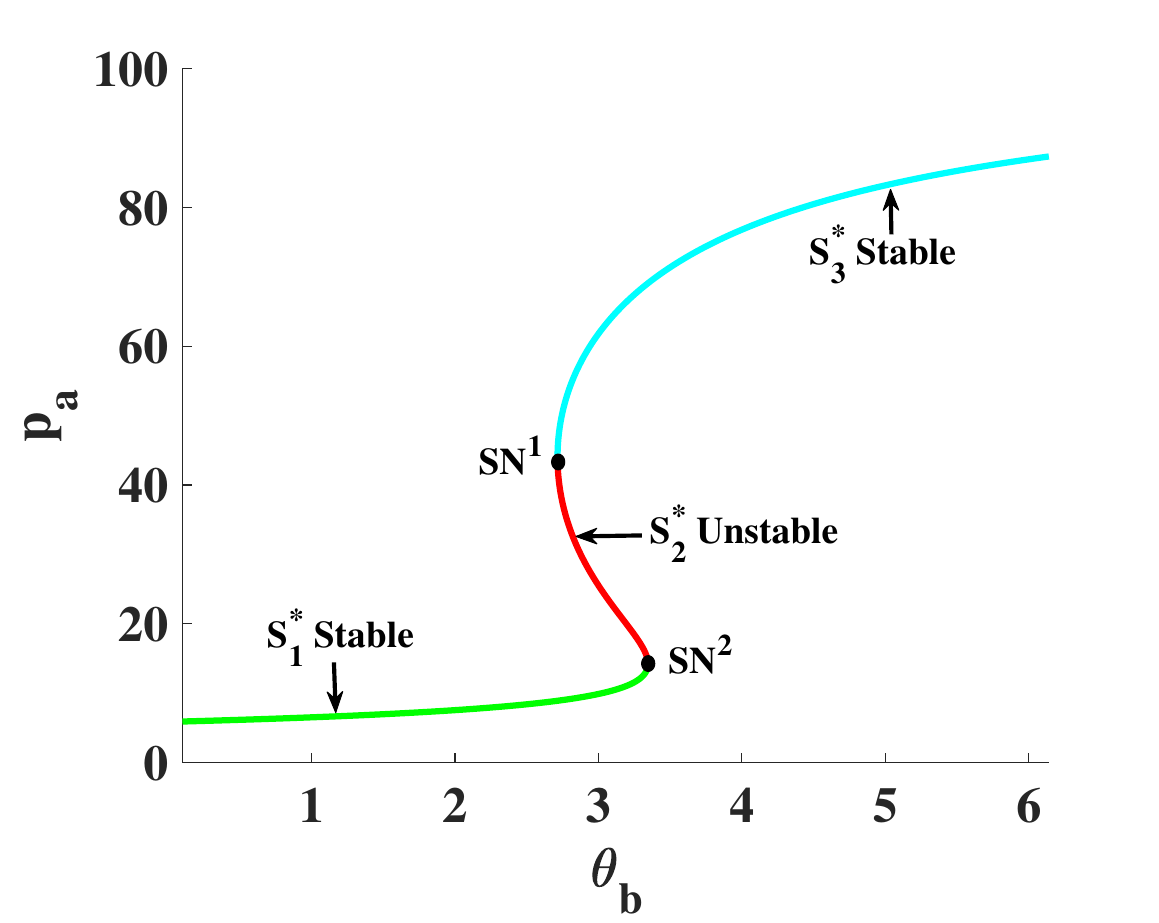}}
\hfill
\subfigure[]{\includegraphics[scale=0.27]{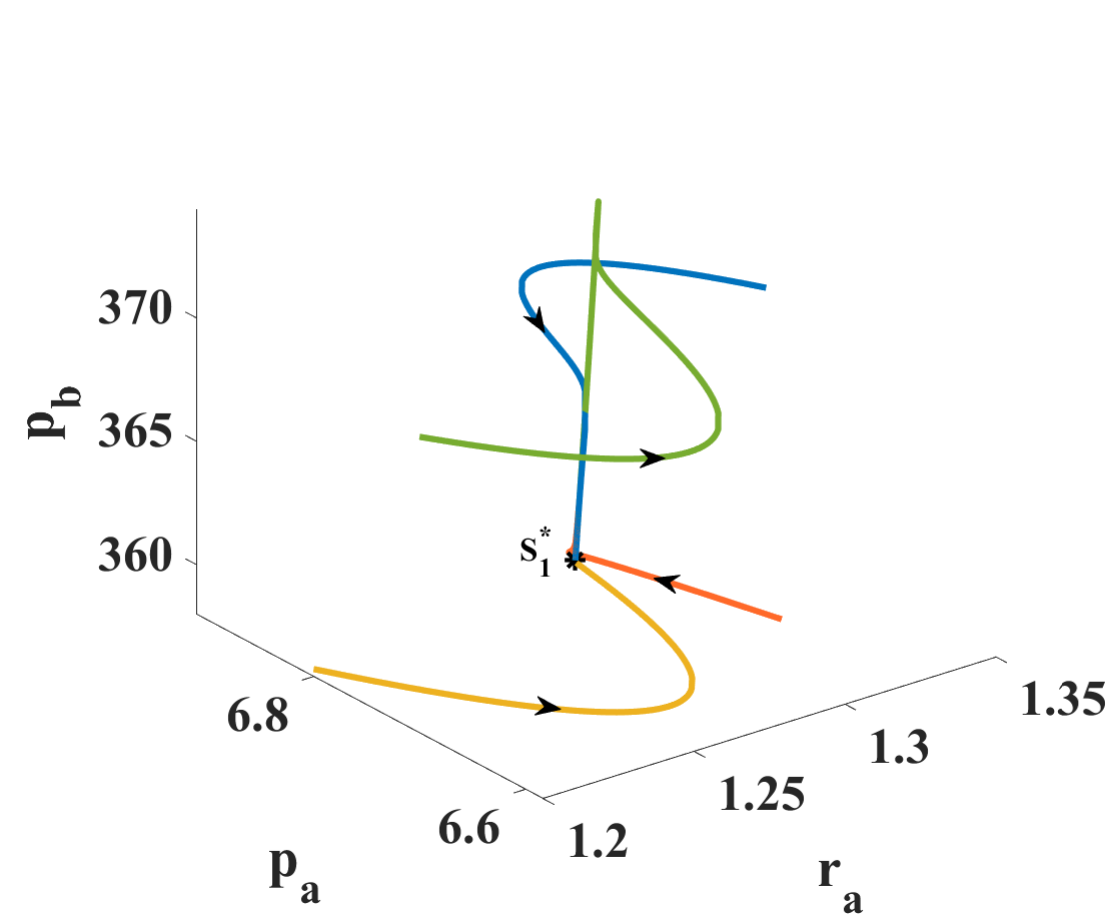}}
	\hfill
	\subfigure[]{\includegraphics[scale=0.27]{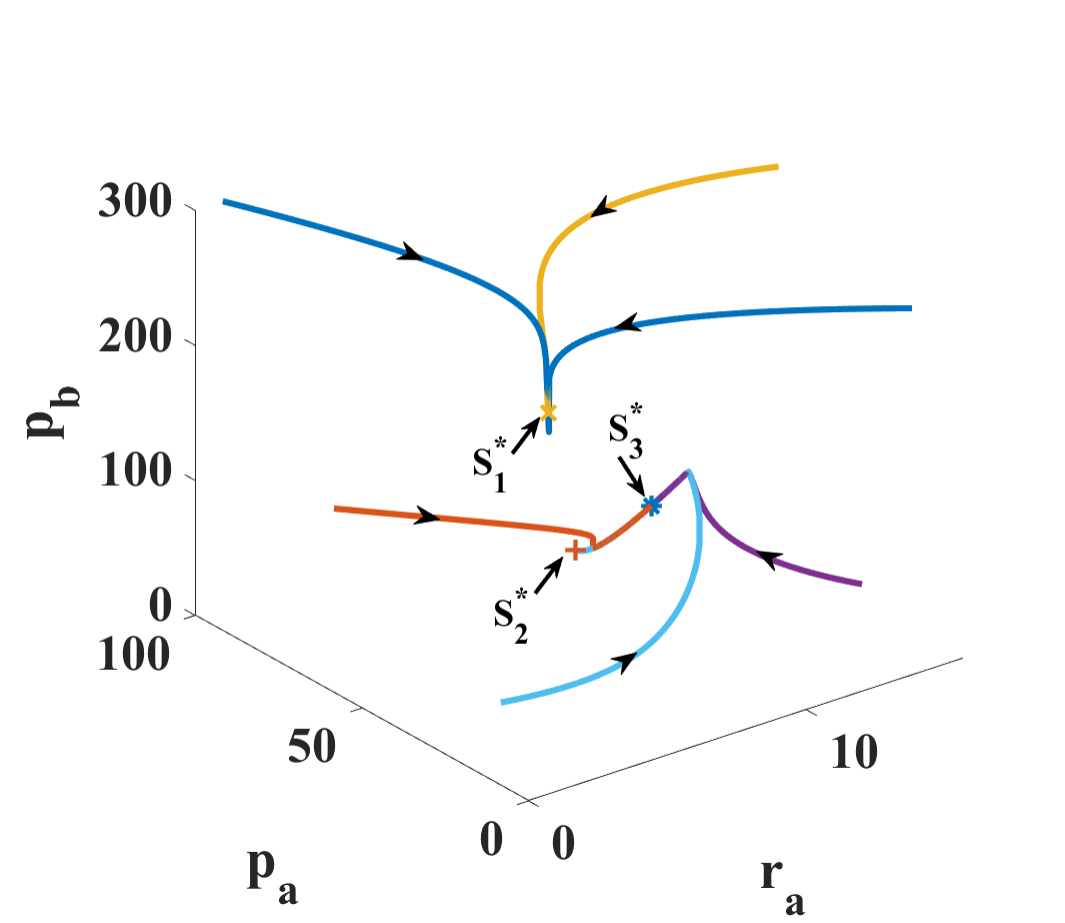}}
	\hfill
	\subfigure[]{\includegraphics[scale=0.27]{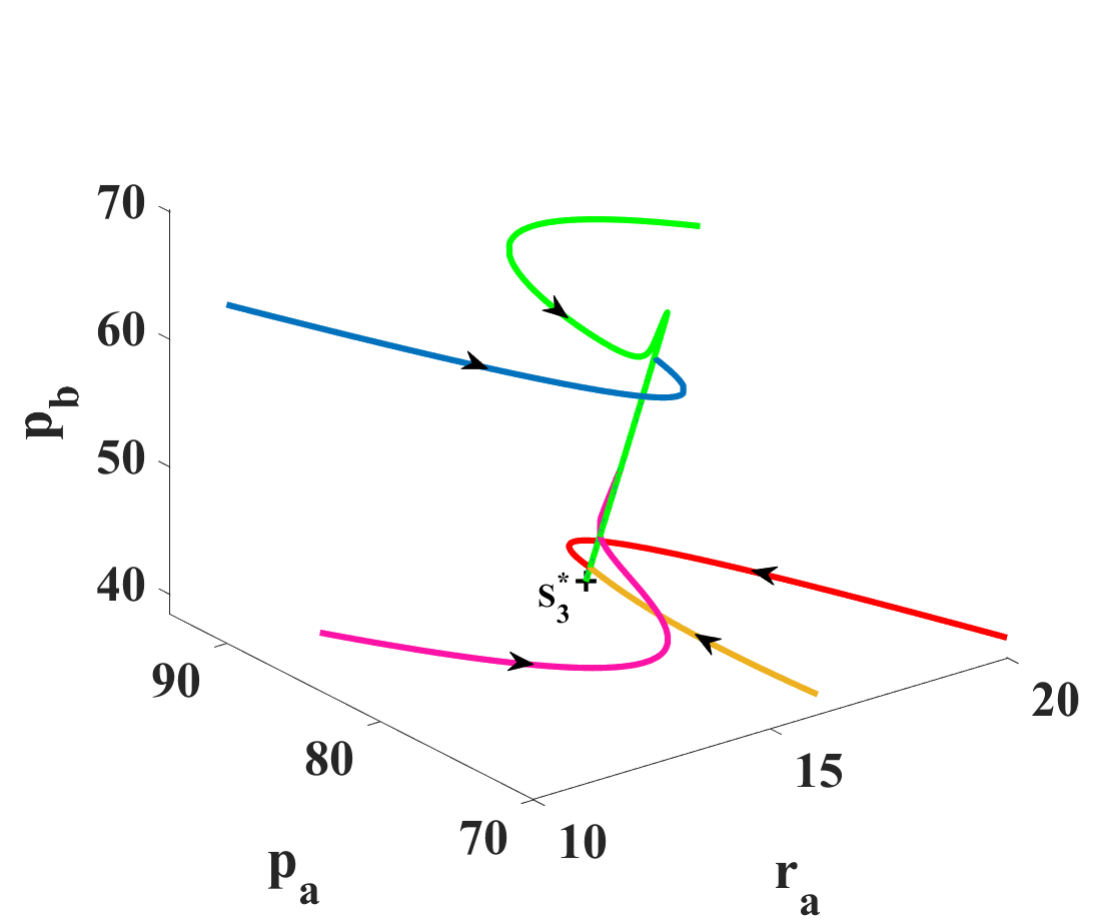}}
	\hfill
 \subfigure[]{\includegraphics[scale=0.27]{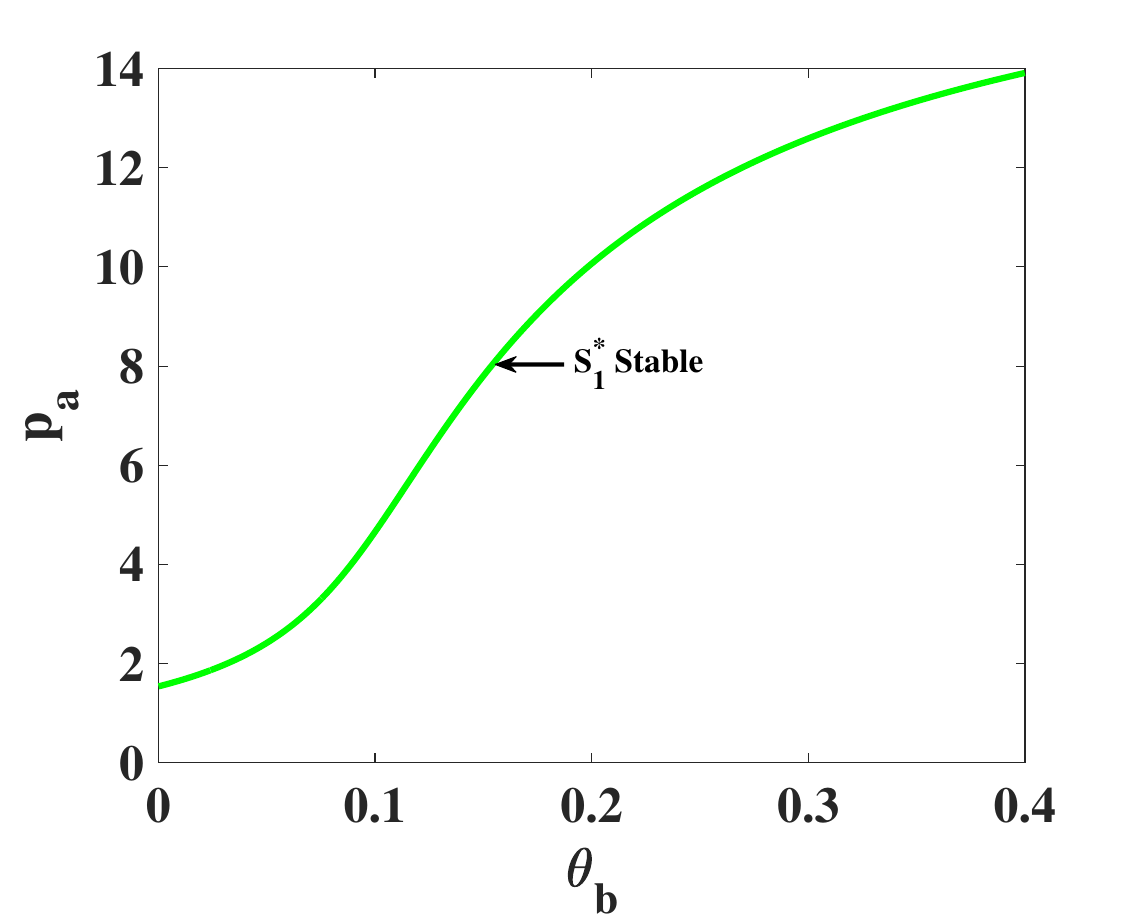}}
		\hfill
		\subfigure[]{\includegraphics[scale=0.27]{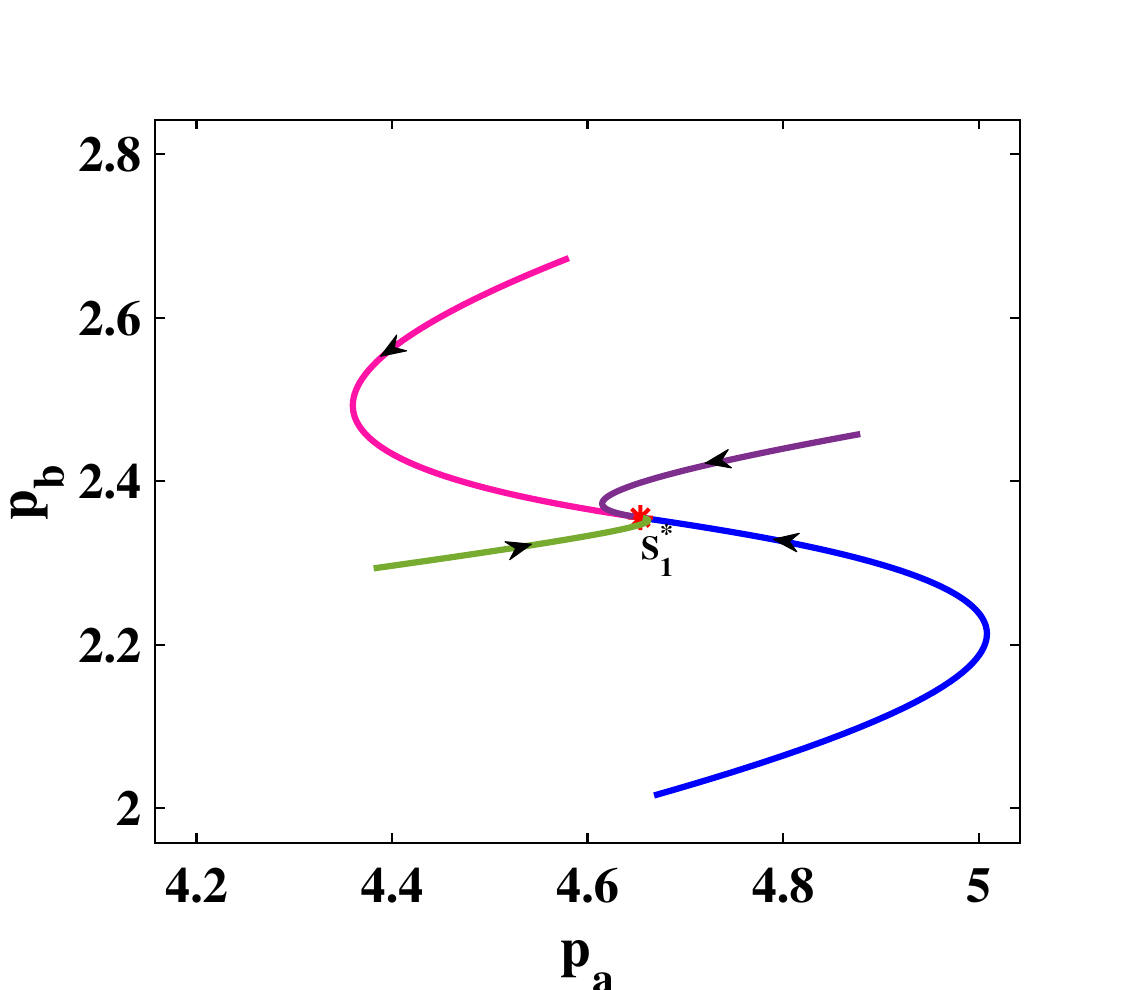}}
		\hfill
	\caption{For $n_a=1,\,n_b=1,\,n_{aa}=1$; (a) Bifurcation plot exhibiting hysteresis effect with three steady states (b) Phase portrait at $\theta_b=1.364$ (c) Phase portrait at $\theta_b=3.051$ depicting bistability (d)  Phase portrait at $\theta_b=5.795$ (e) Bifurcation plot with unique stable steady state (f) Phase portrait at $\theta_b=0.1.$}
	\label{case1traj_new}
\end{figure}
	\end{example}
\begin{example}$n_a=n_b=1,\,n_{aa}=2$ (case 2), and the other parameters are chosen as: 
	$\delta_a = 0.9,\,  \delta_b = 0.04,\,  \gamma_a = 0.604,\,  \gamma_b = 0.3,\,  k_a = 0.525,\,  k_b = 0.9,\,  m_b = 0.01,\,   m_a = 15,\,  \theta_a = 0.4,\,  A_1 =0.65,\,  B_1 =0.785,\,  \theta_{aa} =0.7,\,  
	\theta_b =(0.1,0.99).$ In this range of $\theta_b,$ there is unique ($S_1^*$) steady state in $\theta_b^1$, three ($S_1^*,\,S_2^*,\,S_3^*$) in $\theta_b^2$ and unique ($S_3^*$) in $\theta_b^3,$ such that $\theta_b=\theta_b^1\cup\theta_b^2\cup\theta_b^2$ where $\theta_b^1=(0.1,{\theta_b^{SN}}^1),\,\theta_b^2=[{\theta_b^{SN}}^1,{\theta_b^{SN}}^2]$ and $\theta_b^2=({\theta_b^{SN}}^2,0.99).$ The steady states $S_1^*,\,S_3^*$ are stable and ($S_2^*$) is unstable. The corresponding bifurcation diagram is shown in Fig. \ref{F3} (a).\\
To ensure the stability of $S_1^*,$ we chosen $\theta_b=0.3\in\theta_b^1$ and the coefficients of characteristic equation at $S_1^*=(1.2011,2.6288,0.7006,59.1476)$ are  calculated  as $\epsilon_1=1.8440>0,\,\epsilon_2=0.9545>0,\,\epsilon_3=0.1646>0,\,\epsilon_4= 0.0052>0,$ and $\epsilon_1\epsilon_2\epsilon_3-\epsilon_1^2\epsilon_4 - \epsilon_3^2=0.2450>0.$ Therefore, from Theorem \ref{existence}, we conclude that $S_1^*$ is locally asymptotically stable. We draw the phase portrait for the considered parameters, shown in Fig. \ref{F3} (b). We can see that all trajectories starting from any initial points are converging to $S_1^*.$ \\
Further, we chosen $\theta_b=0.7\in\theta_b^2$ and the coefficients of characteristic equation at $S_1^*=(1.5546,2.6269,0.9069,59.1046)$ are  calculated  as $\epsilon_1=1.8440>0,\,\epsilon_2=0.7388>0,\,\epsilon_3=0.0913>0,\,\epsilon_4=0.0026>0,$ and $\epsilon_1\epsilon_2\epsilon_3-\epsilon_1^2\epsilon_4 - \epsilon_3^2=0.1073
>0.$ The coefficients of characteristic equation at $S_2^*=(4.2128,2.6213,2.4575,58.98)$ are  calculated  as $\epsilon_1=1.8440>0,\,\epsilon_2=0.3597>0,\,\epsilon_3=-0.0376<0,\,\epsilon_4=-0.0020<0,$ and $\epsilon_1\epsilon_2\epsilon_3-\epsilon_1^2\epsilon_4 - \epsilon_3^2=-0.0196<0.$ The coefficients of characteristic equation at $S_3^*=(20.1513,2.6178,11.7549,58.8997)$ are  calculated  as $\epsilon_1=1.844>0,\,\epsilon_2=0.8283>0,\,\epsilon_3= 0.1217>0,\,\epsilon_4=0.0037>0,$ and $\epsilon_1\epsilon_2\epsilon_3-\epsilon_1^2\epsilon_4 - \epsilon_3^2=0.1587>0.$  Therefore, from Theorem \ref{existence}, we conclude that $S_1^*,\,S_3^*$ is locally asymptotically stable and $S_2^*$ is unstable.  We draw the phase portrait for the considered parameters, shown in Fig. \ref{F3} (c). We can see that all trajectories are converging to $S_1^*,\,S_3^*,$ and nearby trajectories of $S_2^*$ are moving away from $S_2^*.$ This shows the bistable behavior of the system.\\
We chosen $\theta_b=0.9\in\theta_b^3,$
the coefficients of characteristic equation at $S_3^*=(21.7247,2.6177,$\\$12.6727, 58.8979)$ are  calculated  as $\epsilon_1=1.8440>0,\,\epsilon_2=0.8928>0,\,\epsilon_3= 0.1437>0,\,\epsilon_4=0.0044>0,$ and $\epsilon_1\epsilon_2\epsilon_3-\epsilon_1^2\epsilon_4 - \epsilon_3^2=0.2008>0.$ Therefore, from Theorem \ref{existence}, we conclude that $S_3^*$ is locally asymptotically stable. We draw the phase portrait for the chosen parameters, shown in Fig. \ref{F3} (d). We can see that all trajectories starting from any initial points are converging to $S_3^*.$ \\
In this example, we also discuss the case of saddle-node bifurcation. We have seen that the system (\ref{1}) has two steady states $S_2^*,\,S_3^*$ in $\theta_b^2$ such that as the value of $\theta_b$ decreases, the two steady states collide at $\theta_b^{{{SN}^1}}=0.466081314925648$ and denoted as $S_{{SN}^1}^*=(12.3079,2.6184,7.1796,58.9146).$ The Jacobian matrix $J=Df(S_{{SN}^1}^*,\theta_b^{{{SN}^1}})$ has a simple eigenvalue $\lambda=0$ and  $v=(0.8638,-1.1694, 0.5039,-0.0026)^T,$\\$w=(0.1932,-0.9066,0.2222,-0.3022)^T$, are the eigenvectors of $J$ and $J^T,$ respectively. Both the tranversality conditions $w^Tf_{\theta_b}(S_{{SN}^1}^*,\theta_b^{{{SN}^1}})=1.528\neq 0$ and \\ $	w^T[D^2f(S_{{SN}^1}^*,\theta_b^{{{SN}^1}})(v,v)]=-0.00089\neq 0$ are satisfied. Hence, the system (1) experiences saddle-node bifurcation at $\theta_b=\theta_b^{{{SN}^1}}.$ Also, the system (\ref{1}) has two steady states $S_1^*,\,S_2^*$ in $\theta_b^2$ such that as the value of $\theta_b$ increases, the two steady states collide at $\theta_b^{{{SN}^2}}=0.845253711971854$ and denoted as $S_{{SN}^2}^*=(2.2633,2.6244,1.3202,59.0494).$ The Jacobian matrix $J=Df(S_{{SN}^2}^*,\theta_b^{{{SN}^2}})$ has a simple eigenvalue $\lambda=0$ and  $v=(0.8627,-0.0023,0.5032,-0.0510)^T,\,w=(0.5647,-0.4828,0.6497,-0.1609)^T$, are the eigenvectors of $J$ and $J^T,$ respectively. Both the tranversality conditions $w^Tf_{\theta_b}(S_{{SN}^2}^*,\theta_b^{{{SN}^2}})=0.4497\neq 0$ and  $	w^T[D^2f(S_{{SN}^2}^*,\theta_b^{{{SN}^2}})(v,v)]=-0.0081\neq 0$ are satisfied. Hence, the system (\ref{1}) experiences saddle-node bifurcation at $\theta_b=\theta_b^{{{SN}^2}}.$\\
The system (\ref{1}) exhibits a hysteresis effect in $\theta_b^2$ where multiple steady states coexist, as
shown in Fig. \ref{case1traj_new} (a) and Fig \ref{F3} (a). The two outer steady states are stable, while the interior steady state (red) is
unstable. \\
For this case, again we choose   $\gamma_a = 0.14,\, m_a = 2,\,    A_1 =0.35,\, 
\theta_b\in(0.1,1.8)$ and other parameters are the same as above. We obtained a unique steady state that is stable throughout the range of considered $\theta_b.$ The corresponding bifurcation diagram is shown in Fig. \ref{F3} (e). The phase trajectory for $\theta_b=0.3,$ is shown in Fig. \ref{F3} (f). All trajectories are converging to $S_1^*=(2.91359, 2.62302,  1.6996,59.0179).$
	\begin{figure}[H]
		\subfigure[]{\includegraphics[scale=0.25]{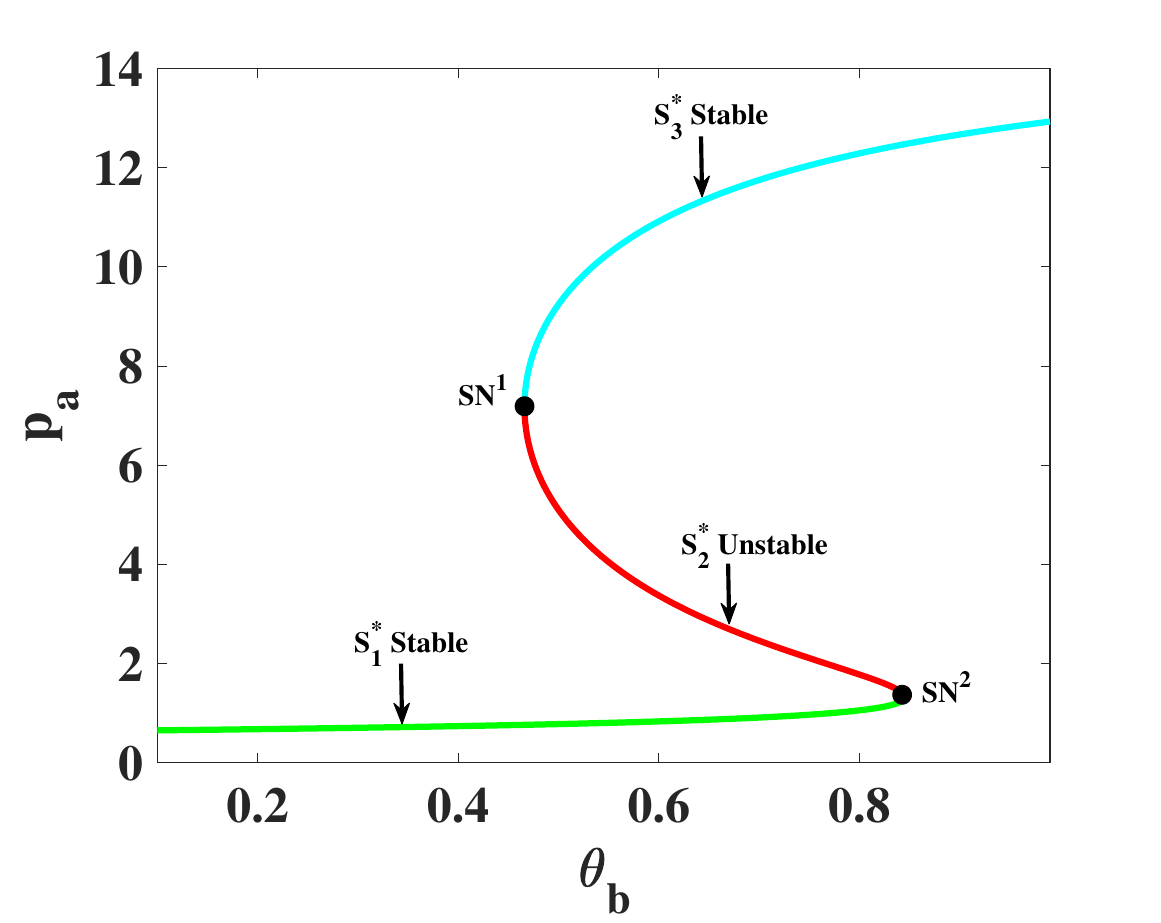}}
		\hfill
		\subfigure[]{\includegraphics[scale=0.25]{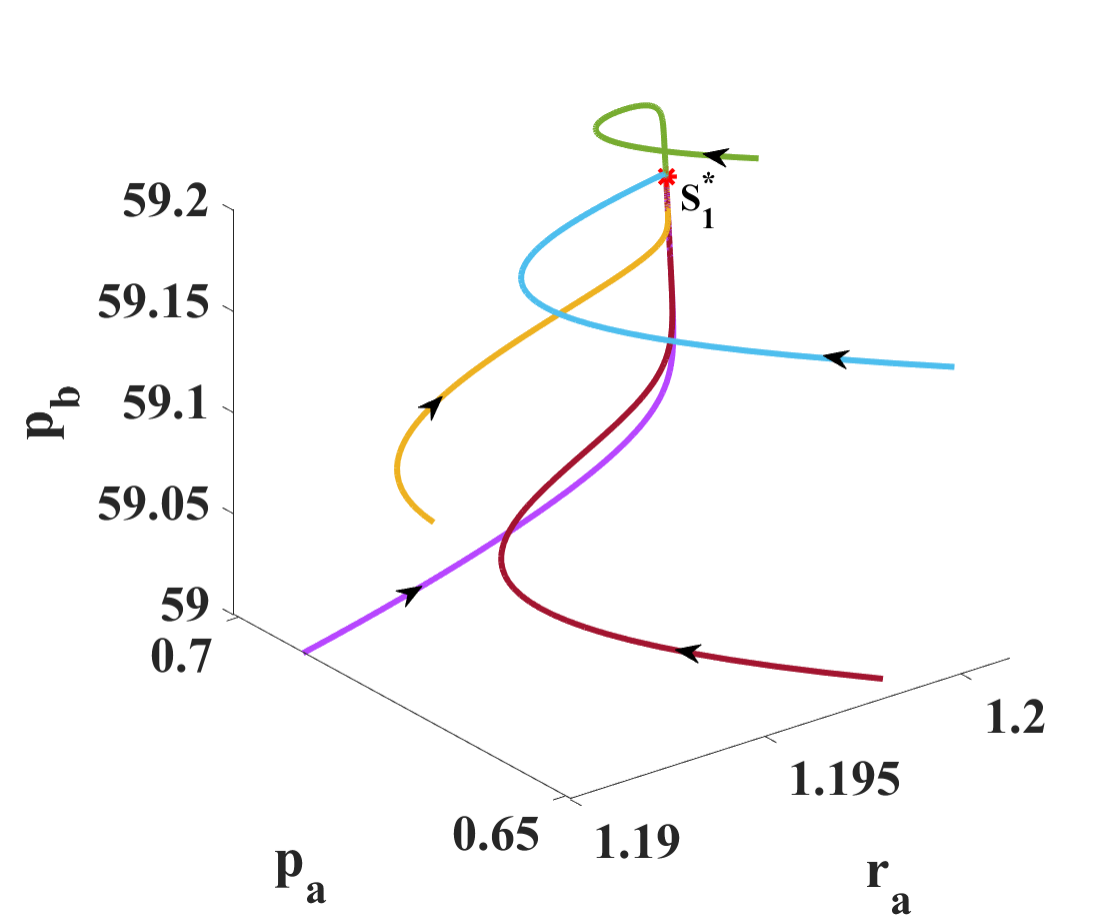}}
		\hfill
		\subfigure[]{\includegraphics[scale=0.25]{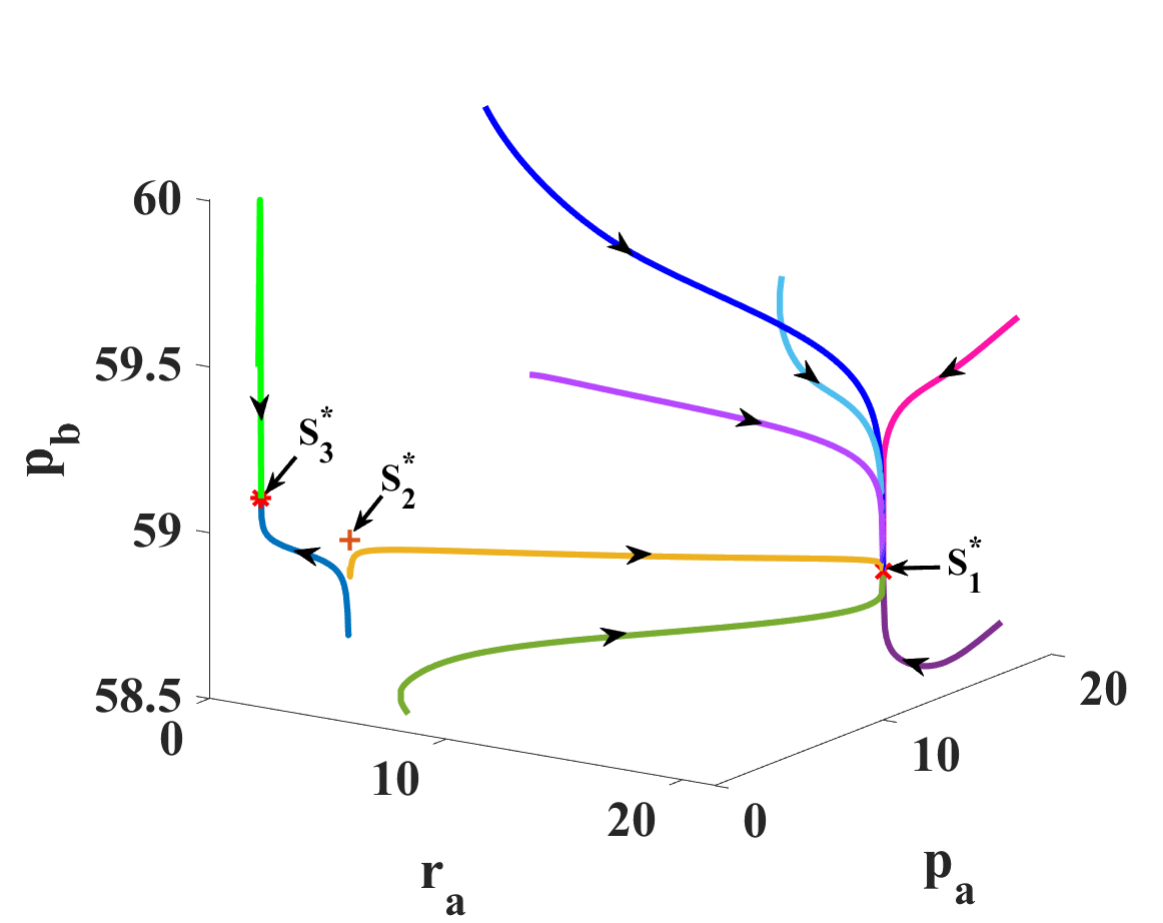}}
		\hfill
		\subfigure[]{\includegraphics[scale=0.25]{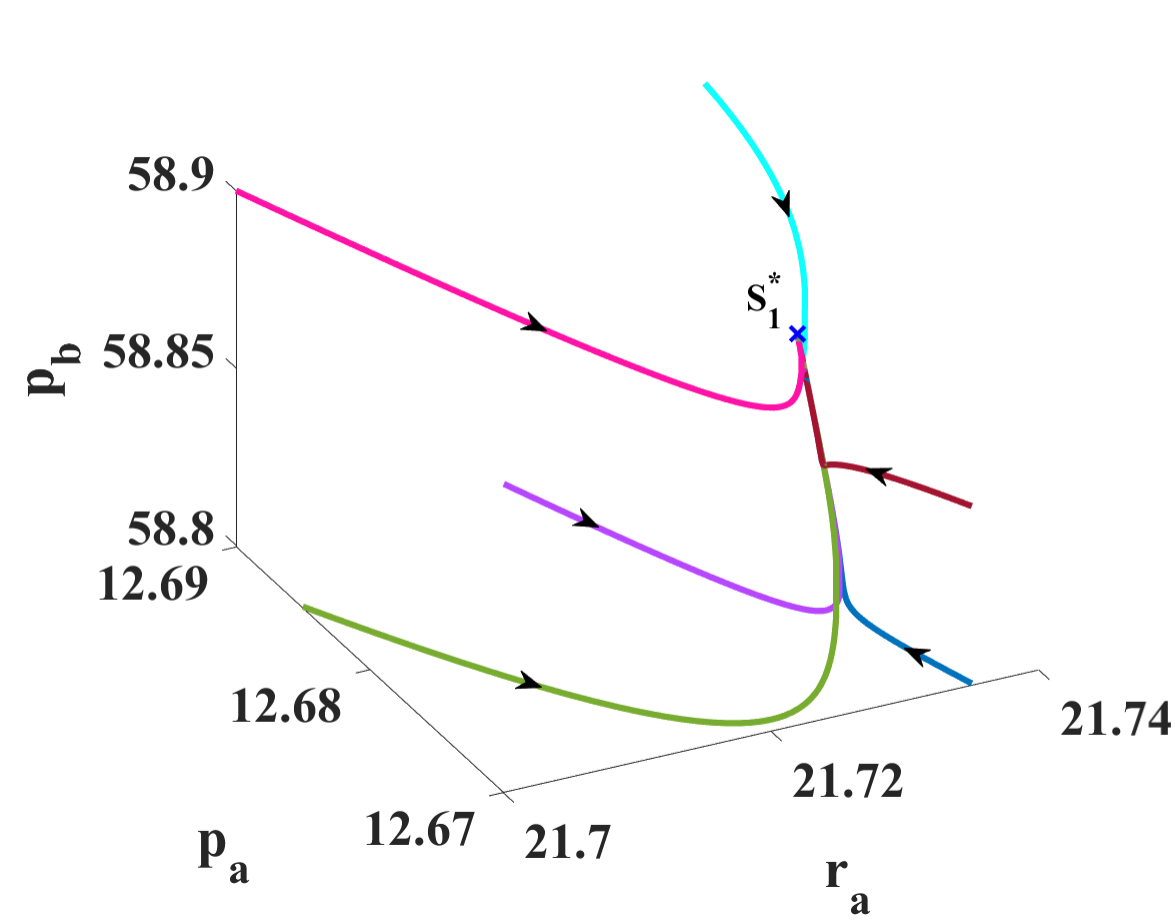}}
		\hfill
  \subfigure[]{\includegraphics[scale=0.25]{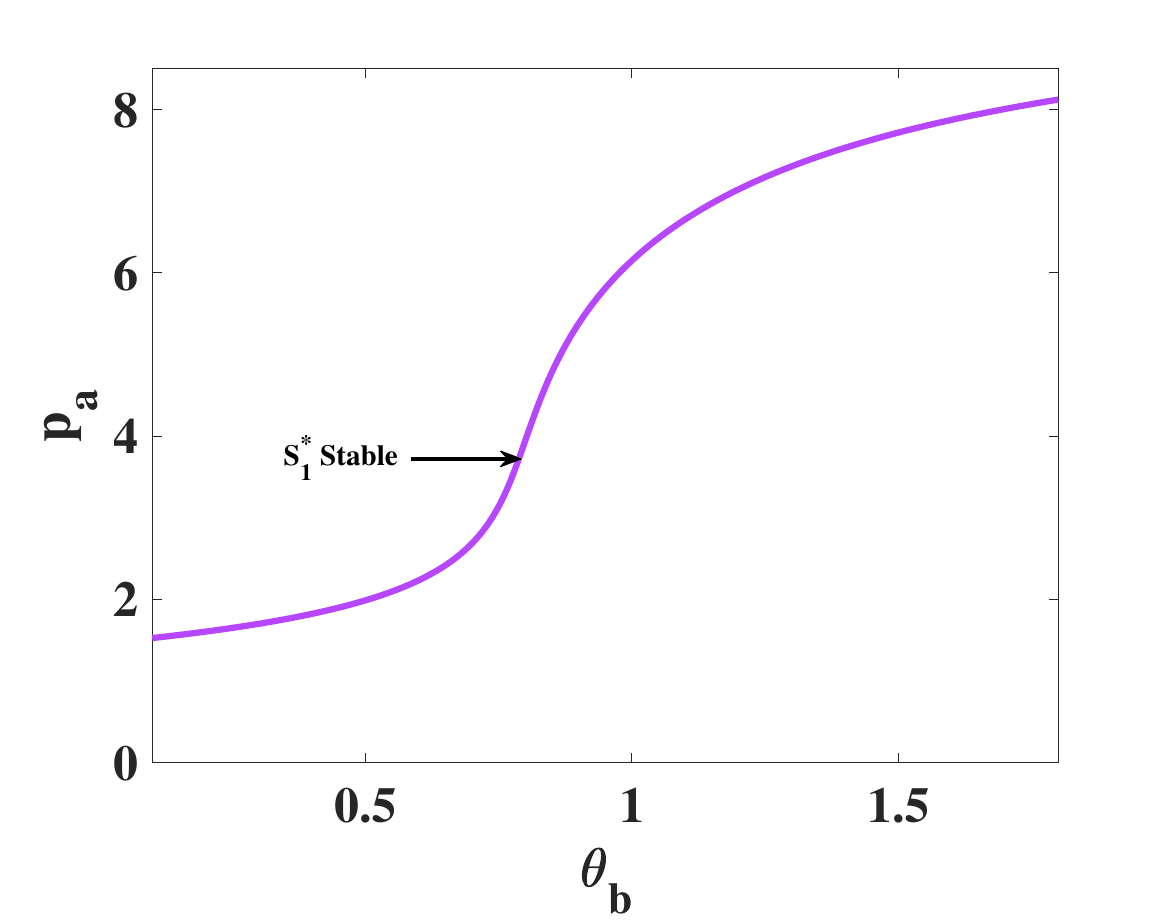}}
	\hfill
	\subfigure[]{\includegraphics[scale=0.25]{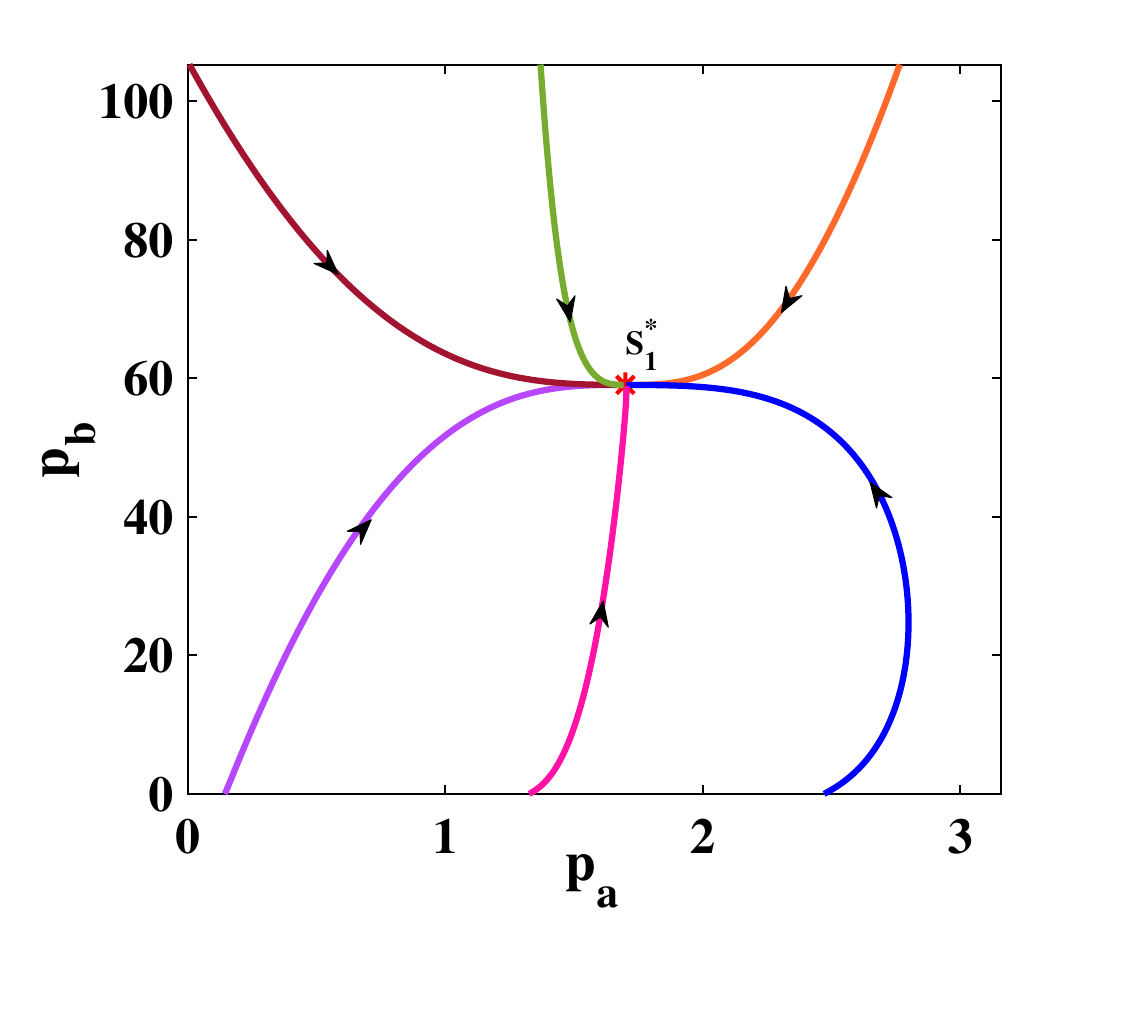}}
	\hfill
		\caption{ For $n_a=1,\,n_b=1,\,n_{aa}=2$; (a) Bifurcation plot exhibiting hysteresis effect with three steady states (b) Phase portrait at $\theta_b=0.3$ (c) Phase portrait at $\theta_b=0.7,$ depicting bistability (d)  Phase portrait at $\theta_b=0.9$  (e) Bifurcation plot with unique stable steady state (f) Phase portrait at $\theta_b=0.3.$}
		\label{F3}
	\end{figure}
\end{example}
\begin{example}
$n_a = 2,\, n_b =1,\, n_{aa} = 2$ (case 3), and the other parameters are chosen as:
$\delta_a = 0.4,\,  \delta_b = 0.001,\, \gamma_a = 0.67,\, \gamma_b = 0.046,\, k_a = 0.7,\,k_b = 0.017,\, m_b = 8.375,\,  m_a =9,\, \theta_a = 0.6,\, A_1 =1,\, B_1 =0.01,\,\theta_{aa} =2.8,\,\theta_b\in(0.1,2.9). $
 In this range of $\theta_b,$ there is unique ($S_3^*$) steady state in $\theta_b^1$, three ($S_1^*,\,S_2^*,\,S_3^*$) in $\theta_b^2$ and unique in $\theta_b^3,$ such that $\theta_b=\theta_b^1\cup\theta_b^2\cup\theta_b^2$ where $\theta_b^1=(0.1,{\theta_b^{SN}}^1),\,\theta_b^2=[{\theta_b^{SN}}^1,{\theta_b^{SN}}^2]$ and $\theta_b^2=({\theta_b^{SN}}^2,0.99).$ The steady states $S_1^*,\,S_3^*$ are stable and ($S_2^*$) is unstable. The corresponding bifurcation diagram is shown in Fig. \ref{F4} (a). All other descriptions are the same as those discussed in Example 2.
 \end{example}
 \begin{example}
 $n_a = 2,\, n_b =2,\, n_{aa} = 1$ (case 4), and the other parameters are chosen as:
 $\delta_a = 0.9,\,  \delta_b = 0.001,\, \gamma_a = 1.23,\, \gamma_b = 0.046,\, k_a = 5.1,\,k_b = 0.017,\, m_b = 10,\,  m_a =14,\, \theta_a = 0.6,\, A_1 =1,\, B_1 =0.01,\,\theta_{aa} =2.7,\,\theta_b\in(0.8,5.8). $
 In this range of $\theta_b,$ there is unique ($S_3^*$) steady state in $\theta_b^1$, three ($S_1^*,\,S_2^*,\,S_3^*$) in $\theta_b^2$ and unique in $\theta_b^3,$ such that $\theta_b=\theta_b^1\cup\theta_b^2\cup\theta_b^2$ where $\theta_b^1=(0.1,{\theta_b^{SN}}^1),\,\theta_b^2=[{\theta_b^{SN}}^1,{\theta_b^{SN}}^2]$ and $\theta_b^2=({\theta_b^{SN}}^2,0.99).$ The steady states $S_1^*,\,S_3^*$ are stable and ($S_2^*$) is unstable. The corresponding bifurcation diagram is shown in Fig. \ref{F4} (b). All other descriptions are the same as discussed in Example 2.\\
 \begin{figure}[H]
 \centering
\subfigure[]{\includegraphics[scale=0.35]{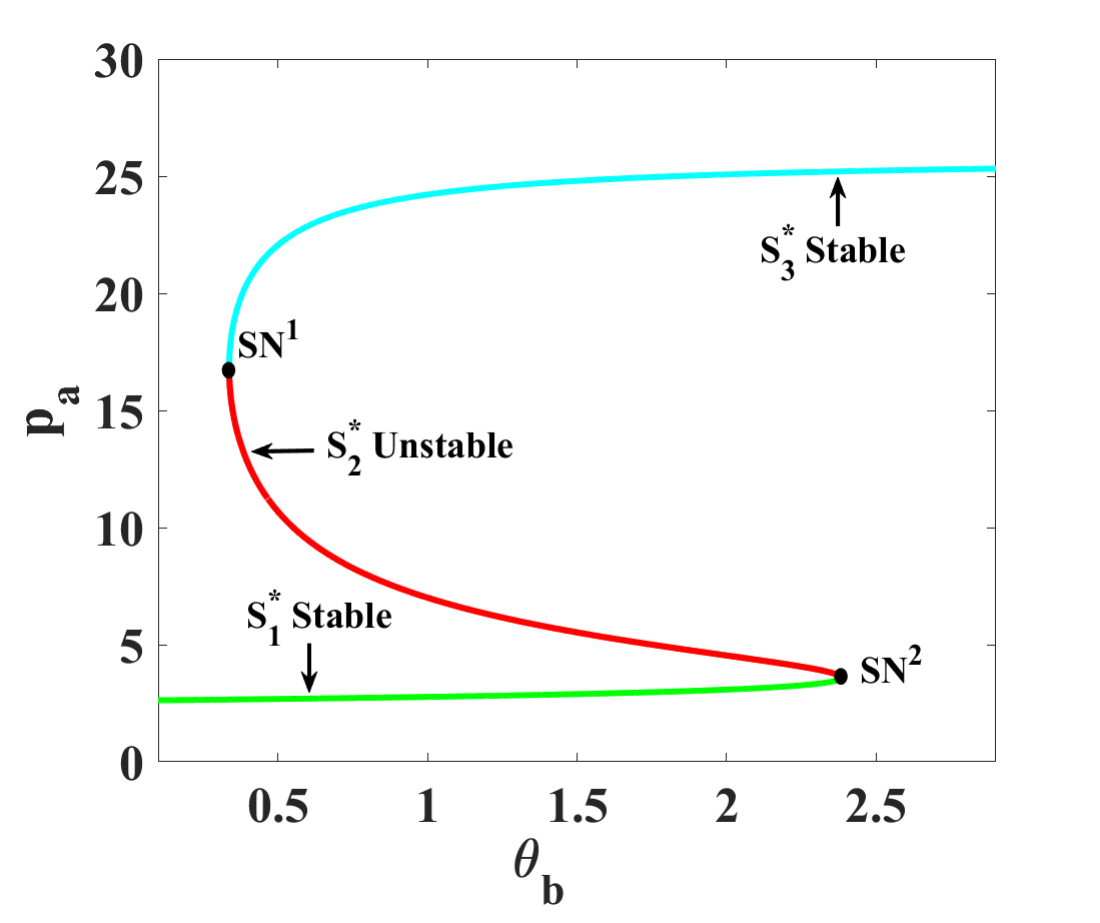}}
 	\hfil
 	\subfigure[]{\includegraphics[scale=0.35]{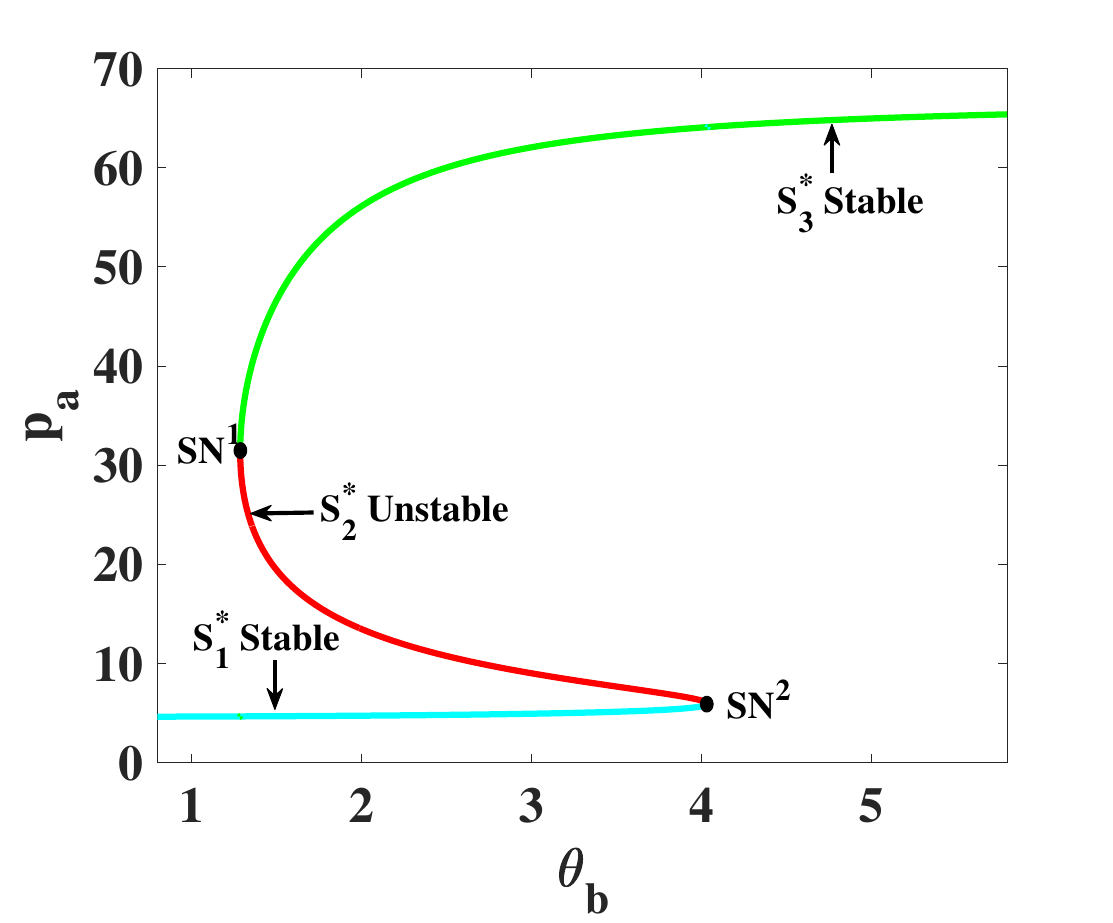}}
 \caption{(a) For $n_a=2,\,n_b=1,\,n_{aa}=2;$ Bifurcation plot exhibiting hysteresis effect with three steady states (b)   For $n_a=2,\,n_b=2,\,n_{aa}=1;$ Bifurcation plot exhibiting hysteresis effect with three steady states.}
 	\label{F4}
 \end{figure}
  \end{example}
 Now, we discuss some other cases in which we observe the monostable and bistable behavior of the system. Instead of showing bifurcation diagrams, here we show the phase portraits with null clines.\\
If $n_a = 2,\, n_b =2,\, n_{aa} = 2$, and the other parameters are chosen as:
 $\delta_a = 0.2,\,  \delta_b = 0.001,\, \gamma_a = 0.67,\, \gamma_b = .046,\, k_a = 0.7,\,k_b = 0.017,\, m_b = 8.375,\,  m_a =9,\, \theta_a = 0.6,\, A_1 =1,\, B_1 =0.01,\,\theta_{aa} =2.8,\,\theta_b=0.7. $ For this parameter, we obtained three steady states $S_1^*=(1.50492, 2.5496,5.2672, 43.3431),\, S_2^*=(5.00231, 0.430963,17.5081, 7.3264)$ and \\$S_3^*=(13.3447,  0.247431,46.7065, 4.2063)$	which are nodal sink (stable), saddle point (unstable) and nodal sink, respectively. The corresponding phase portrait diagram is shown in Fig. \ref{F5} (a).  Further, for this case, we choose $\theta_b=0.3,$ we obtained unique steady state  $S_1^*=(1.49473,2.5811,5.23155, 43.8786),$ which is locally asymptotically stable (nodal sink), shown in Fig. \ref{F5} (b).
 \begin{figure}[H]
 \centering
\subfigure[]{\includegraphics[scale=0.38]{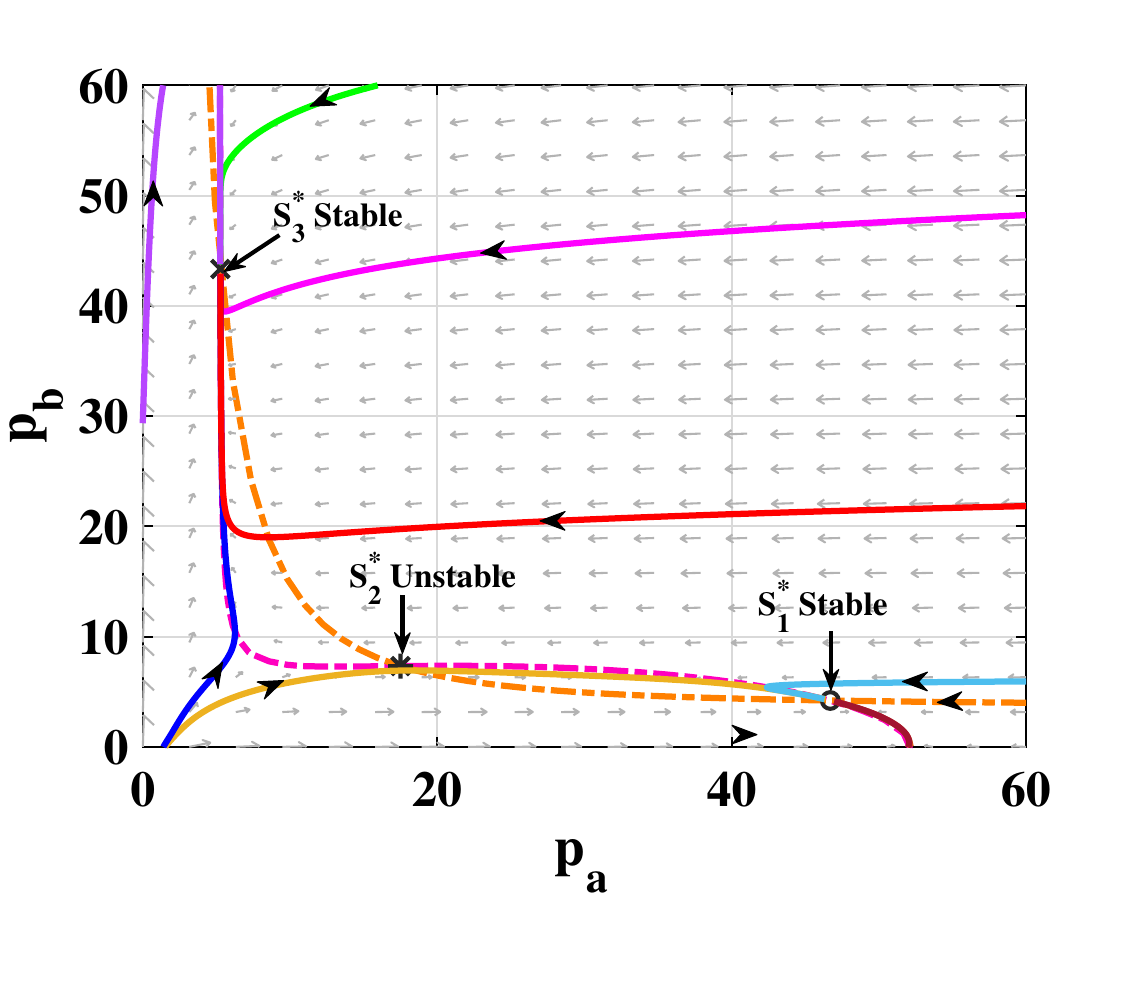}}
	\hfil
	\subfigure[]{\includegraphics[scale=0.38]{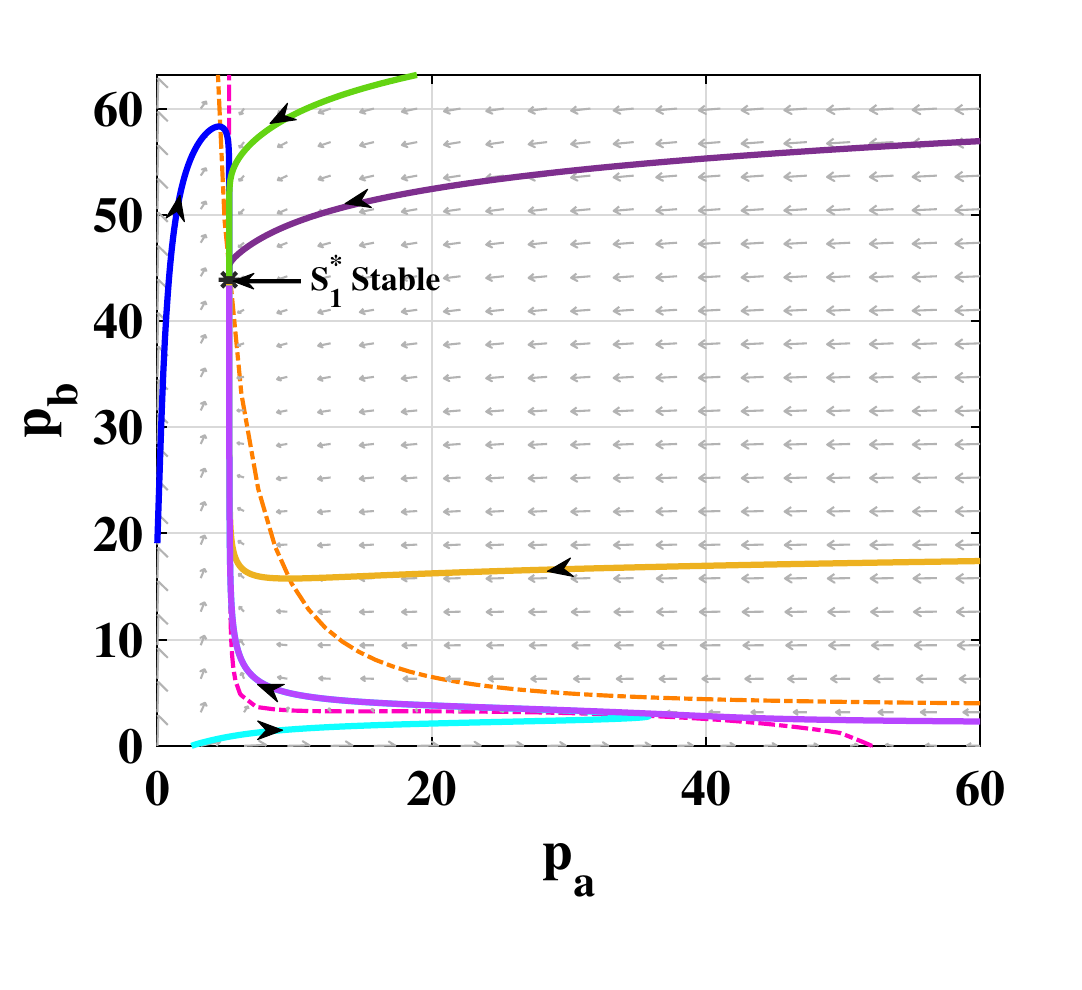}}
	\caption{For $n_a =n_b =n_{aa} = 2;$ (a) Phase portrait at $\theta_b=0.7,$ depicting bistability (b) Phase portrait at $\theta_b=0.3$.}
	\label{F5}
\end{figure}
If $n_a = 3,\, n_b =3,\, n_{aa} = 3$, and the other parameters are chosen as same as the case of $n_a = 2,\, n_b =2,\, n_{aa} = 2.$ For this parameter, we obtained three steady states $S_1^*=(1.56647, 0.455699,  5.48266, 7.74688),\,S_2^*=(2.2693,  0.295845,  7.94254,  5.02937)$ and $S_3^*=(14.6048,  0.217686, 51.1168, 3.70066)$	which are nodal sink (stable), saddle point (unstable) and nodal sink, respectively. Fig. \ref{Fn} (a) shows the corresponding phase portrait diagram.  Further, for this case, we choose $\theta_b=0.3,$ we obtained unique steady state  $S_1^*=(1.49663, 0.490592,  5.2382,  8.34006),$ which is locally asymptotically stable (nodal sink), shown in Fig. \ref{Fn} (b).
\begin{figure}[H]
\hfill
	\subfigure[]{\includegraphics[scale=0.38]{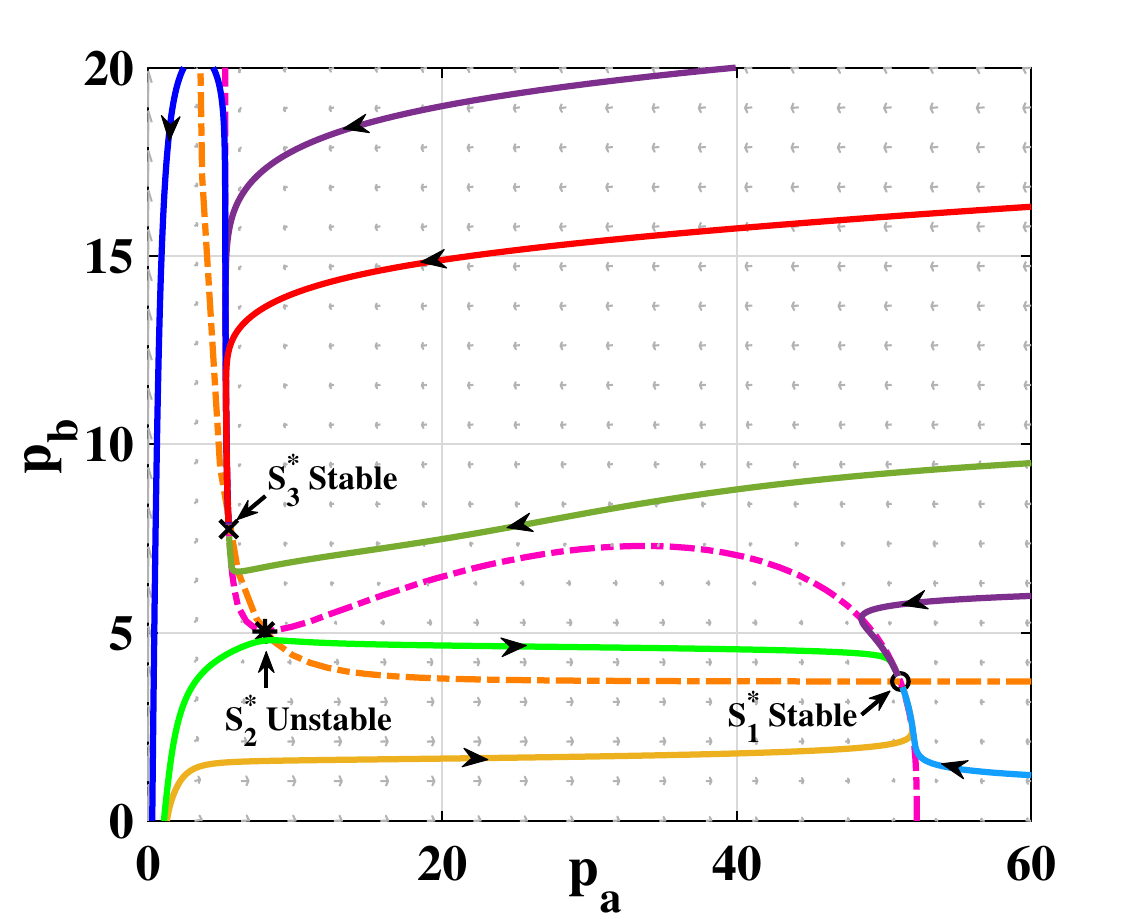}}
	\hfill
	\subfigure[]{\includegraphics[scale=0.38]{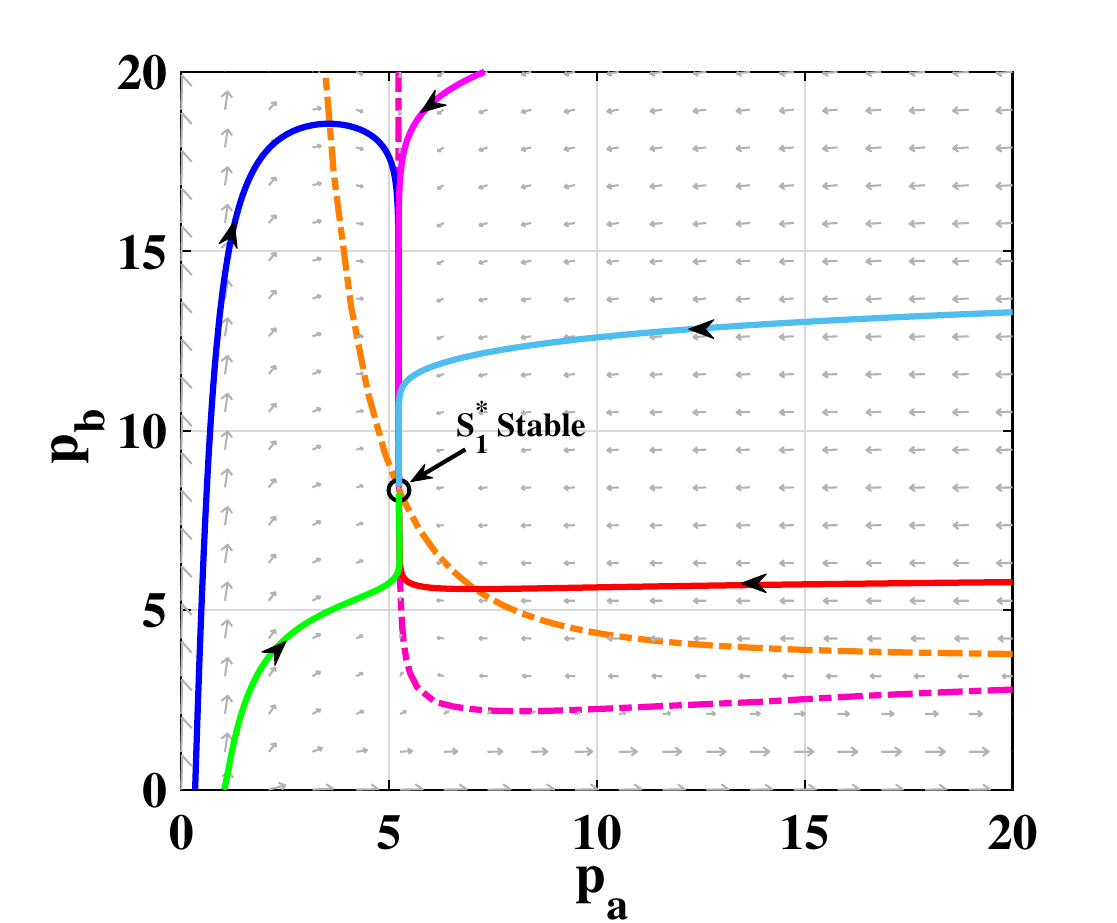}}
	\hfill
	\caption{For $n_a =n_b =n_{aa} =3;$ (a) Phase portrait at $\theta_b=0.7,$ depicting bistability (b) Phase portrait at $\theta_b=0.3$.}
	\label{Fn}
\end{figure}
\subsection{Conclusion and biological interpretation for this section}
In this section, we assume that the two genes $a$ and $b$ compete for binding to the promoter regions and any of the binding events can elicit gene expression. This type of biological regulation is modeled using a competitive OR logic function. Our results show that with any combination and degree of multimerization (Hill coefficient values), the network will exhibit at most bistability. We also found that cell state transition can occur through two mechanisms: (i) hysteresis; in which the cell remembers the past state and abruptly switches back and forth to another state as the input signal crosses a certain threshold (Fig.'s \ref{case1traj_new}a, \ref{F3}a, \ref{F4}). (ii) biphasic transition; in which we observe that cells can transit to the next state without a step-like switch (\ref{case1traj_new}c, \ref{F3}e). In this situation, there is no bifurcating attractor and hence no cells are "forced" to transition, instead, gene expression changes in a smooth continuous way. \\

\section{Mathematical Model [AND Logic]}
In this section, we model the co-regulation of node $a$ using non-competitive AND logic \cite{logics}. With this formalism, the model system takes the following form.\\
\begin{equation}\label{and}
	\begin{aligned}
		&\frac{dr_a}{dt}   =  \frac{m_a(\frac{p_a}{\theta_{aa}})^{n_{aa}}}{1+(\frac{p_a}{\theta_{aa}})^{n_{aa}}(\frac{p_b}{\theta_b})^{n_{b}}}-\gamma_a r_a+A_1, \\
		&\frac{dr_b}{dt}  =\frac{m_b}{1+(\frac{p_a}{\theta_a})^{n_{a}}}-\gamma_b r_b+B_1, \\
		&\frac{dp_a}{dt} =k_ar_a-\delta_ap_a,\\
		&\frac{dp_b}{dt} =k_br_b-\delta_bp_b,
	\end{aligned}
\end{equation}
with $r_a(0)\geq 0,\,r_b(0)\geq 0,\,p_a(0)\geq 0,\,p_b\geq 0.$ The parameter description is same as Section \ref{MM_OR}. 
\subsection{Mathematical Analysis}\label{3.1}
This subsection again presents the existence and stability analysis of the steady states for the model system (\ref{and}). We use the Routh-Hurwitz criterion to prove the local asymptotic stability of the steady states, while Sotomayor's theorem is used to derive the transversality conditions of saddle-node bifurcation. 

We obtained the steady state $S^*=(r_a^*,r_b^*,p_a^*,p_b^*)$ by setting the right hand side of the model (\ref{1}) to zero.
From third and fourth equations of model (\ref{1}), we obtained ${r_a}^*=\frac{\delta_a {p_a}^*}{k_a}$ and ${r_b}^*=\frac{\delta_b {p_b}^*}{k_b},$ respectively, as $\frac{dp_a}{dt}=0$ and $\frac{dp_b}{dt}=0.$ Using  $\frac{dr_b}{dt}=0$ and then putting the value of $r_b,$ we get ${p_b}^*=\frac{1}{\gamma_b}\bigg\{B_1'+\frac{m_2}{\theta_a^{n_a}+{{p_a}^*}^{n_a}}\bigg\}.$ Finally, using the value of $p_b$ and $r_a$ in $\frac{dr_a}{dt}=0,$ we get a polynomial in $p_a$ as
\begin{equation}\label{2and}
	\gamma_b^{n_b}(\theta_a^{n_a}+p_a^{n_a})^{n_b}\bigg(m_1p_a^{n_{aa}}+\theta_b^{n_b}\theta_{aa}^{n_{aa}}(A_1'-\gamma_ap_a)\bigg)+(A_1'-\gamma_ap_a)p_a^{n_{aa}}\bigg(m_2+B_1'(\theta_a^{n_a}+p_a^{n_a})\bigg)^{n_b}=0,
\end{equation}
where $A_1'=\frac{A_1 k_a}{\delta_a},\,B_1'=\frac{B_1 k_b}{\delta_b},\,m_1=\frac{m_a k_a \theta_b^{n_b}}{\delta_a},$ and $m_2=\frac{m_b k_b \theta_a^{n_a}}{\delta_b}.$ Here, $p_a^*$ is a positive real root of (\ref{2and}) which we discuss in different cases explicitly. \\
\textbf{Case 1:} When $n_{aa}=n_a=n_b=1,$ then from (\ref{2and}), we get a cubic polynomial in $p_a$
\begin{equation}\label{3and}
	f(p_a)=Ap_a^3+Bp_a^2+Cp_a+D,
\end{equation}
where $A=-\gamma_a B_1',\,B=\gamma_b (m_1-\gamma_a \theta_{aa}\theta_b)-\gamma_a(m_2+B_1'\theta_a)+A_1'B_1'\\
C=\gamma_bA_1'\theta_b\theta_{aa}+\gamma_b\theta_a (m_1-\gamma_a \theta_{aa}\theta_b)+A_1'(m_2+B_1'\theta_a)$ and $D=\gamma_b\theta_aA_1'\theta_b\theta_{aa}.$
The nature of the curve $f$ can be observe from (\ref{3and}) as 
\begin{enumerate}
	\item[(i)] $f(0)>0$ because $D>0,$
	\item[(ii)] $f\to -\infty\, (\infty)$ as $p_a\to \infty\,(-\infty)$ because $A<0,$
\item[(iii)] The number of positive real roots of equation (\ref{3and}) can be seen using Descarte's rule of sign, mentioned in below Table (\ref{T1and}).
	\begin{table}[H]
		\centering
		\begin{tabular}{|cccc|c|}
			\hline
			\multicolumn{4}{|c|}{\textbf{Coefficients}}                                                                       & \multirow{2}{*}{\textbf{\begin{tabular}[c]{@{}c@{}}Number of possible positive\\ real roots of $f(p_a)=0$\end{tabular}}} \\ \cline{1-4}
			\multicolumn{1}{|c|}{\textbf{A}} & \multicolumn{1}{c|}{\textbf{B}} & \multicolumn{1}{c|}{\textbf{C}} & \textbf{D} &                                                                                                                         \\ \hline
			\multicolumn{1}{|c|}{-}          & \multicolumn{1}{c|}{-}          & \multicolumn{1}{c|}{-}          & +          & 1                                                                                                                       \\ \hline
			\multicolumn{1}{|c|}{-}          & \multicolumn{1}{c|}{+}          & \multicolumn{1}{c|}{-}          & +         & 3,1                                                                                                                       \\ \hline
			\multicolumn{1}{|c|}{-}          & \multicolumn{1}{c|}{+}          & \multicolumn{1}{c|}{+}          & +          & 1                                                                                                                       \\ \hline
			\multicolumn{1}{|c|}{-}          & \multicolumn{1}{c|}{-}          & \multicolumn{1}{c|}{+}          & +          & 1                                                                                                                     \\ \hline
		\end{tabular}
		\caption{Number of possible positive real roots of $f(p_a)=0.$}
		\label{T1and}
	\end{table}
\end{enumerate}
Therefore, $f(p_a)=0$ will have either unique or three positive real roots. Now, to discard the existence of three positive real roots, we use the nullcline method. From the system (\ref{and}), we have
\begin{equation}\label{6_AND}
	\begin{aligned}
		&F_1(p_a,p_b):=\frac{m_{a} k_a\theta_bp_a}{\delta_a(\theta_{aa}\theta_{b}+p_ap_b)}+A_1'-\gamma_{a}p_a \\
		&F_2(p_a,p_b):=\frac{m_{b}k_b\theta_a}{\delta_b(\theta_{a}+p_a)}+B_{1}'-\gamma_{b}p_b,
	\end{aligned}
\end{equation}
First to observe the nature of curve $F_2(p_a,p_b)=0,$ we reduce it in the form of
\begin{equation}\label{7and}
	p_b=g_1(p_a):=\frac{1}{k_2}\Big\{\frac{m_{b}\theta_a}{(\theta_{a}+p_a)}+B_{1}\Big\}\tag{$\mathscr{A}_1$}.
\end{equation}
The following conclusion can be made from (\ref{7and})
\begin{enumerate}
	\item[$\mathscr{A}_1$(i)] $g_1(0)=\frac{1}{k_2}(m_b +B_{1})>\frac{B_1}{k_2}>0$ 
	\item[$\mathscr{A}_1$(ii)] $g_1(p_a)\to \frac{B_1}{k_2}$ as $p_a\to \pm\infty$ 
	\item[$\mathscr{A}_1$(iii)] $g'_1(p_a)<0$, i.e, $g_1(p_a) $ is a decreasing function of $p_a.$
\end{enumerate}
The graph of $F_2(p_a,p_b)=0$ (or $p_b=g_1(p_a)$) is shown in Fig. \ref{fig:g1and}.\\
Now, to observe the nature of curve $F_1(p_a,p_b)=0,$ we reduce it in the form of
\begin{equation}\label{8_AND}
	p_b=g_2(p_a):=\frac{\theta_b(m_ap_a+(A_1-k_1p_a)\theta_{aa})}{p_a(k_1p_a-A_1)}.\tag{$\mathscr{B}_1$}
\end{equation}
Equivalently
\begin{equation}\label{9_AND}
	p_b=g_2(p_a):=\frac{p_a\big(m_a-k_1\theta_{aa}\big)+A_1k_1}{p_a(k_1p_a-A_1)}.\tag{$\mathscr{B}_2$}
\end{equation}
Also, we find
\begin{equation}\label{10_AND}
	g_2'(p_a)=\frac{-\theta_b\Big(k_1p_a^2\big(m_a-A_1\theta_{aa}\big)+k_1p_aA_1\theta_{aa}+A_1\theta_{aa}\big(k_1p_a-A_1\big)\Big)}{p_a^2(k_1p_a-A_1)^2}.\tag{$\mathscr{C}_1$}
\end{equation}
Equivalently
\begin{equation}\label{11_AND}
	g_2'(p_a)=\frac{-\theta_b\bigg(k_1p_a\Big(p_a(m_a-k_1\theta_{aa})+A_1\theta_{aa}\Big)+A_1\theta_{aa}\big(k_1p_a-A_1\big)\bigg)}{p_a^2(k_1p_a-A_1)^2}.\tag{$\mathscr{C}_2$}
\end{equation}
The following conclusion can be made from (\ref{8_AND})
\begin{enumerate}
	\item[$\mathscr{B}$(i)]  $g_2(p_a)$ is a rational function, and it has two vertical asymptotes at $p_a=0$ and $p_a=\frac{A_1}{k_{1}}$ and $p_b=0$ (i.e., $p_a$-axis) is a horizontal asymptote. 
	\item[$\mathscr{B}$(ii)]We note that for $p_a\in \big(0,\frac{A_1}{k_{1}}\big)$, we have $g_2(p_a)<0$ so if positive real roots of $g_2(p_a)=0$ exits then it always lie on the right-hand side of vertical asymptote $p_a=\frac{A_1}{k_{1}}$.
	\item[$\mathscr{B}$(iii)]The $p_a-$intercept can be find with $g_2(p_a)=0,$ which implies that $p_a=-\frac{A_1\theta_{aa}}{(m_a-k_1\theta_{aa})}.$ The sign of $(m_a-k_1\theta_{aa})$ will provide the unique $p_a-$intercept of $g_2(p_a).$ Now, we will discuss three cases to understand the nature of curve $p_b=g_2(p_a).$
	\begin{itemize}	
		\item[(a)]When $m_a-k_1\theta_{aa}>0$ and $p_a>\frac{A_1}{k_{1}}$   then there will no positive intercept; also $g_2(p_a)>0$ and $g_2'(p_a)<0$ which implies that $g_2(p_a)$ is strictly decreasing. The corresponding diagram is shown in Fig. \ref{AND_NULL} (a).
		\item[(b)] When $m_a-k_1\theta_{aa}<0$ then there will unique positive intercept $p_a=-\frac{A_1\theta_{aa}}{(m_a-k_1\theta_{aa})}.$  \ref{AND_NULL} (c). When $\frac{A_1}{k_{1}}<p_a<-\frac{A_1\theta_{aa}}{(m_a-k_1\theta_{aa})}$ then $g_2(p_a)>0$ and $g_2'(p_a)<0$ which implies that $g_2(p_a)$ is strictly decreasing. The corresponding diagram is shown in Fig. \ref{AND_NULL} (c).
	 When $p_a>-\frac{A_1\theta_{aa}}{(m_a-k_1\theta_{aa})}$ then $g_2(p_a)<0,$ and never cross the $p_a$-axis. 
	\end{itemize}	
\end{enumerate}	
All possibilities of the existence of a steady state are shown in Fig. \ref{AND_NULL} (b) and Fig. \ref{AND_NULL} (d) corresponding to \ref{AND_NULL} (a) and Fig. \ref{AND_NULL} (c), respectively. Therefore, from the above discussion, we discard the existence of three positive steady states.
\begin{figure}[H]
	\centering
	\includegraphics[scale=0.4]{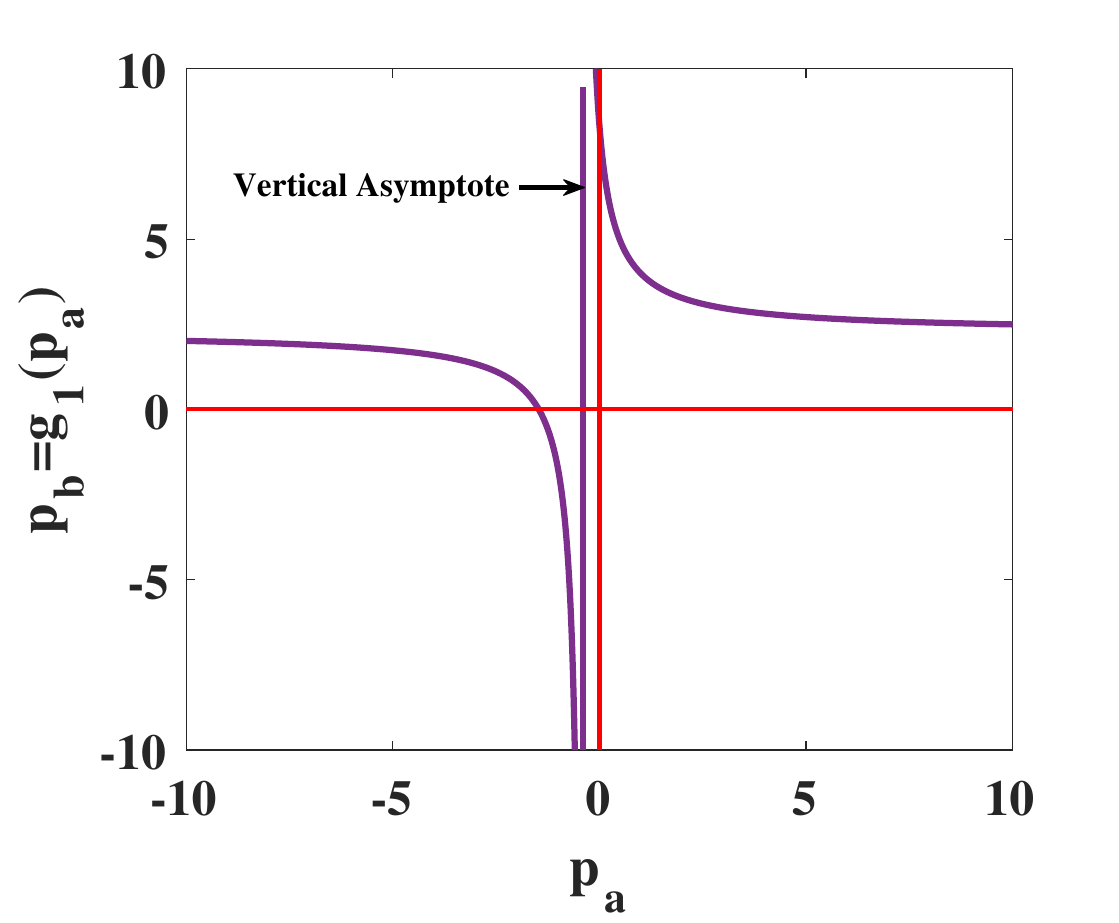}
	\caption{For $ n_a=n_b=n_{aa}=1;$ Plot of $g_1(p_a)$ as a decreasing function of $p_a.$ }
	\label{fig:g1and}
\end{figure}
\begin{figure}[H]
\centering
	\subfigure[]{\includegraphics[scale=0.4]{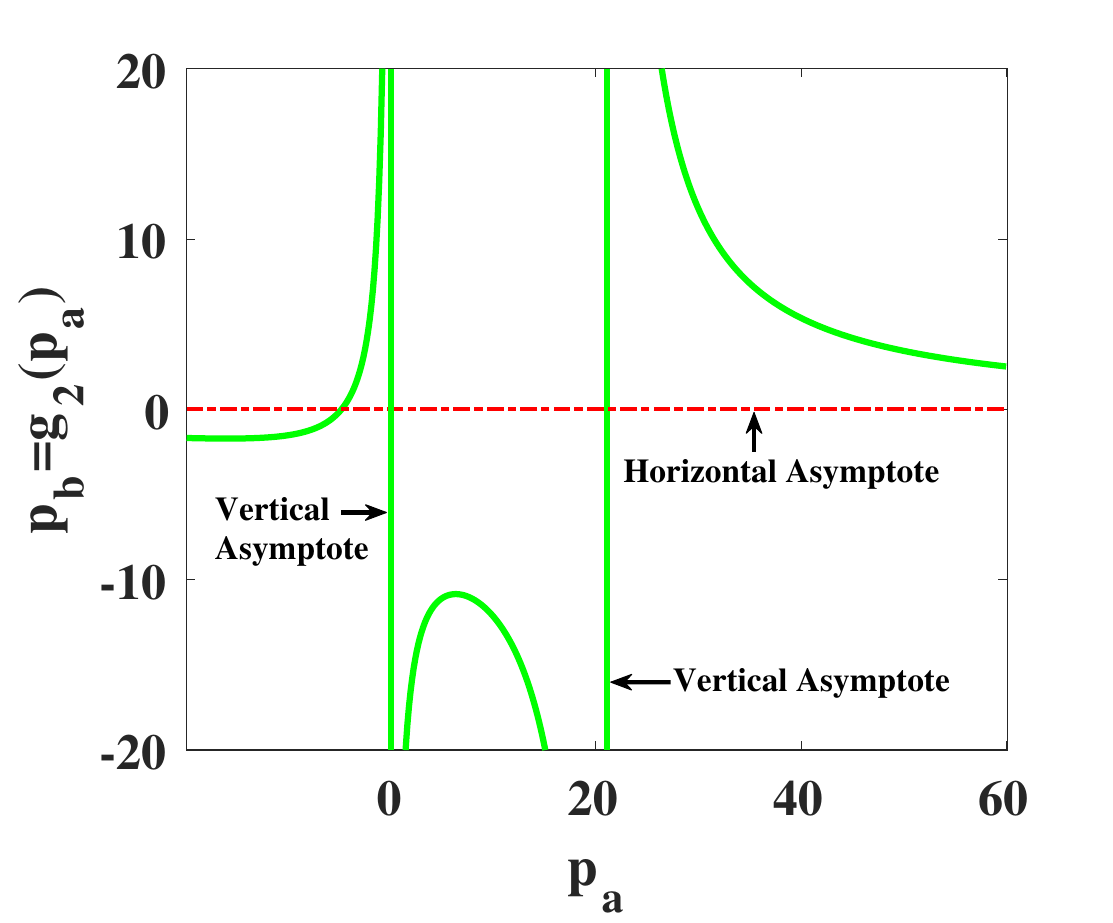}}
	\hfil
	\subfigure[]{\includegraphics[scale=0.4]{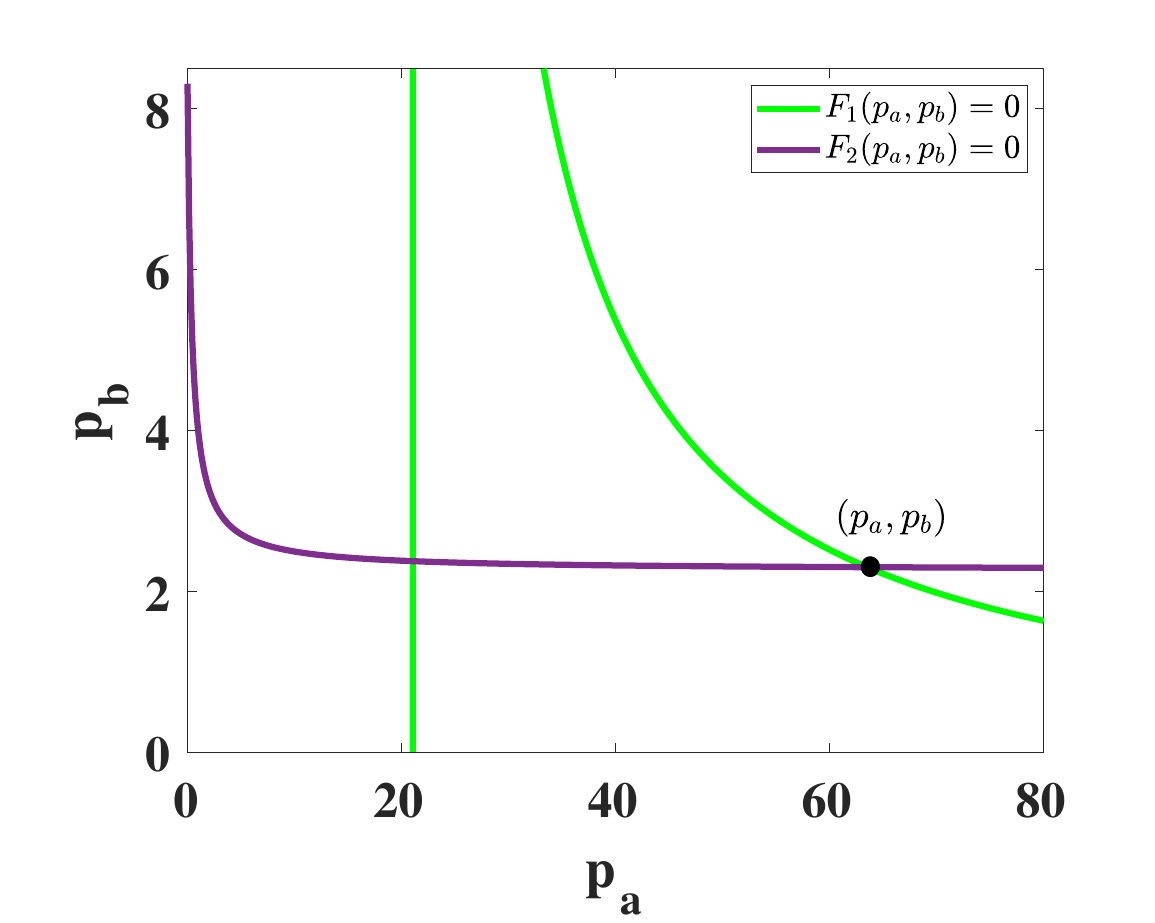}}
	\hfil
 \end{figure}
 \begin{figure}[H]
 \centering
	\subfigure[]{\includegraphics[scale=0.4]{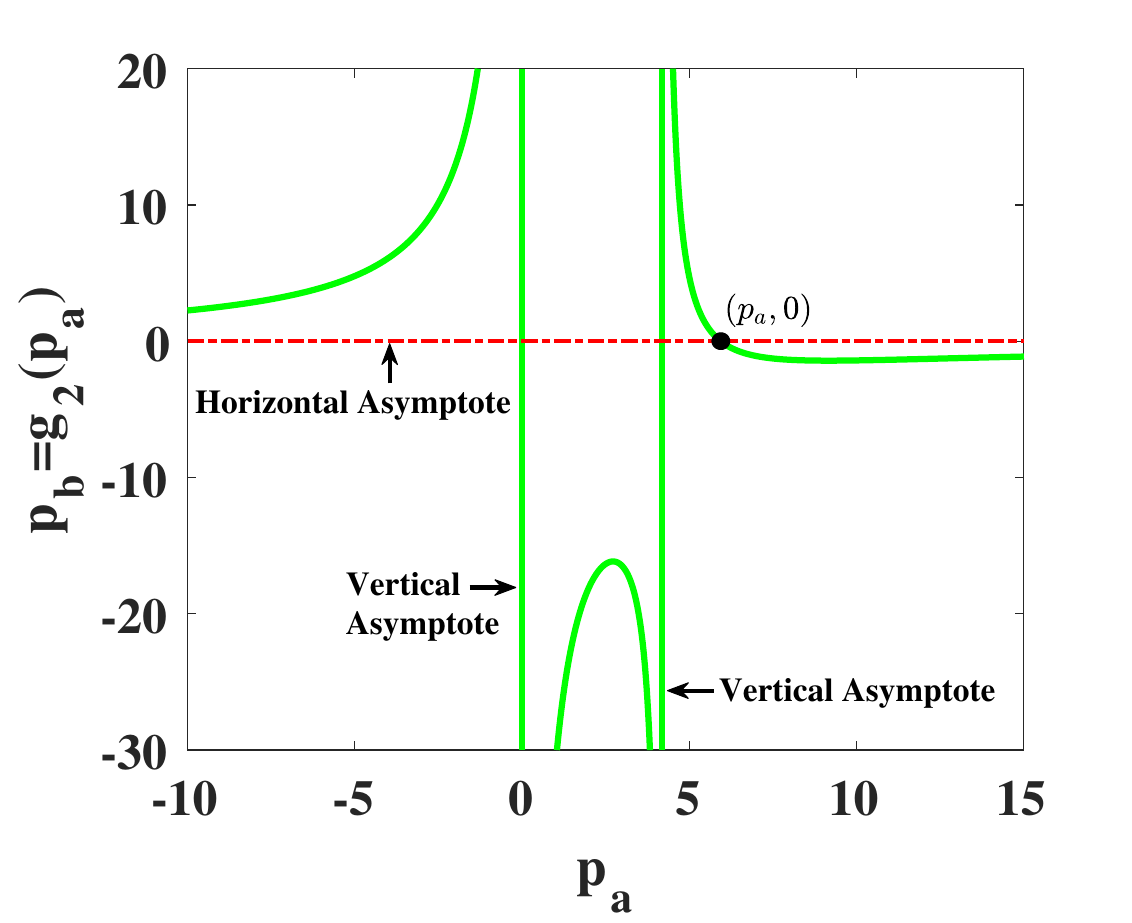}}
	\hfil
	\subfigure[]{\includegraphics[scale=0.4]{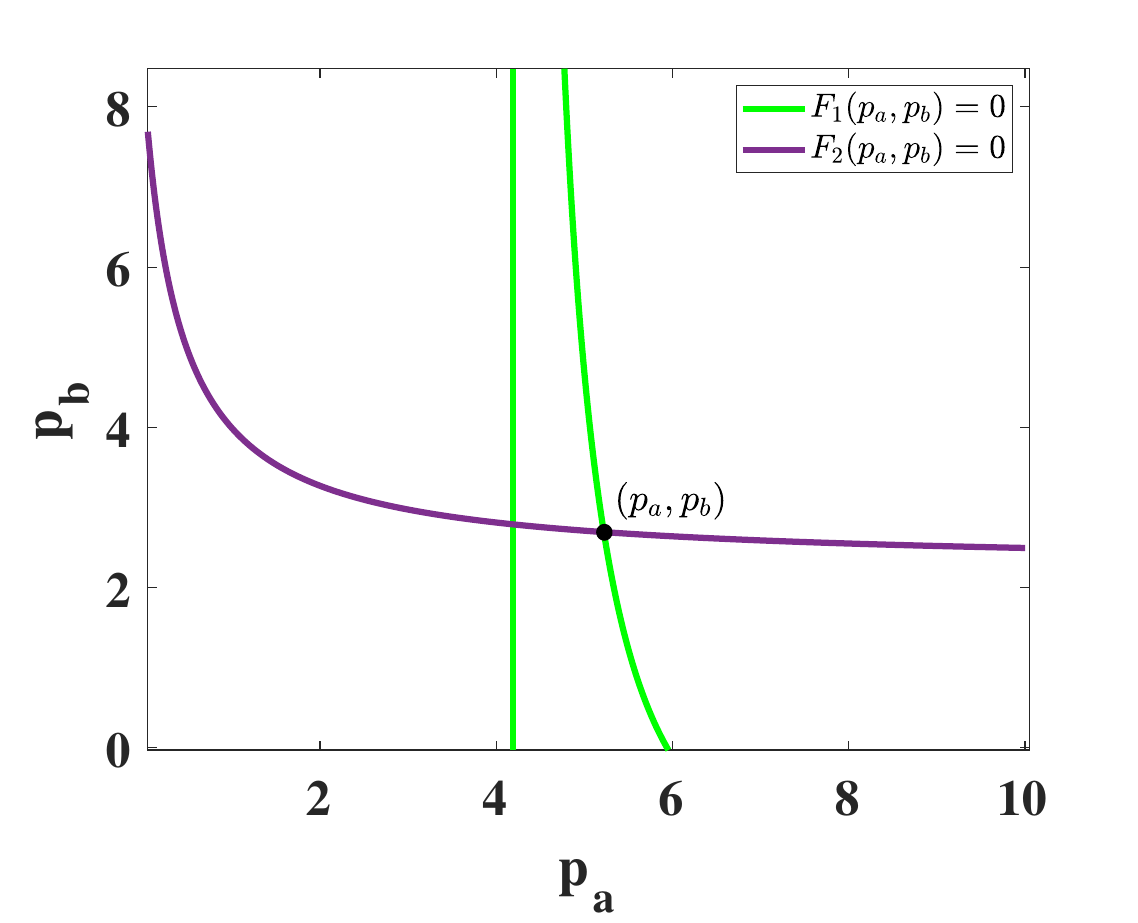}}
	\caption{For $ n_a=n_b=n_{aa}=1;$ (a) The graph depicting no positive $p_a$-intercept of $g_2(p_a)$ (b)  Existence of unique positive steady state corresponding to Fig. \ref{AND_NULL} (a). (c) The graph depicting unique positive $p_a$-intercept of $g_2(p_a)$ (d) Existence of unique positive steady state corresponding to Fig. \ref{AND_NULL} (c). }
	\label{AND_NULL}
\end{figure}
\textbf{Case 2:} When $n_a=n_b=1,\,n_{aa}=2,$ then from (\ref{and}), we get a quartic polynomial in $p_a$
\begin{equation}\label{4and}
	f(p_a)=Ap_a^4+Bp_a^3+Cp_a^2+Dp_a+E,
\end{equation}
where $A=B_1'\,\gamma_{a},\,B=-A_1'\,B_1'-\gamma_{b}\,m_{1}+\gamma_{a}\,\left(m_{2}+B_1'\,\theta _{a}\right),\,
C=\gamma_{a}\,\gamma_{b}\,\theta _{b}\,{\theta _{aa}}^2-A_1'\,\left(m_{2}+B_1'\,\theta _{a}\right)\\-\gamma_{b}\,m_{1}\,\theta _{a},\,D=-\gamma_{b}\,\theta _{b}\,{\theta _{aa}}^2\,\left(A_1'-\gamma_{a}\,\theta _{a}\right),$ and $E=-A_1'\,\gamma_{b}\,\theta _{a}\,\theta _{b}\,{\theta _{aa}}^2.$\\
The number of positive real roots of equation (\ref{4and}) can be seen using Descarte's rule of sign (mentioned in below Table (\ref{T2})); therefore, $f(p_a)=0$ will have either unique or three positive real roots.
\subsubsection{Local stability of steady state}
\begin{theorem}\label{existence_AND}
	The steady state $S_1^{*}$ is locally asymptotically stable if $\epsilon_1>0,\,\epsilon_3>0,\,\epsilon_4>0,$ and $\epsilon_1\epsilon_2\epsilon_3-\epsilon_1^2\epsilon_4 - \epsilon_3^2 >0,$ where $\epsilon_i,\,i=1,..,4$ are provided in the proof.
\end{theorem}
\begin{proof}At $S_1^{*}, $ the Jacobian matrix of the model system (\ref{1}) is given by
	\begin{center}
		$J_{S_1^*}=
		\left[\begin{array}{cccc} -\gamma_a & 0 & X&Y\\ 0 & -\gamma_b & Z&0\\ k_a &0  & -\delta_a&0\\0&-k_b&0&-\delta_b \end{array}\right].$
	\end{center}
	where $X=\frac{\mathrm{m_a}\,n_{aa}\,{\mathrm{p_a}}^{n_{aa}-1}\,{\theta _{b}}^{2\,n_{b}}\,{\theta _{aa}}^{n_{aa}}}{{\left({\mathrm{p_a}}^{n_{aa}}\,{\mathrm{p_b}}^{n_{b}}+{\theta _{b}}^{n_{b}}\,{\theta _{aa}}^{n_{aa}}\right)}^2}
	,\,Y=-\frac{\mathrm{m_a}\,n_{b}\,{\mathrm{p_a}}^{2\,n_{aa}}\,{\mathrm{p_b}}^{n_{b}-1}\,{\theta _{b}}^{n_{b}}}{{\left({\mathrm{p_a}}^{n_{aa}}\,{\mathrm{p_b}}^{n_{b}}+{\theta _{b}}^{n_{b}}\,{\theta _{aa}}^{n_{aa}}\right)}^2},$ and $Z=-\frac{\mathrm{m}_b\,n_{a}\,{\mathrm{p_a}}^{n_{a}-1}\,{\theta _{a}}^{n_{a}}}{{\left({\mathrm{p_a}}^{n_{a}}+{\theta _{a}}^{n_{a}}\right)}^2}.$
	The characteristic equation of $J_{S_1^*}$ is given as
	\begin{equation}\label{9N}
\lambda^4+\epsilon_1\lambda^3+\epsilon_2\lambda^2+\epsilon_3\lambda+\epsilon_4=0.
	\end{equation}
	Here, $\epsilon_1=\gamma_a+\gamma_b+\delta_a+\delta_b,\epsilon_2=\gamma_a\gamma_b+(\gamma_a+\gamma_b)(\delta_a+\delta_b)+\delta_a\delta_b-k_aX,\,\epsilon_3=\delta_a\delta_b(\gamma_a+\gamma_b)+(\delta_a+\delta_b)(\gamma_a\gamma_b)-k_a(\delta_b+\gamma_b)X,$ and  $\epsilon_4=\gamma_a\gamma_b\delta_a\delta_b-k_ak_bYZ-k_a\gamma_b\delta_bX.$
	By the Routh-Hurwitz criteria, all the roots of (\ref{9N}) have negative real parts provided $\epsilon_1>0,\,\epsilon_3>0,\,\epsilon_4>0,$ and $\epsilon_1\epsilon_2\epsilon_3-\epsilon_1^2\epsilon_4 - \epsilon_3^2 >0.$
\end{proof}
\subsubsection{Saddle Node Bifurcation}
We find the transversality condition the same as discussed in Subsection \ref{sn_OR}.  Define $f=(f_1,f_2,f_3,f_4)^T,$ where $f_1$ are defined as
\begin{equation*}
f_1  =  \frac{m_a(\frac{p_a}{\theta_{aa}})^{n_{aa}}}{1+(\frac{p_a}{\theta_{aa}})^{n_{aa}}(\frac{p_b}{\theta_b})^{n_{b}}}-\gamma_a r_a+A_1, 
	\end{equation*}
and $f_2 ,\,f_3,$ and $f_4$ expression is same as Subsection \ref{sn_OR}.
Then the model system (\ref{and})
experiences a saddle-node bifurcation  at $(S_{SN}^*,\theta_b^{SN})$, if the following transversality conditions \cite{perko_2013} are satisfied:
\begin{equation*}
	w^Tf_{\theta_b}(S_{SN}^*,\theta_b^{SN})=w_1\frac{m_a n_b {\theta_b^{n_b-1}}^{SN}{p_b^{n_b}}^{{SN}^*}{p_a^{2n_{aa}}}^{{SN}^*}}{(\mathcal{X})^2}\neq0. \tag{$\mathcal{T}_{11}$}\end{equation*}
and  
\begin{align*}
	w^T[D^2f(S_{SN}^*,\theta_b^{SN})(v,v)]=&\frac{w_1}{\mathcal{X}^3}\Biggl\{\mathcal{O}\bigg[v_3^2\Big(\theta_{aa}^{n_{aa}}{\theta_{b}^{n_{b}}}^{SN}(n_{aa}-1){p_a^{(n_{aa}-2)}}^{{SN}^*}-(n_{aa}+1){p_b^{n_b}}^{{SN}^*}\\&{p_a^{(2n_{aa}-2)}}^{{SN}^*}\Big)-4v_3v_4n_b{p_a^{(2n_{aa}-1)}}^{{SN}^*}{p_b^{(n_{b}-1)}}^{{SN}^*}-v_4^2\mathcal{O}_1\\&\Big(\theta_{aa}^{n_{aa}}{\theta_{b}^{n_{b}}}^{SN}(n_{b}-1){p_b^{(n_{b}-2)}}^{{SN}^*}-(n_{b}+1){p_a^{n_{aa}}}^{{SN}^*}{p_b^{(2n_{b}-2)}}^{{SN}^*}\Big)\bigg]\Biggr\}\\&	-\frac{w_2v_3^2m_b n_a \theta_a^{n_a}}{(\theta_a^{n_a}+{p_a^{n_a}}^{{SN}^*})^3} \bigg\{{p_a^{n_a-2}}^{{SN}^*}\theta_a^{n_a}(n_a-1)-(n_a+1){p_a^{2n_a-2}}^{{SN}^*}\bigg\}	\\&\neq 0,\tag{$\mathcal{T}_{22}$}
\end{align*}
where $\mathcal{X}=\theta_{aa}^{n_{aa}}\theta_b^{n_b}+{p_a^{n_{aa}}}^{{SN}^*}{p_b^{n_b}}^{{SN}^*},\,\mathcal{O}=n_{aa}m_a\theta_b^{2n_b}\theta_{aa}^{n_{aa}},$ and $\mathcal{O}_1=m_an_b\theta_b^{n_b}{p_{aa}^{2n_{aa}}}^{{SN}^*}.$	
\subsubsection{Numerical simulation and discussion corresponding to model system (\ref{and})}
The purpose of this section is to provide validation for the theoretical results that were covered in Subsection \ref{3.1}. We have talked about the system's behavior in terms of monostability and bi-stability with fictitious parameters.
\begin{example} $=n_a=n_b=n_{aa}=1$ (case 1) and the other parameters are chosen as $\delta_a =0.99,\,   \delta_b = 0.06,\, \gamma_a = 3.9,\, \gamma_b =0.35,\, k_a = 0.5,\, k_b = 0.09,\,m_a =0.02,\, m_b =0.9,\,  \theta_a = 0.0086,\, A_1 = 0.8,\, B_1 = 0.009,\, \theta_{aa} = 0.086,\,\theta_b = (0.0001,5).$ For this set of parameters, we obtain a unique stable steady state as discussed in Case 1(and). The corresponding bifurcation plot is shown in Fig. \ref{case1traj_AND} (a). To ensure the stability of $S_1^*,$  we choose $\theta_b=2.5$ from the range of $\theta_b.$  The coefficients of characteristic equation at $S_1^*$ are  calculated as $\epsilon_1=5.3>0,\,\epsilon_2=5.8008>0,\,\epsilon_3=1.6504>0,\,\epsilon_4=0.0790>0,$ and $\epsilon_1\epsilon_2\epsilon_3-\epsilon_1^2\epsilon_4 - \epsilon_3^2=45.795>0.$ Therefore, from Theorem \ref{existence_AND}, we conclude that $S_1^*$ is locally asymptotically stable. We draw the phase portrait for the considered $\theta_b$, shown in Fig. \ref{case1traj_AND} (b). We can see that all trajectories starting from any initial points are converging to $S_1^*=(0.2106,0.2181, 0.1064, 0.3271).$
\begin{figure}[H]
\centering
	\subfigure[]{\includegraphics[scale=0.35]{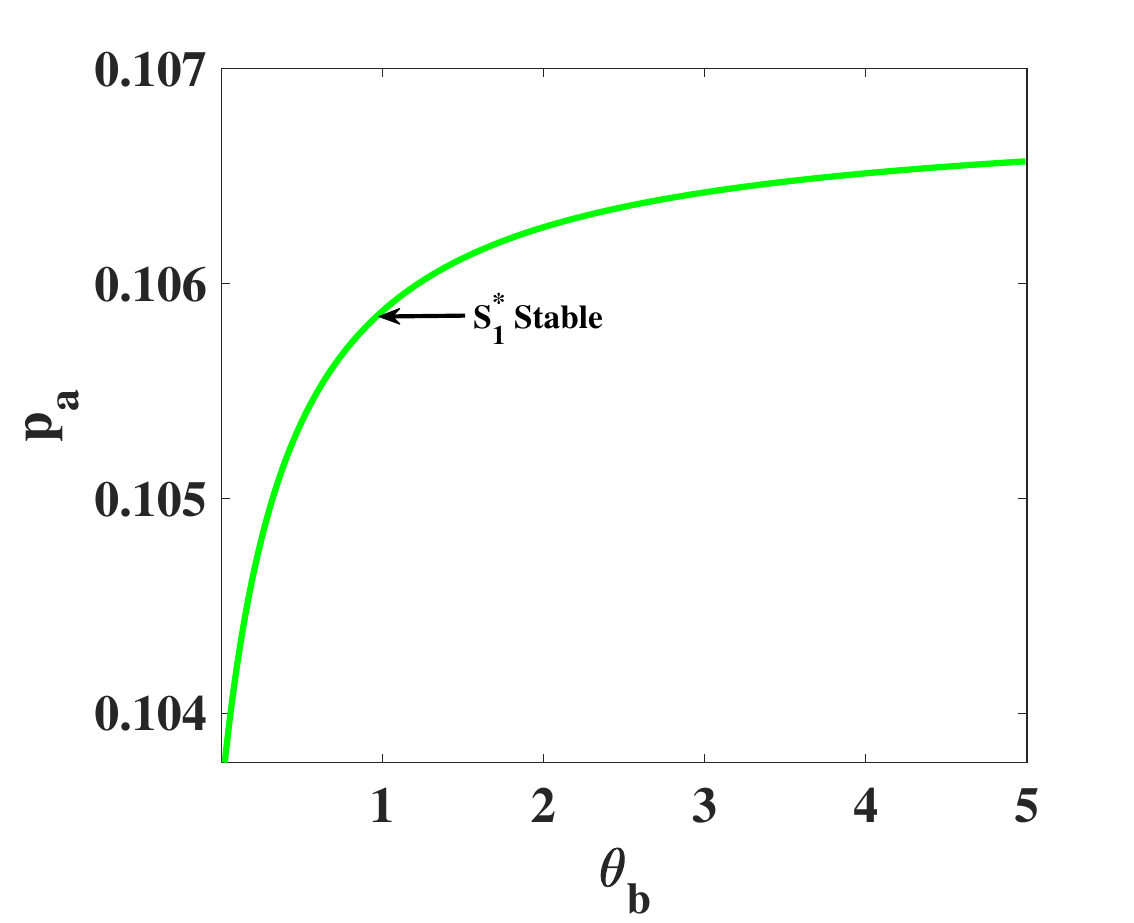}}
	\hfil
	\subfigure[]{\includegraphics[scale=0.35]{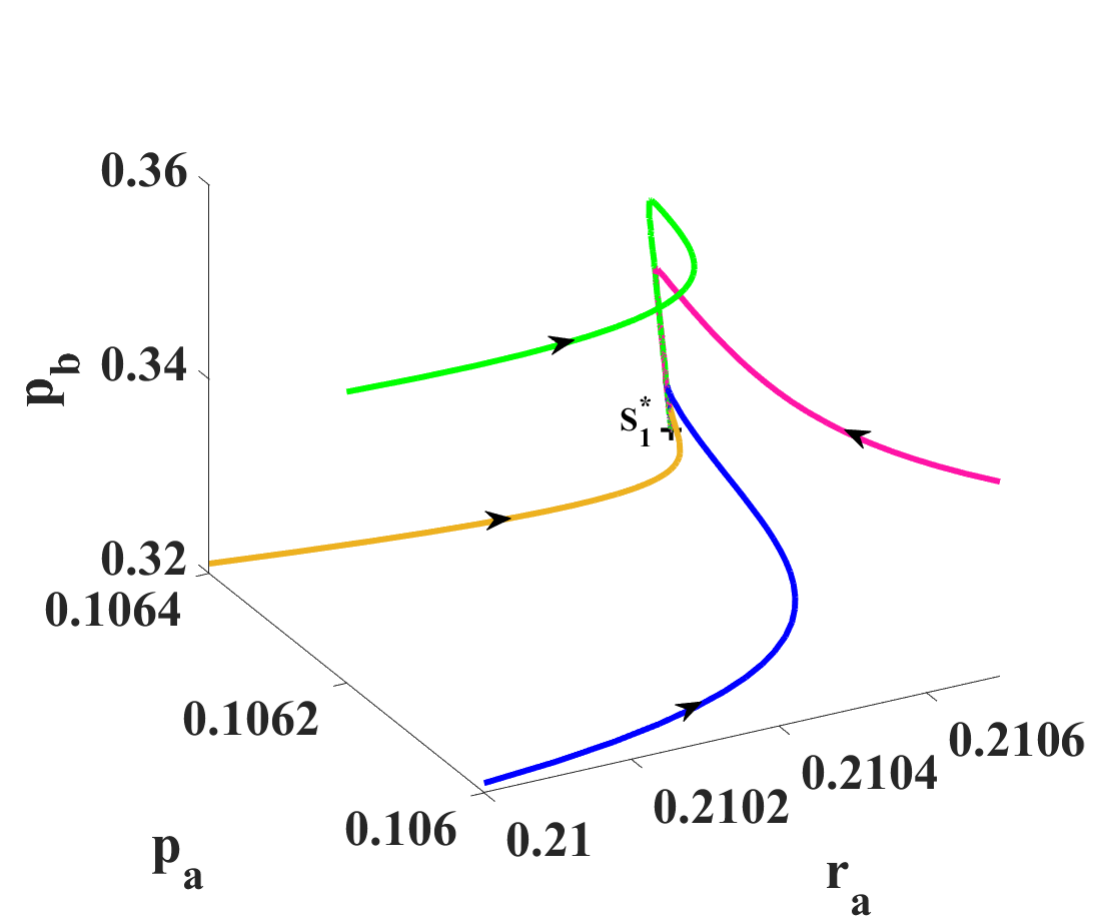}}
	\caption{For $n_a=1,\,n_b=1,\,n_{aa}=1$; (a) Bifurcation plot with unique stable steady state (b) Phase portrait at $\theta_b=2.5.$}
	\label{case1traj_AND}
\end{figure}
\end{example}
\begin{example} $n_a=n_b=1,\,n_{aa}=2$ (case 2), and the other parameters are chosen as: 
$\delta_a = 0.8,\, \delta_b = 0.03,\,  \gamma_a = 0.91,\,  \gamma_b = 0.235,\, k_a = 0.1,\,  k_b = 0.04,\,  m_b = 0.2,\,  m_a = 6,\,  \theta_a = 0.55,\,  A_1 = 5.8,\, B_1 = 0.02,\, \theta_{aa} = 1.6,\,  
\theta_b = (0.4,18).$ In this range of $\theta_b,$ there is unique ($S_1^*$) steady state in $\theta_b^1$, three ($S_1^*,\,S_2^*,\,S_3^*$) in $\theta_b^2$ and unique ($S_3^*$) in $\theta_b^3,$ such that $\theta_b=\theta_b^1\cup\theta_b^2\cup\theta_b^2$ where $\theta_b^1=(0.4,{\theta_b^{SN}}^1),\,\theta_b^2=[{\theta_b^{SN}}^1,{\theta_b^{SN}}^2]$ and $\theta_b^2=({\theta_b^{SN}}^2,18).$ The steady states $S_1^*,\,S_3^*$ are stable and ($S_2^*$) is unstable. The corresponding bifurcation diagram is shown in Fig. \ref{F3_AND} (a).\\	
To ensure the stability of $S_1^*,$ we chosen $\theta_b=0.6\in\theta_b^1$ and the coefficients of characteristic equation at $S_1^*=(8.4987,0.3754,1.0623,0.5006)$ are  calculated  as $\epsilon_1=1.975>0,\,\epsilon_2=0.922>0,\,\epsilon_3=0.1344>0,\,\epsilon_4=0.0031>0,$ and $\epsilon_1\epsilon_2\epsilon_3-\epsilon_1^2\epsilon_4 - \epsilon_3^2=0.2147>0.$ Therefore, from Theorem \ref{existence_AND}, we conclude that $S_1^*$ is locally asymptotically stable. We draw the phase portrait for the considered parameters, shown in Fig. \ref{F3_AND} (b). We can see that all trajectories starting from any initial points are converging to $S_1^*.$\\
Further, we chosen $\theta_b=6\in\theta_b^2$ and the coefficients of characteristic equation at $S_1^*=(10.9693,0.3288,1.3712,0.4383)$ are  calculated  as $\epsilon_1=1.975>0,\,\epsilon_2=0.6093>0,\,\epsilon_3=0.0516>0,\,\epsilon_4=0.000992>0,$ and $\epsilon_1\epsilon_2\epsilon_3-\epsilon_1^2\epsilon_4 - \epsilon_3^2=0.0555>0.$ The coefficients of characteristic equation at $S_2^*=(17.4281, 0.25666,2.17851,0.342213 )$ are  calculated  as $\epsilon_1=1.975>0,\,\epsilon_2=0.3530>0,\,\epsilon_3=-0.01634<0,\,\epsilon_4=-0.000922<0,$ and $\epsilon_1\epsilon_2\epsilon_3-\epsilon_1^2\epsilon_4 - \epsilon_3^2=-0.0081<0.$ The coefficients of characteristic equation at $S_3^*=(282.698 ,0.0981496,35.3372,0.130866 )$ are  calculated  as $\epsilon_1=1.975>0,\,\epsilon_2=1.0659>0,\,\epsilon_3= 0.1726>0,\,\epsilon_4=0.0036>0,$ and $\epsilon_1\epsilon_2\epsilon_3-\epsilon_1^2\epsilon_4 - \epsilon_3^2=0.319201502799312>0.$  Therefore, from Theorem \ref{existence_AND}, we conclude that $S_1^*,\,S_3^*$ is locally asymptotically stable and $S_2^*$ is unstable.  We draw the phase portrait for the considered parameters, shown in Fig. \ref{F3_AND} (c). We can see that all trajectories are converging to $S_1^*,\,S_3^*,$ and nearby trajectories of $S_2^*$ are moving away from $S_2^*.$ This shows the bistable behaviour of the system.\\
We chosen $\theta_b=17\in\theta_b^3,$
the coefficients of characteristic equation at $S_3^*=(924.286, \\0.0891386, 115.536,0.118851)$ are  calculated  as $\epsilon_1=1.975>0,\,\epsilon_2=1.1496>0,\,\epsilon_3= 0.1947>0,\,\epsilon_4=0.0046>0,$ and $\epsilon_1\epsilon_2\epsilon_3-\epsilon_1^2\epsilon_4 - \epsilon_3^2=0.3861>0.$ Therefore, from Theorem \ref{existence_AND}, we conclude that $S_3^*$ is locally asymptotically stable. We draw the phase portrait for the chosen parameters, shown in Fig. \ref{F3_AND} (d). We can see that all trajectories starting from any initial points are converging to $S_3^*.$ \\
	In this example, we also discuss the case of saddle-node bifurcation. We have seen that the system (\ref{and}) has two steady states $S_2^*,\,S_3^*$ in $\theta_b^2$ such that as the value of $\theta_b$ decreases, the two steady states collide at $\theta_b^{{{SN}^1}}=2.389134923929666$ and denoted as $S_{{SN}^1}^*=(45.7614,0.159759,5.72018,0.21301).$ The Jacobian matrix $J=Df(S_{{SN}^1}^*,\theta_b^{{{SN}^1}})$ has a simple eigenvalue $\lambda=0$ and  $v=(-0.9923,0.0015, -0.1240,0.002)^T,$\\$w=(-0.00033,0.1678,-0.0030,0.9858)^T$, are the eigenvectors of $J$ and $J^T,$ respectively. Both the tranversality conditions $w^Tf_{\theta_b}(S_{{SN}^1}^*,\theta_b^{{{SN}^1}})=-0.0026\neq 0$ and \\ $	w^T[D^2f(S_{{SN}^1}^*,\theta_b^{{{SN}^1}})(v,v)]=0.00000345\neq 0$ are satisfied. Hence, the system (1) experiences saddle-node bifurcation at $\theta_b=\theta_b^{{{SN}^1}}.$ Also, the system (\ref{and}) has two steady states $S_1^*,\,S_2^*$ in $\theta_b^2$ such that as the value of $\theta_b$ increases, the two steady states collide at $\theta_b^{{{SN}^2}}=15.667663989915999$ and denoted as $S_{{SN}^2}^*=(12.9963,0.300364,1.62454,0.400485 ).$ The Jacobian matrix $J=Df(S_{{SN}^2}^*,\theta_b^{{{SN}^2}})$ has a simple eigenvalue $\lambda=0$ and $v=(-0.9921,\\0.0123,-0.124,0.0164)^T,\,w=(-0.0627,0.1374,-0.5705,0.8073)^T$, are the eigenvectors of $J$ and $J^T,$ respectively. Both the tranversality conditions $w^Tf_{\theta_b}(S_{{SN}^2}^*,\theta_b^{{{SN}^2}})=-0.00062\\ \neq 0$ and  $	w^T[D^2f(S_{{SN}^2}^*,\theta_b^{{{SN}^2}})(v,v)]=-0.0040\neq 0$ are satisfied. Hence, the system (\ref{and}) experiences saddle-node bifurcation at $\theta_b=\theta_b^{{{SN}^2}}.$\\
The system (\ref{and}) exhibits a hysteresis effect in $\theta_b^2$ where multiple steady states coexist, as
shown in Fig \ref{F3_AND} (a). The two outer steady states are stable, while the interior steady state (red) is
unstable.\\
Again, in this case, the parameters are chosen as $\delta_a = 0.1,\,   \delta_b = 0.03,\, \gamma_a = 0.99,\, \gamma_b = 0.005,\, k_a = 0.525,\, k_b = 0.9,\, m_b = 0.01,\,  m_a =5,\, \theta_a = 0.4,\, A_1 =0.25,\, B_1 =0.00285,\, \theta_{aa}=20.7,\,\theta_b = (0.001,1.2).$ For this set of parameters, we obtain a unique stable steady state as discussed in Case 2 (AND). The corresponding bifurcation plot is shown in Fig. \ref{F3_AND} (e). To ensure the stability of $S_1^*,$ we choose $\theta_b=0.9$ from the range of $\theta_b.$  The coefficients of characteristic equation at $S_1^*$ are  calculated as $\epsilon_1=1.1250>0,\,\epsilon_2=0.1243>0,\,\epsilon_3=0.0032>0,\,\epsilon_4=0.00001284>0,$ and $\epsilon_1\epsilon_2\epsilon_3-\epsilon_1^2\epsilon_4 - \epsilon_3^2=0.0004171>0.$ Therefore, from Theorem \ref{existence_AND}, we conclude that $S_1^*$ is locally asymptotically stable. We draw the phase portrait for the considered $\theta_b$, shown in Fig. \ref{F3_AND} (f). We can see that all trajectories starting from any initial points are converging to $S_1^*=(0.273443,1.00583,1.43558,30.1749).$
\begin{figure}[H]
	\subfigure[]{\includegraphics[scale=0.27]{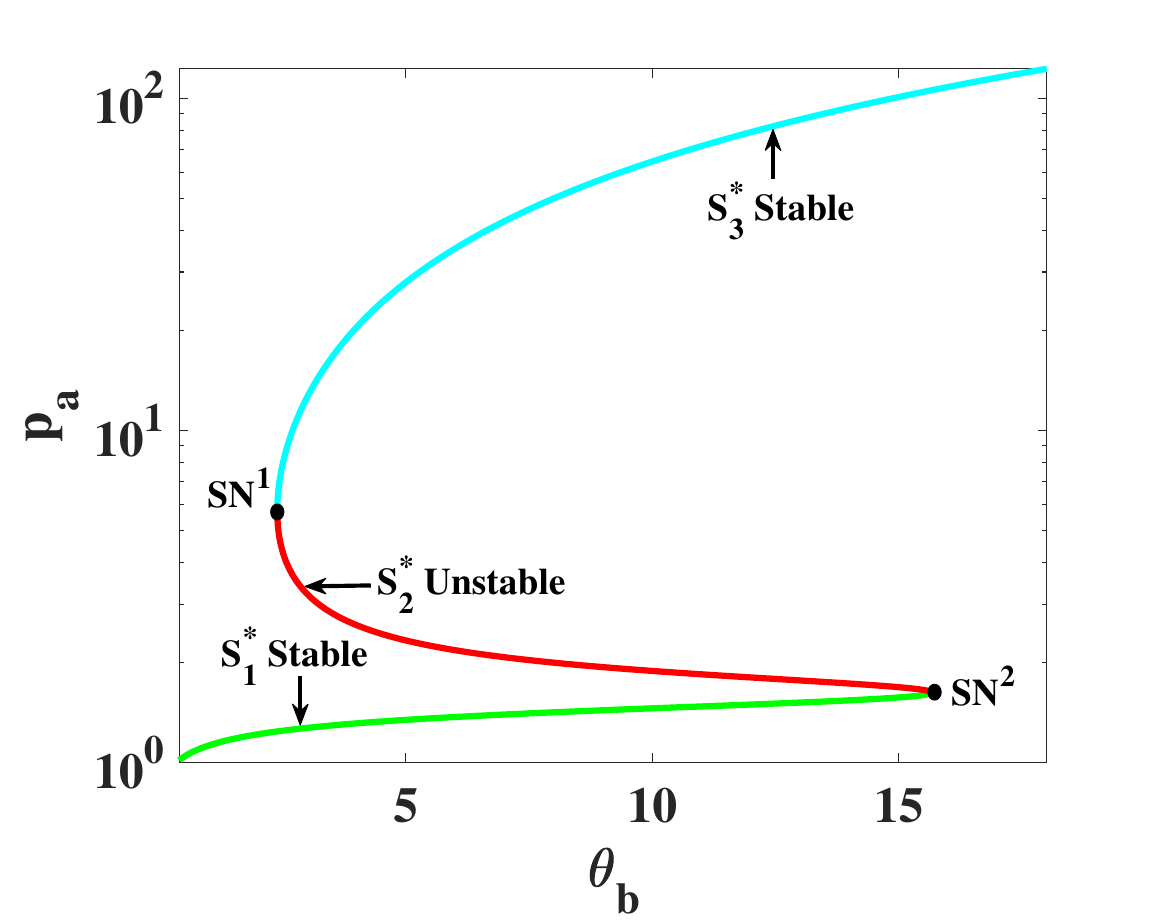}}
	\hfill
	\subfigure[]{\includegraphics[scale=0.27]{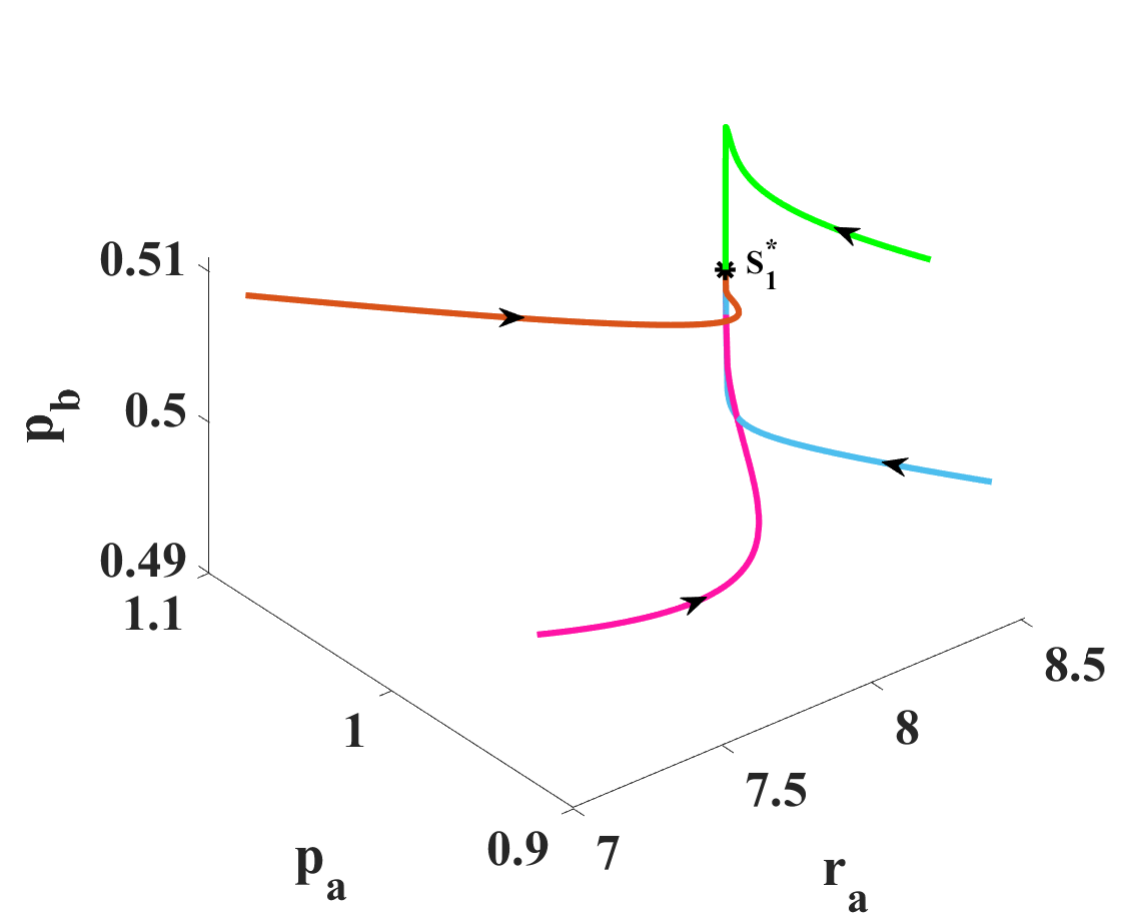}}
	\hfill
	\subfigure[]{\includegraphics[scale=0.27]{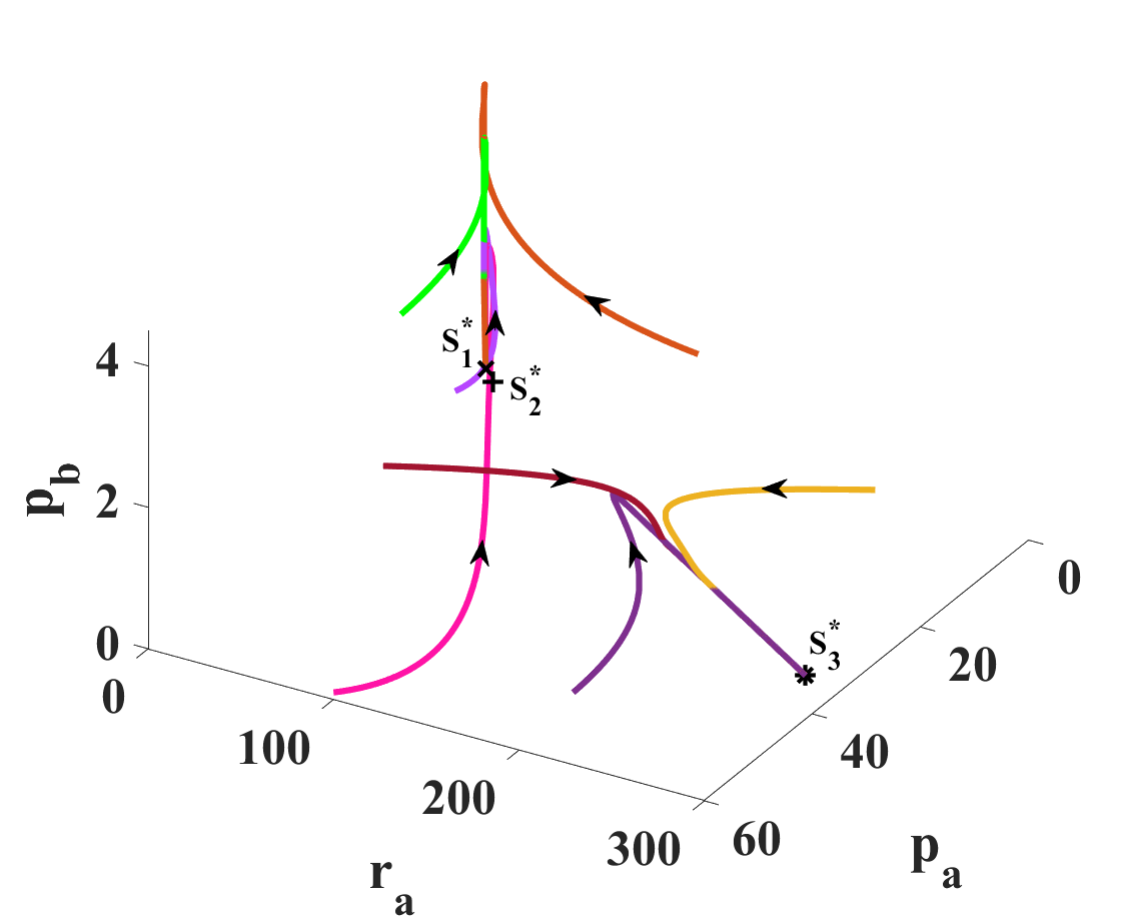}}
\end{figure}
\begin{figure}[H]
	\subfigure[]{\includegraphics[scale=0.27]{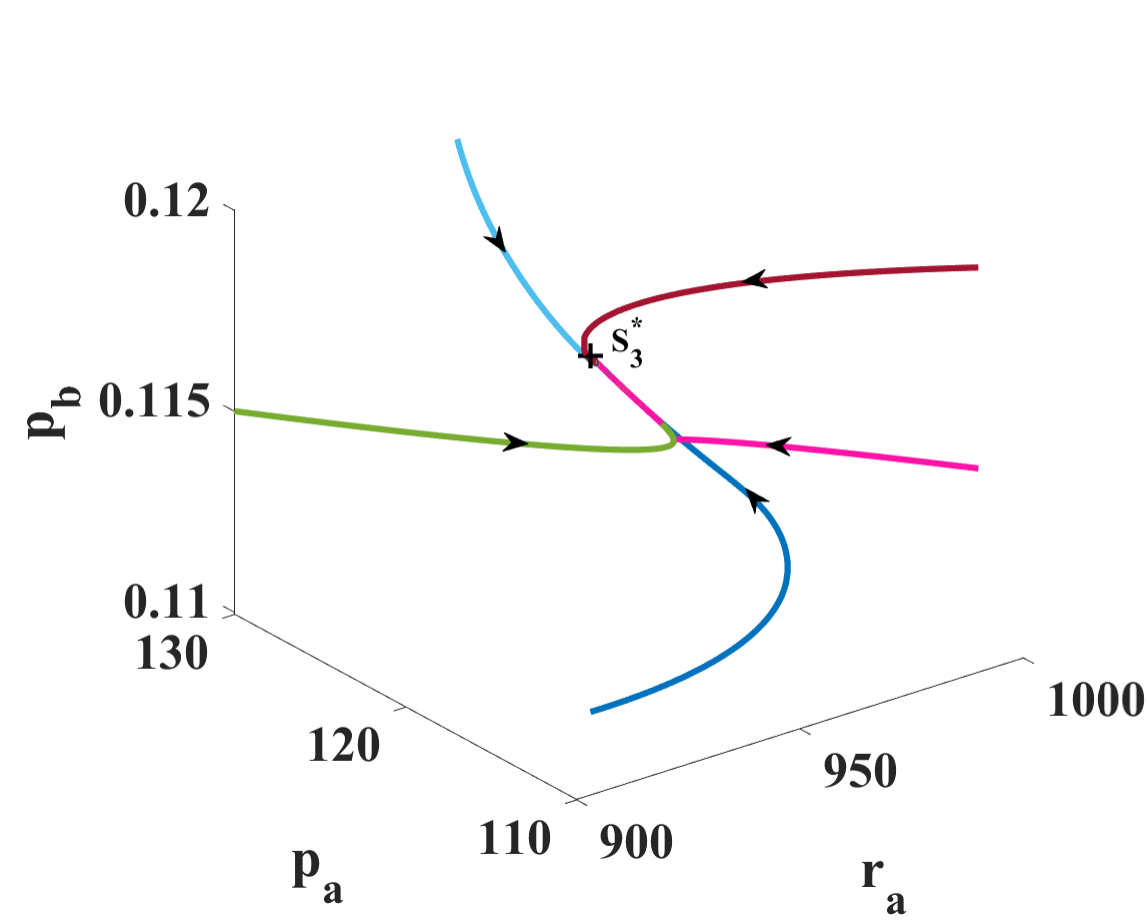}}
	\hfill
 \subfigure[]{\includegraphics[scale=0.26]{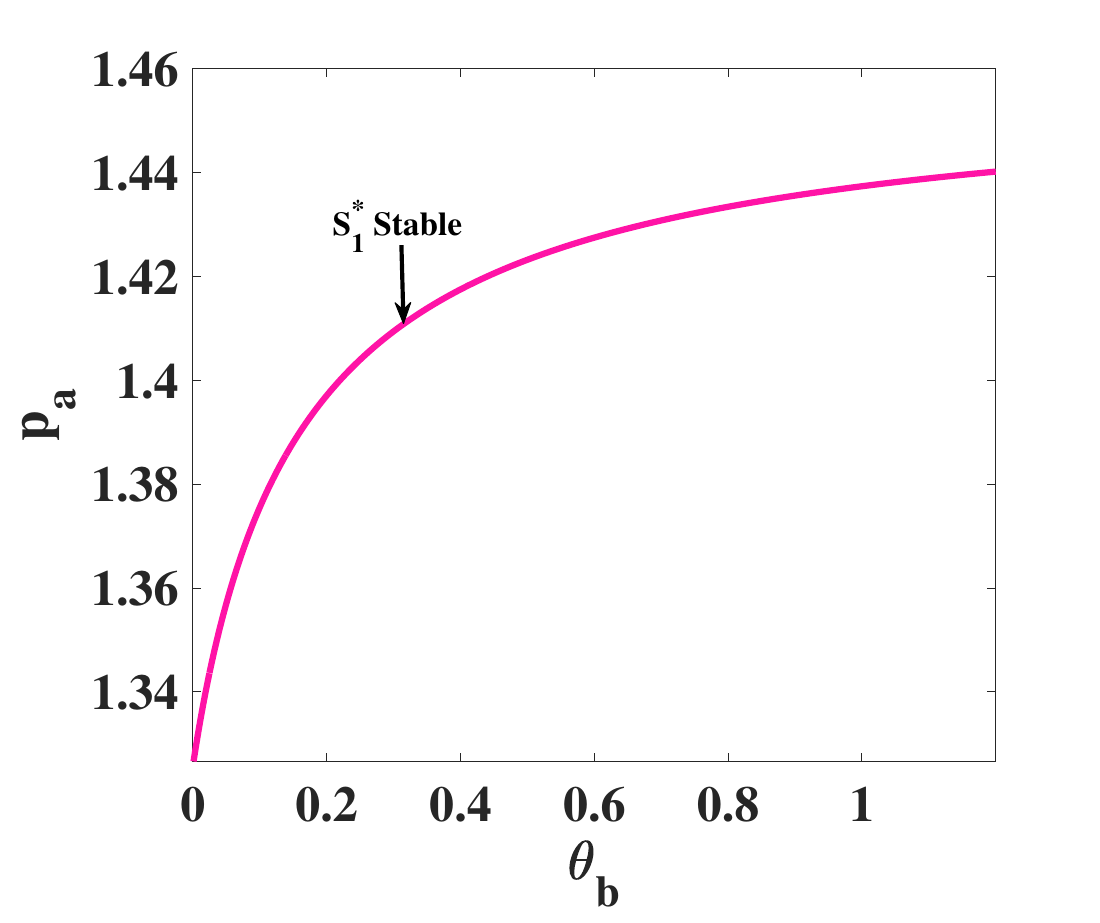}}
	\hfill
	\subfigure[]{\includegraphics[scale=0.26]{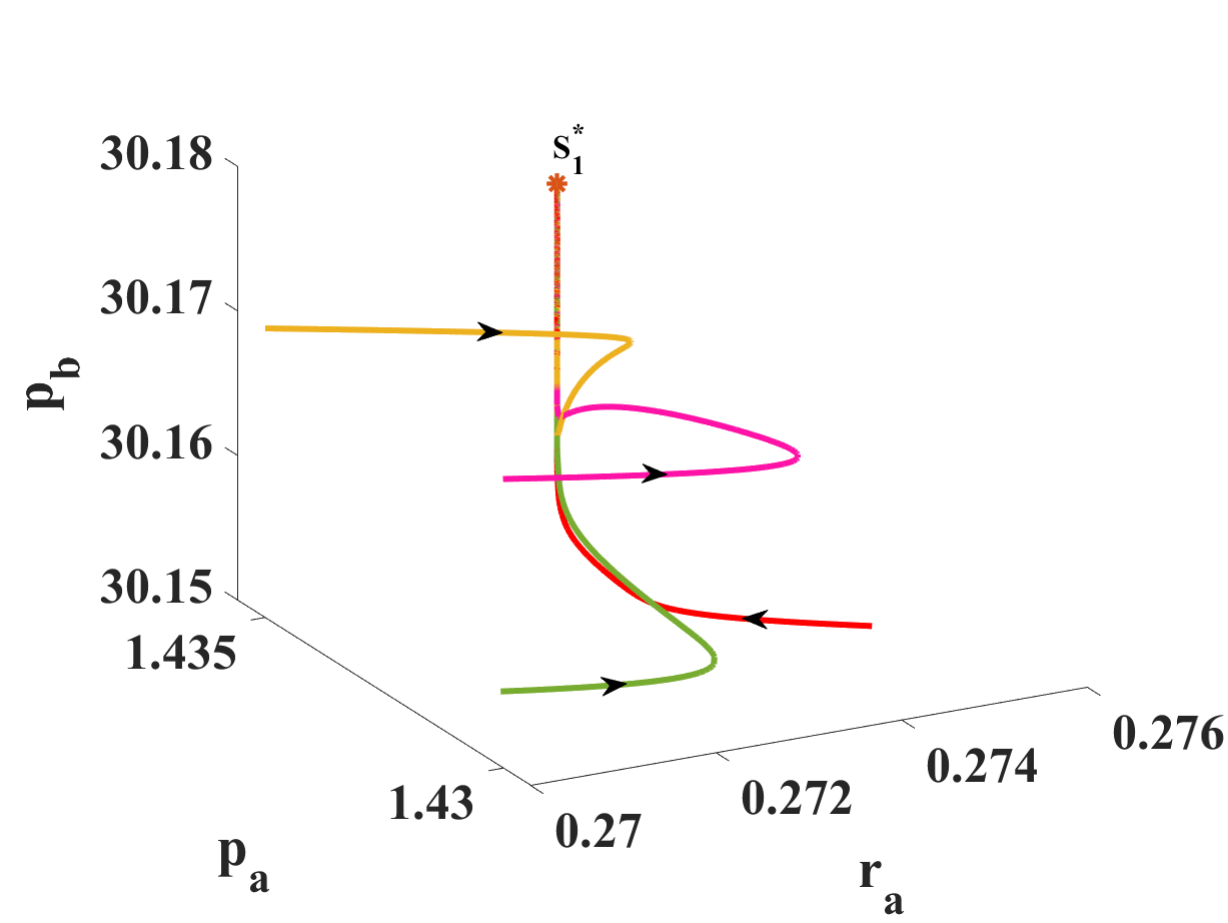}}
	\hfill
	\caption{For $n_a=1,\,n_b=1,\,n_{aa}=2$; (a) Bifurcation plot exhibiting hysteresis effect with three steady states (b) Phase portrait at $\theta_b=0.6$ (c) Phase portrait at $\theta_b=6$ depicting bistablility (d)  Phase portrait at $\theta_b=17$ (e) Bifurcation plot with the unique stable steady state (f) Phase portrait at $\theta_b=0.9.$}
	\label{F3_AND}
\end{figure}
\end{example}

\subsection{Conclusion and biological interpretation for this section}
In this section, we assume that the two genes $a$ and $b$ don't compete for binding to the promoter regions. However, only both binding events can elicit gene expression. This type of biological regulation is modeled using a non-competitive AND logic function. Our results show that, unlike competitive binding with OR logic, AND logic can't produce bistability with monomeric regulations (Hill coefficients of degree 1). We show that bistability requires dimeric (Hill coefficients of degree 2) autoloop. We also show that only through "forced" cell state transition (hysteresis) and not the smooth state swap (biphasic transition) is possible with this logic function for all combinations of dimerizations. \\

\section{Mathematical Model [Type-II $\mathcal{OR}$ Logic]}
 In this section, we model the co-regulation of node $a$ using non-competitive OR logic \cite{logics}. With this formalism, the model system takes the following form.\\
\begin{equation}\label{1_OR2}
	\begin{aligned}
		&\frac{dr_a}{dt}= m_a\biggl(\frac{(\frac{p_a}{\theta_{aa}})^{n_{aa}}}{1+(\frac{p_a}{\theta_{aa}})^{n_{aa}}}+\frac{1}{1+(\frac{p_b}{\theta_b})^{n_{b}}}\biggr)-\gamma_a r_a+A_1, \\
		&\frac{dr_b}{dt}  =\frac{m_b}{1+(\frac{p_a}{\theta_a})^{n_{a}}}-\gamma_b r_b+B_1, \\
		&\frac{dp_a}{dt} =k_ar_a-\delta_ap_a,\\
		&\frac{dp_b}{dt} =k_br_b-\delta_bp_b,
	\end{aligned}
\end{equation}
with $r_a(0)\geq 0,\,r_b(0)\geq 0,\,p_a(0)\geq 0,\,p_b\geq 0.$ Parameter description is same as Section \ref{MM_OR}.
\subsection{Mathematical Analysis}\label{4.1}
We get the steady state $S^*=(r_a^*,r_b^*,p_a^*,p_b^*)$ by setting the right hand side of the model (\ref{1_OR2}) to zero.
From third and fourth equations of model (\ref{1_OR2}), we obtained ${r_a}^*=\frac{\delta_a {p_a}^*}{k_a}$ and ${r_b}^*=\frac{\delta_b {p_b}^*}{k_b},$ respectively, as $\frac{dp_a}{dt}=0$ and $\frac{dp_b}{dt}=0.$ Using  $\frac{dr_b}{dt}=0$ and then putting the value of $r_b,$ we get ${p_b}^*=\frac{1}{\gamma_b}\bigg\{B_1'+\frac{m_2}{\theta_a^{n_a}+{{p_a}^*}^{n_a}}\bigg\}.$ Finally, using the value of $p_b$ and $r_a$ in $\frac{dr_a}{dt}=0,$ we get a polynomial in $p_a$ as
\begin{equation}\label{2_OR2}
	\begin{aligned}	\bigg(m_2+&B_1'\Big(\theta_a^{n_a}+p_a^{n_a}\Big)\bigg)^{n_b}\biggl\{-\gamma_a p_{a}^{n_{aa+1}}+p_{a}^{n_{aa}}\big(m_3+A_1'\big)-p_a\gamma_a\theta_{aa}^{n_{aa}}+A_1'\theta_{aa}^{n_{aa}}\biggr\}+\theta_{b}^{n_{b}}\gamma_{b}^{n_{b}}\\&(\theta_{a}^{n_{a}}+p_{a}^{n_{a}})^{n_b}\biggl\{-\gamma_a p_{a}^{n_{aa+1}}+p_{a}^{n_{aa}}\big(2m_3+A_1'\big)-p_a\gamma_a\theta_{aa}^{n_{aa}}+(m_3+A_1')\theta_{aa}^{n_{aa}}\biggr\}=0,
	\end{aligned}
\end{equation}
where $A_1'=\frac{A_1 k_a}{\delta_a},\,B_1'=\frac{B_1 k_b}{\delta_b},\,m_2=\frac{m_b k_b \theta_a^{n_a}}{\delta_b},$ and $m_3=\frac{m_a k_a}{\delta_a}.$ Here, $p_a^*$ is a positive real root of (\ref{2_OR2}) which we discuss in different cases explicitly. \\
\textbf{Case 1:} When $n_a=n_b=n_{aa}=1,$ then from (\ref{2_OR2}), we get a cubic polynomial in $p_a$
\begin{equation}\label{3OR2}
	f(p_a)=Ap_a^3+Bp_a^2+Cp_a+D,
\end{equation}
where $A=-\gamma_a (B_1'+\gamma_b\theta_b),\,B=-\gamma_a (B_1'\theta_a+m_2)+B_1'(m_3+A_1'-\gamma_a\theta_{aa})-\gamma_a\gamma_b\theta_a\theta_b+\theta_b\gamma_b(2m_3+A_1'-\gamma_a\theta_{aa}),\,
C=A_1'B_1'\theta_{aa}+(B_1'\theta_a+m_2)(m_3+A_1'-\gamma_a\theta_{aa})+\theta_b\gamma_b\theta_{aa}(m_3+A_1')+\theta_a\gamma_b\theta_b(2m_3+A_1'-\gamma_a\theta_{aa})$ and $D=\theta_{aa}\bigl\{(B_1'\theta_a+m_2)A_1'+\theta_a\gamma_b\theta_b(m_3+A_1')\bigr\}.$
The nature of the curve $f$ can be observe from (\ref{3OR2}) as 
\begin{enumerate}
	\item[(i)] $f(0)>0$ because $D>0,$
	\item[(ii)] $f\to -\infty\, (\infty)$ as $p_a\to \infty\,(-\infty)$ because $A<0,$
	
	\item[(iii)] The number of positive real roots of equation (\ref{3OR2}) can be seen using Descarte's rule of sign, mentioned in Table (\ref{T1and}).
\end{enumerate}
Therefore, $f(p_a)=0$ will have either unique or three positive real roots. Now, to discard the existence of three positive real roots, we use the nullcline method. From the system (\ref{1_OR2}), we have
\begin{equation}\label{6_OR2}
	\begin{aligned}
		&F_1(p_a,p_b):=\frac{m_{a} k_a\theta_bp_a}{\delta_a(\theta_{aa}+p_a)(\theta_{b}+p_b)}+A_1'-\gamma_{a}p_a \\
		&F_2(p_a,p_b):=\frac{m_{b}k_b\theta_a}{\delta_b(\theta_{a}+p_a)}+B_{1}'-\gamma_{b}p_b,
	\end{aligned}
\end{equation}
First to observe the nature of curve $F_2(p_a,p_b)=0,$ we reduce it in the form of
\begin{equation}\label{7or2}
	p_b=g_1(p_a):=\frac{1}{k_2}\Big\{\frac{m_{b}\theta_a}{(\theta_{a}+p_a)}+B_{1}\Big\}\tag{$\mathscr{A}_2$}.
\end{equation}
The following conclusion can be made from (\ref{7or2})
\begin{enumerate}
	\item[$\mathscr{A}_2$(i)] $g_1(0)=\frac{1}{k_2}(m_b +B_{1})>\frac{B_1}{k_2}>0$ 
	\item[$\mathscr{A}_2$(ii)] $g_1(p_a)\to \frac{B_1}{k_2}$ as $p_a\to \pm\infty$ 
	\item[$\mathscr{A}_2$(iii)] $g'_1(p_a)<0$, i.e, $g_1(p_a) $ is a decreasing function of $p_a.$
\end{enumerate}
The graph of $F_2(p_a,p_b)=0$ (or $p_b=g_1(p_a)$) is shown in Fig. \ref{fig:g1and}.\\
Now, to observe the nature of curve $F_1(p_a,p_b)=0,$ we reduce it in the form of
\begin{equation}\label{8_OR2}
p_b=g_2(p_a):=\frac{-\theta_{b}\left(m_{a}\left(2p_a+\theta_{aa}\right)+\left(A_1'-k_1p_a\right)\left(\theta_{aa}+p_a\right)\right)}{m_{1}p_a+\left(A_1'-k_1p_a\right)\left(\theta_{aa}+p_a\right)}\tag{$\mathscr{D}$}
\end{equation}
The following conclusion can be made from (\ref{8_OR2})
\begin{enumerate}
		\item[$\mathscr{D}$(i)]The $p_a-$intercept of $g_2(p_a)$ are $x_1=\frac{(2m_a+A_1-k_1\theta_{aa})+\sqrt{(2m_a+A_1-k_1\theta_{aa})^2+4k_1A_1\theta_{aa}}}{2k_1}$ and $x_2=\frac{(2m_a+A_1-k_1\theta_{aa})-\sqrt{(2m_a+A_1-k_1\theta_{aa})^2+4k_1A_1\theta_{aa}}}{2k_1}$. We note that $x_1>0$ and $x_2<0,$ which means we have one positive and one negative $p_a-$intercept. Also  $p_b-$intercept is negative, i.e. $g_2(0)<0.$
		\item[$\mathscr{D}$(ii)]  $g_2(p_a)$ is a rational function, and it has two vertical asymptotes at\\ $x_3=\frac{(m_a+A_1-k_1\theta_{aa})+\sqrt{(m_a+A_1-k_1\theta_{aa})^+4k_1A_1\theta_{aa}}}{2k_1}$ and \\$x_4=\frac{(m_a+A_1-k_1\theta_{aa})-\sqrt{(m_a+A_1-k_1\theta_{aa})^2+4k_1A_1\theta_{aa}}}{2k_1}$. \\We note that $x_3>0$ and $x_4<0,$ which means we have one positive and one negative vertical asymptotes. Moreover, $p_b=-\theta_b$ is a horizontal asymptote. 
			\end{enumerate}
We note that $x_3<x_1.$ For $p_a\in (0,x_3),$ we have $g_2(p_a)<0$ and for $p_a\in (x_3,x_1),$ we have $g_2(p_a)>0.$ Also, $g_2'(p_a)<0$ which implies that $g_2(p_a)$ is strictly decreasing. The corresponding diagram is shown in 	Fig. \ref{OR2_NULL} (a). Therefore, from the above discussion regarding the nature of $g_1$ and $g_2$, we conclude that there exists a unique steady state, shown in Fig. \ref{OR2_NULL} (b), and we discard the case of three positive steady states.
\begin{figure}[H]
\centering
	\subfigure[]{\includegraphics[scale=0.35]{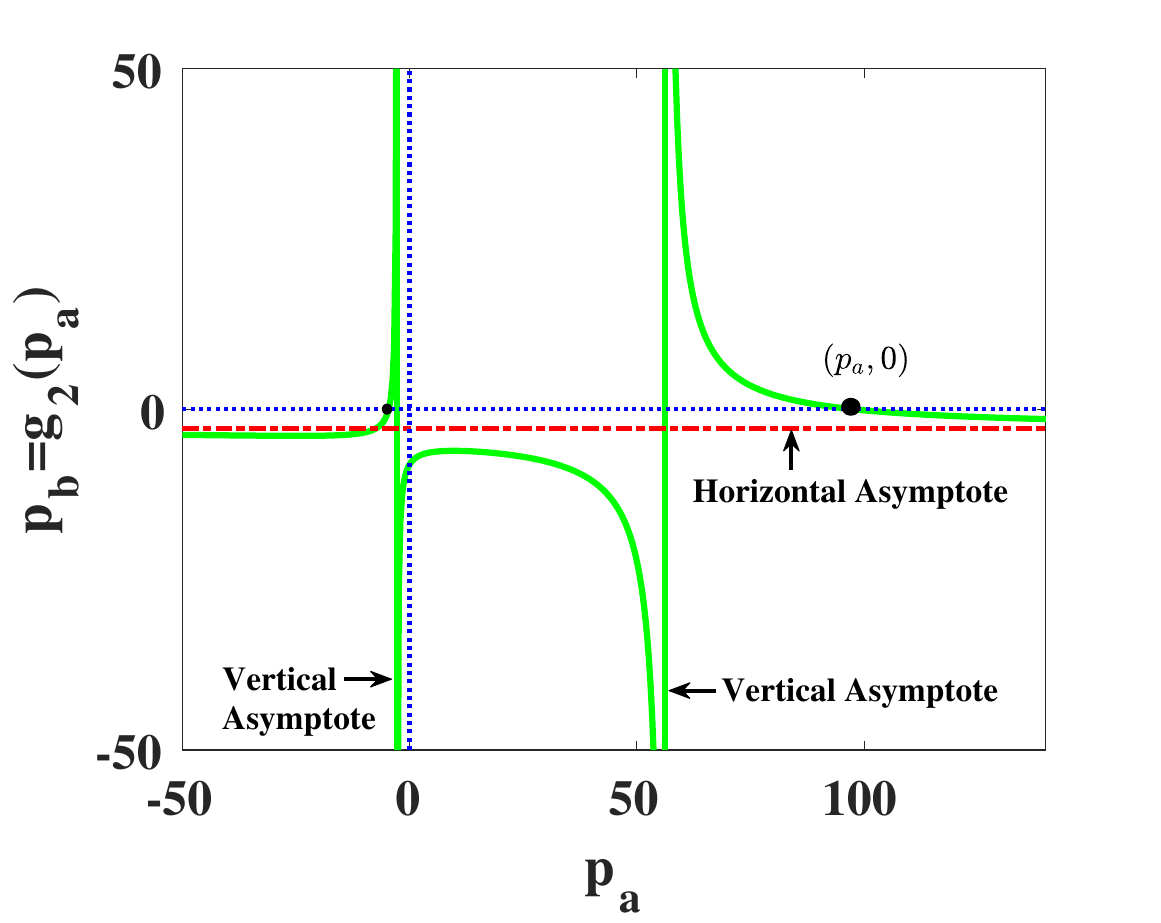}}
	\hfil
	\subfigure[]{\includegraphics[scale=0.35]{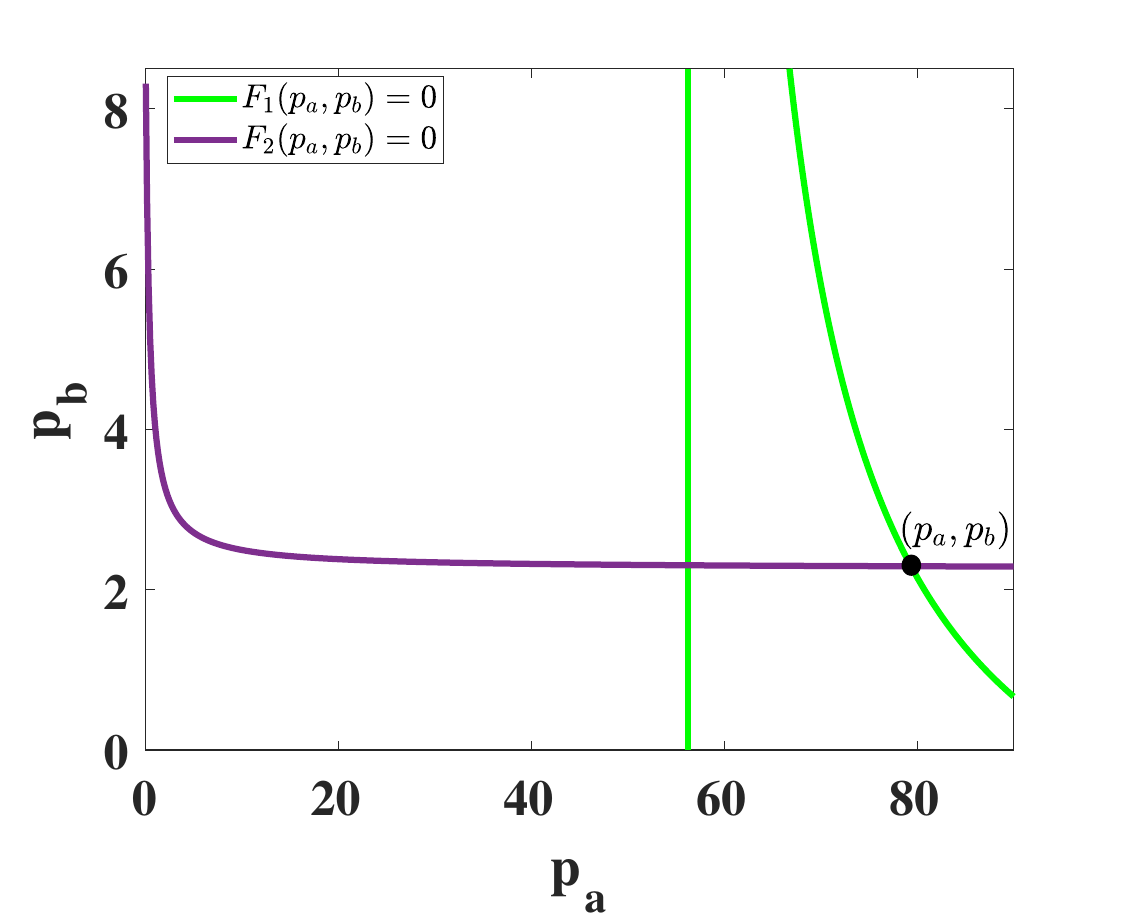}}
	\caption{For $n_a = n_b =  n_{aa} = 1;$ (a)
The graph depicting unique positive $p_a$-intercept of $g_2(pa)$ (b) Existence of unique positive steady state corresponding to Fig. \ref{OR2_NULL} (a). }
	\label{OR2_NULL}
\end{figure}
\textbf{Case 2:} When $n_a=n_b=1,\,n_{aa}=2,$ then from (\ref{2_OR2}), we get a quartic polynomial in $p_a$
\begin{equation}\label{4or2}
	f(p_a)=Ap_a^4+Bp_a^3+Cp_a^2+Dp_a+E,
\end{equation}
where $A=-\gamma_a (B_1'+\gamma_b\theta_b),\,B=-\gamma_a (B_1'\theta_a+m_2)+B_1'(m_3+A_1')-\gamma_a\gamma_b\theta_a\theta_b+\theta_b\gamma_b(2m_3+A_1'),\,
C=(B_1'\theta_{a}+m_2)(m_3+A_1')+\theta_b\gamma_b\theta_{a}(2m_3+A_1')-\gamma_a\theta_{aa}^2(B_1'+\gamma_b\theta_b),\,D=\theta_{aa}\bigl\{B_1'A_1'+\gamma_b\theta_b(m_3+A_1')\bigr\}-\gamma_a\theta_{aa}^2(B_1'\theta_a+m_2+\gamma_b\theta_b\theta_a)$ and $E={\theta _{aa}}^2\bigl\{(B_1'\theta_a+m_2)A_1'+\gamma_b\theta_b\theta_a(m_3+A_1')\bigr\}.$\\
The number of positive real roots of equation (\ref{4or2}) can be seen using Descarte's rule of sign (mentioned in Table (\ref{T2})); therefore, $f(p_a)=0$ will have either unique or three positive real roots.\\
\textbf{Case 3:} When $n_a=2,\,n_b=3,\,n_{aa}=2,$ then from (\ref{2_OR2}), we get a ninth degree polynomial in $p_a$
\begin{equation}\label{5or2}
f(p_a)=Ap_a^9+Bp_a^8+Cp_a^7+Dp_a^6+Ep_a^5+Fp_a^4+Gp_a^3+Hp_a^2+Ip_a+J,
\end{equation}
where $A=-\gamma _{a}\,\left({B_1'}^3+{\gamma _{b}}^{n_{b}}\,{\theta _{b}}^{n_{b}}\right),\,B={B_1'}^3\,\left(A_1'+m_{3}\right)+{\gamma _{b}}^{n_{b}}\,{\theta _{b}}^{n_{b}}\,\left(A_1'+2\,m_{3}\right),$\\
$C=-\gamma _{a}\,\left({B_1'}^2\,\left(3\,m_{2}+3\,B_1'\,{\theta _{a}}^{n_{a}}+B_1'\,{\theta _{aa}}^{n_{aa}}\right)+{\gamma _{b}}^{n_{b}}\,{\theta _{b}}^{n_{b}}\,\left(3\,{\theta _{a}}^{n_{a}}+{\theta _{aa}}^{n_{aa}}\right)\right),$\\$D=    {B_1'}^2\,\left(\left(A_1'+m_{3}\right)\,\left(3\,m_{2}+3\,B_1'\,{\theta _{a}}^{n_{a}}\right)+A_1'\,B_1'\,{\theta _{aa}}^{n_{aa}}\right)+{\gamma _{b}}^{n_{b}}\,{\theta _{b}}^{n_{b}}\\\left(3\,{\theta _{a}}^{n_{a}}\,\left(A_1'+2\,m_{3}\right)+{\theta _{aa}}^{n_{aa}}\,\left(A_1'+m_{3}\right)\right),$
\\$E=-3\,B_1'\,\gamma _{a}\,\left(m_{2}+B_1'\,{\theta _{a}}^{n_{a}}\right)\,\left(m_{2}+B_1'\,{\theta _{a}}^{n_{a}}+B_1'\,{\theta _{aa}}^{n_{aa}}\right)-3\,\gamma _{a}\,{\gamma _{b}}^{n_{b}}\,{\theta _{a}}^{n_{a}}\,{\theta _{b}}^{n_{b}}\,\left({\theta _{a}}^{n_{a}}+{\theta _{aa}}^{n_{aa}}\right),$
\\$F=3\,B_1'\,\left(\left(A_1'+m_{3}\right)\,\left(m_{2}+B_1'\,{\theta _{a}}^{n_{a}}\right)+A_1'\,B_1'\,{\theta _{aa}}^{n_{aa}}\right)\,\left(m_{2}+B_1'\,{\theta _{a}}^{n_{a}}\right)+3\,{\gamma _{b}}^{n_{b}}\,{\theta _{a}}^{n_{a}}\,{\theta _{b}}^{n_{b}}\\\left({\theta _{a}}^{n_{a}}\,\left(A_1'+2\,m_{3}\right)+{\theta _{aa}}^{n_{aa}}\,\left(A_1'+m_{3}\right)\right),$
\\$G=-\gamma _{a}\,{\left(m_{2}+B_1'\,{\theta _{a}}^{n_{a}}\right)}^2\,\left(m_{2}+B_1'\,{\theta _{a}}^{n_{a}}+3\,B_1'\,{\theta _{aa}}^{n_{aa}}\right)-\gamma _{a}\,{\gamma _{b}}^{n_{b}}\,{\theta _{a}}^{2\,n_{a}}\,{\theta _{b}}^{n_{b}}\,\left({\theta _{a}}^{n_{a}}+3\,{\theta _{aa}}^{n_{aa}}\right),$
\\$H=\left(\left(A_1'+m_{3}\right)\,\left(m_{2}+B_1'\,{\theta _{a}}^{n_{a}}\right)+3\,A_1'\,B_1'\,{\theta _{aa}}^{n_{aa}}\right)\,{\left(m_{2}+B_1'\,{\theta _{a}}^{n_{a}}\right)}^2+{\gamma _{b}}^{n_{b}}\,{\theta _{a}}^{2\,n_{a}}\,{\theta _{b}}^{n_{b}}\\\left({\theta _{a}}^{n_{a}}\,\left(A_1'+2\,m_{3}\right)+3\,{\theta _{aa}}^{n_{aa}}\,\left(A_1'+m_{3}\right)\right),$
\\$I=-\gamma _{a}\,{\theta _{aa}}^{n_{aa}}\,\left({\left(m_{2}+B_1'\,{\theta _{a}}^{n_{a}}\right)}^3+{\gamma _{b}}^{n_{b}}\,{\theta _{a}}^{3\,n_{a}}\,{\theta _{b}}^{n_{b}}\right),$ and\\
$J={\theta _{aa}}^{n_{aa}}\,\left(A_1'\,{\left(m_{2}+B_1'\,{\theta _{a}}^{n_{a}}\right)}^3+{\gamma _{b}}^{n_{b}}\,{\theta _{a}}^{3\,n_{a}}\,{\theta _{b}}^{n_{b}}\,\left(A_1'+m_{3}\right)\right)
.$\\
The number of positive real roots of equation (\ref{5or2}) can be obtained using Descarte's rule of sign. Due to complexity of polynomial, we discuss the existence  of steady states via numerical simulation.
\subsubsection{Local stability of steady state}
\begin{theorem}\label{existence_OR2}
	The steady state $S_1^{*}$ is locally asymptotically stable if $\epsilon_1>0,\,\epsilon_3>0,\,\epsilon_4>0,$ and $\epsilon_1\epsilon_2\epsilon_3-\epsilon_1^2\epsilon_4 - \epsilon_3^2 >0,$ where $\epsilon_i,\,i=1,..,4$ are provided in the proof.
\end{theorem}
\begin{proof}At $S_1^{*}, $ the Jacobian matrix of the model system (\ref{1_OR2}) is given by
	\begin{center}
		$J_{S_1^*}=
		\left[\begin{array}{cccc} -\gamma_a & 0 & X&Y\\ 0 & -\gamma_b & Z&0\\ k_a &0  & -\delta_a&0\\0&-k_b&0&-\delta_b \end{array}\right].$
	\end{center}
	where $X=\frac{\mathrm{m}_a\,n_{aa}\,{\theta _{aa}}^{n_{aa}}\,p_a^{n_{aa-1}}}{{\left({\theta _{aa}}^{n_{aa}}+p_a^{n_{aa}}\right)}^2}
	,\,Y=-\frac{\mathrm{m}_a\,n_{b}\,{\theta _{b}}^{n_{b}}\,p_b^{n_{b}-1}}{{\left({\theta _{b}}^{n_{b}}+p_b^{n_{b}}\right)}^2},$ and $Z=-\frac{\mathrm{m}_b\,n_{a}\,{\mathrm{p_a}}^{n_{a}-1}\,{\theta _{a}}^{n_{a}}}{{\left({\mathrm{p_a}}^{n_{a}}+{\theta _{a}}^{n_{a}}\right)}^2}.$
	The characteristic equation of $J_{S_1^*}$ is given as
	\begin{equation}\label{9_OR2}
		\lambda^4+\epsilon_1\lambda^3+\epsilon_2\lambda^2+\epsilon_3\lambda+\epsilon_4=0.
	\end{equation}
	Here, $\epsilon_1=\gamma_a+\gamma_b+\delta_a+\delta_b,\epsilon_2=\gamma_a\gamma_b+(\gamma_a+\gamma_b)(\delta_a+\delta_b)+\delta_a\delta_b-k_aX,\,\epsilon_3=\delta_a\delta_b(\gamma_a+\gamma_b)+(\delta_a+\delta_b)(\gamma_a\gamma_b)-k_a(\delta_b+\gamma_b)X,$ and  $\epsilon_4=\gamma_a\gamma_b\delta_a\delta_b-k_ak_bYZ-k_a\gamma_b\delta_bX.$
	By the Routh-Hurwitz criteria, all the roots of (\ref{9_OR2}) have negative real parts provided $\epsilon_1>0,\,\epsilon_3>0,\,\epsilon_4>0,$ and $\epsilon_1\epsilon_2\epsilon_3-\epsilon_1^2\epsilon_4 - \epsilon_3^2 >0.$
\end{proof}
\subsubsection{Saddle Node Bifurcation}
We find the transversality condition the same as discussed in Subsection \ref{sn_OR}.  Define $f=(f_1,f_2,f_3,f_4)^T,$ where $f_1$ are defined as
\begin{equation*}
	f_1=m_a\biggl(\frac{(\frac{p_b}{\theta_b})^{n_{b}}}{1+(\frac{p_b}{\theta_b})^{n_{b}}}+\frac{1}{1+(\frac{p_b}{\theta_b})^{n_{b}}}\biggr)-\gamma_a r_a+A_1, 
\end{equation*}
and the expression of $f_2 ,\,f_3,$ and $f_4$  are same as Subsection \ref{sn_OR}.
Then the model system (\ref{1_OR2})
experiences a saddle-node bifurcation  at $(S_{SN}^*,\theta_b^{SN})$, if the following transversality conditions \cite{perko_2013} are satisfied:
\begin{equation*}
	w^Tf_{\theta_b}(S_{SN}^*,\theta_b^{SN})=w_1\frac{m_a n_b {\theta_b^{n_b-1}}^{SN}{p_b^{n_b}}^{{SN}^*}}{(\mathcal{X}_2)^2}\neq0. \tag{$\mathcal{T}_{33}$}\end{equation*}
and  
\begin{align*}
	w^T[D^2f(S_{SN}^*,\theta_b^{SN})(v,v)]=&w_1\Biggl\{\frac{\mathcal{O}v_3^2}{\mathcal{X}_1^3}\bigg(\theta_{aa}^{n_{aa}}(n_{aa}-1){p_a^{(n_{aa}-2)}}^{{SN}^*}-(n_{aa}+1){p_a^{(2n_{aa}-2)}}^{{SN}^*}\bigg)\\&-\frac{v_4^2\mathcal{O}_1}{\mathcal{X}_2^3}\bigg({\theta_{b}^{n_{b}}}^{SN}(n_{b}-1){p_b^{(n_{b}-2)}}^{{SN}^*}-(n_{b}+1){p_b^{(2n_{b}-2)}}^{{SN}^*}\bigg)\Biggr\}\\&	-\frac{w_2v_3^2m_b n_a \theta_a^{n_a}}{(\theta_a^{n_a}+{p_a^{n_a}}^{{SN}^*})^3} \bigg\{{p_a^{n_a-2}}^{{SN}^*}\theta_a^{n_a}(n_a-1)-(n_a+1){p_a^{2n_a-2}}^{{SN}^*}\bigg\}\\&\neq 0,\tag{$\mathcal{T}_{44}$}
\end{align*}
where $\mathcal{X}_1=\theta_{aa}^{n_{aa}}+{p_a^{n_{aa}}}^{{SN}^*},\,\mathcal{X}_2=\theta_b^{n_b}+{p_b^{n_b}}^{{SN}^*},\,\mathcal{O}=n_{aa}m_a\theta_{aa}^{n_{aa}},$ and $\mathcal{O}_1=m_an_b\theta_b^{n_b}.$
\subsubsection{Numerical simulation and discussion corresponding to model system (\ref{1_OR2})}
The purpose of this section is to provide validation for the theoretical results that were covered in Subsection \ref{4.1}. We have talked about the system's behavior in terms of monostability, bistability, and tristability with fictitious parameters.
\begin{example}
	$n_a=n_b=n_{aa}=1$ (case 1) and the other parameters are chosen as $\delta_a = 0.09,\,\delta_b = 0.9,\, \gamma_a = 1,\, \gamma_b = 0.9,\, k_a = 0.25,\, k_b = 0.9,\, m_b = 0.1,\,  m_a =2,\,\theta_a = 0.4,\, A_1 =0.0035,\, B_1 =7.85,\, \theta_{aa} =20,\,  ,\,\theta_b = (0.00001,50).$ For this set of parameters, we obtain a unique stable steady state as discussed in Case 1. The corresponding bifurcation plot is shown in Fig. \ref{traj_OR2} (a). To ensure the stability of $S_1^*,$  we choose $\theta_b=10$ from the range of $\theta_b.$  The coefficients of characteristic equation at $S_1^*$ are  calculated as $\epsilon_1=2.89>0,\,\epsilon_2=2.8445>0,\,\epsilon_3=1.0133>0,\,\epsilon_4=0.0790>0,$ and $\epsilon_1\epsilon_2\epsilon_3-\epsilon_1^2\epsilon_4 - \epsilon_3^2=0.0587>0.$ Therefore, from Theorem \ref{existence_OR2}, we conclude that $S_1^*$ is locally asymptotically stable. We draw the phase portrait for the considered $\theta_b$, shown in Fig. \ref{traj_OR2} (b). We can see that all trajectories starting from any initial points are converging to $S_1^*=(1.39595,8.73261,
	3.87764,8.73261 ).$
	\begin{figure}[H]
 \centering
		\subfigure[]{\includegraphics[scale=0.3]{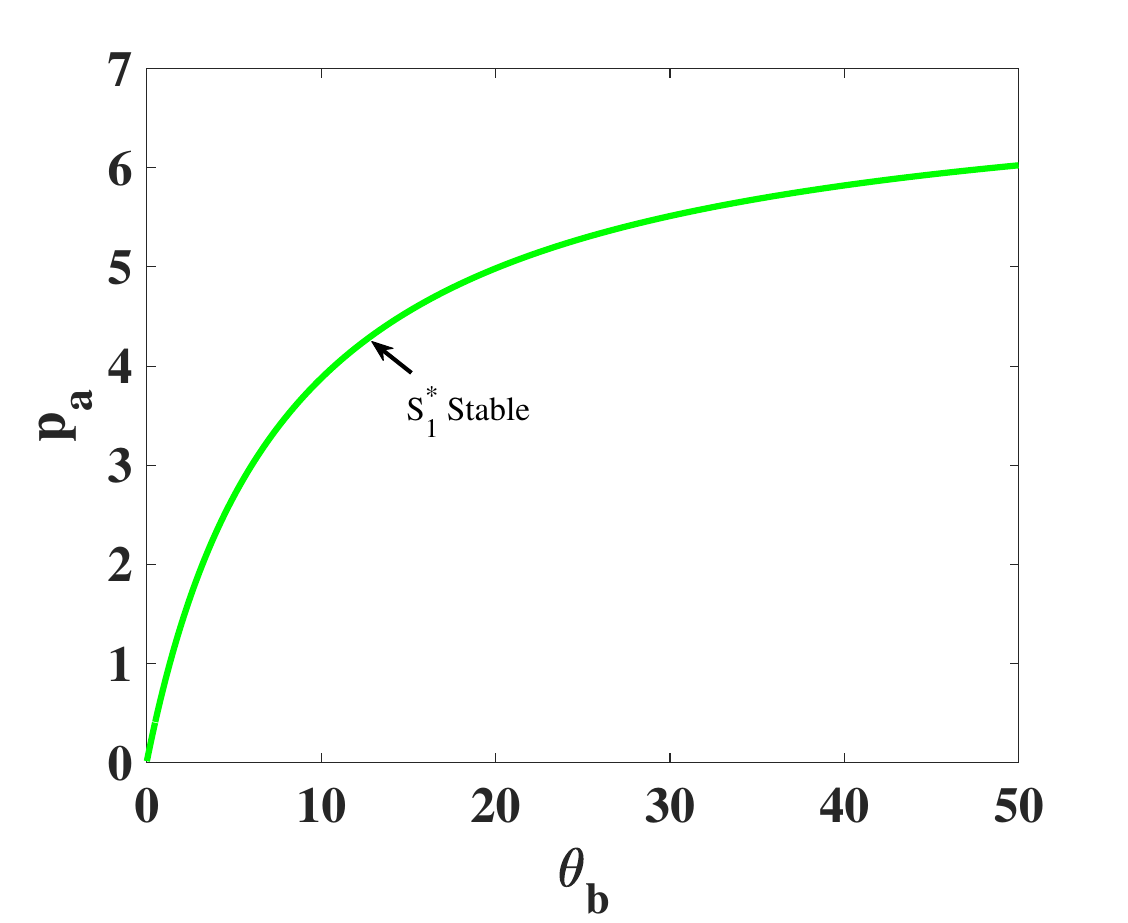}}
		\hfil
		\subfigure[]{\includegraphics[scale=0.3]{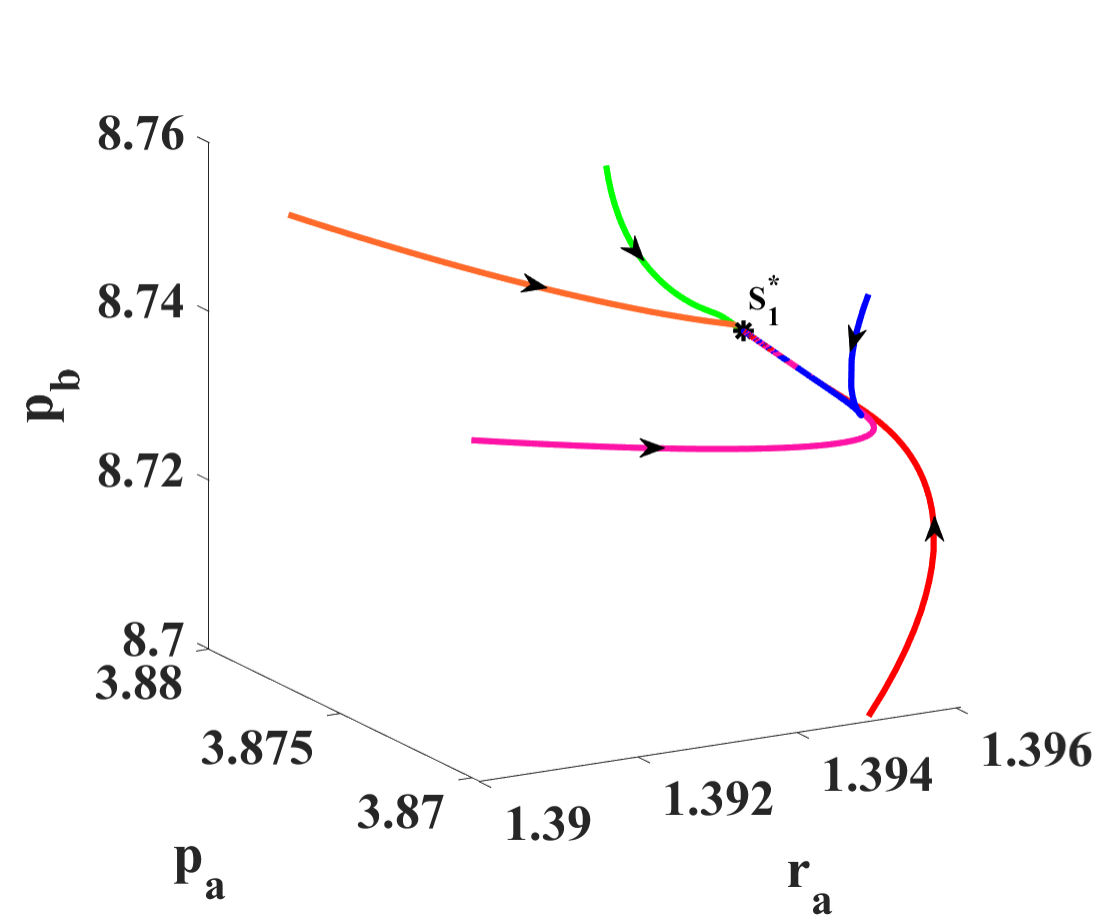}}
		\caption{For $n_a=1,\,n_b=1,\,n_{aa}=1$; (a) Bifurcation plot with the unique stable steady state (b) Phase portrait at $\theta_b=10.$}
		\label{traj_OR2}
	\end{figure}
\end{example}
	\begin{example}$n_a=n_b=1,\,n_{aa}=2$ (case 2), and the other parameters are chosen as: 
		$\delta_a = 0.363,\,\delta_b = 0.35,\, \gamma_a = 0.84,\, \gamma_b = 0.9,\, k_a = 2,\, k_b = 6.3,\, m_b = 5.7,\, m_a = 3.4,\, \theta_a= 6.2,\, A_1 = 0.1,\, B_1 = 0.01,\, \theta_{aa} = 13.1,\,  
		\theta_b = (1.6,6.5).$ In this range of $\theta_b,$ there is unique ($S_1^*$) steady state in $\theta_b^1$, three ($S_1^*,\,S_2^*,\,S_3^*$) in $\theta_b^2$ and unique ($S_3^*$) in $\theta_b^3,$ such that $\theta_b=\theta_b^1\cup\theta_b^2\cup\theta_b^2$ where $\theta_b^1=(1.6,{\theta_b^{SN}}^1),\,\theta_b^2=[{\theta_b^{SN}}^1,{\theta_b^{SN}}^2]$ and $\theta_b^2=({\theta_b^{SN}}^2,6.5).$ The steady states $S_1^*,\,S_3^*$ are stable and ($S_2^*$) is unstable. The correspond bifurcation diagram is shown in Fig. \ref{F3_OR2} (a).\\	
		To ensure the stability of $S_1^*,$ we chosen $\theta_b=1.89\in\theta_b^1$ and the coefficients of characteristic equation at $S_1^*=(0.239206,5.23418,1.31794,94.2152)$ are  calculated  as $\epsilon_1=2.453>0,\,\epsilon_2=2.0213>0,\,\epsilon_3=0.6321>0,\,\epsilon_4=0.0583>0,$ and $\epsilon_1\epsilon_2\epsilon_3-\epsilon_1^2\epsilon_4 - \epsilon_3^2=2.3838>0.$ Therefore, from Theorem \ref{existence_OR2}, we conclude that $S_1^*$ is locally asymptotically stable. We draw the phase portrait for the considered parameters, shown in Fig. \ref{F3_OR2} (b). We can see that all trajectories starting from any initial points are converging to $S_1^*.$\\
		Further, we chosen $\theta_b=3.81\in\theta_b^2$ and the coefficients of characteristic equation at $S_1^*=(0.416123,4.6347,2.29269,
		83.4245)$ are  calculated  as $\epsilon_1=2.453>0,\,\epsilon_2=1.9526>0,\,\epsilon_3=0.5463>0,\,\epsilon_4=0.0317>0,$ and $\epsilon_1\epsilon_2\epsilon_3-\epsilon_1^2\epsilon_4 - \epsilon_3^2=2.1276>0.$ The coefficients of characteristic equation at $S_2^*=(1.1931, 3.08517,6.57356,
		55.5331)$ are  calculated  as $\epsilon_1=2.453>0,\,\epsilon_2=1.7912>0,\,\epsilon_3=0.3445<0,\,\epsilon_4=-0.0187<0,$ and $\epsilon_1\epsilon_2\epsilon_3-\epsilon_1^2\epsilon_4 - \epsilon_3^2=1.5077>0.$ The coefficients of characteristic equation at $S_3^*=(3.15301,1.67693,17.3719,
		30.1848)$ are  calculated  as $\epsilon_1=2.453>0,\,\epsilon_2=1.9428>0,\,\epsilon_3= 0.5339>0,\,\epsilon_4=0.03001>0,$ and $\epsilon_1\epsilon_2\epsilon_3-\epsilon_1^2\epsilon_4 - \epsilon_3^2=2.0785>0.$  Therefore, from Theorem \ref{existence_OR2}, we conclude that $S_1^*,\,S_3^*$ is locally asymptotically stable and $S_2^*$ is unstable.  We draw the phase portrait for the considered parameters, shown in Fig. \ref{F3_OR2} (c). We can see that all trajectories are converging to $S_1^*,\, S_3^*,$ and nearby trajectories of $S_2^*$ are moving away from $S_2^*.$ This shows the bistable behavior of the model system (\ref{1_OR2}).\\
		We chosen $\theta_b=5.94\in\theta_b^3,$
		the coefficients of characteristic equation at $S_3^*=(3.7429,\\ 1.47508,20.622, 26.5515)$ are  calculated  as $\epsilon_1=2.453>0,\,\epsilon_2=1.9886>0,\,\epsilon_3= 0.5912>0,\,\epsilon_4=0.0417>0,$ and $\epsilon_1\epsilon_2\epsilon_3-\epsilon_1^2\epsilon_4 - \epsilon_3^2=2.2838>0.$ Therefore, from Theorem \ref{existence_OR2}, we conclude that $S_3^*$ is locally asymptotically stable. We draw the phase portrait for the chosen parameters, shown in Fig. \ref{F3_OR2} (d). We can see that all trajectories starting from any initial points are converging to $S_3^*.$\\
		In this example, we also discuss the case of saddle-node bifurcation. We have seen that the system (\ref{1_OR2}) has two steady states $S_2^*,\,S_3^*$ in $\theta_b^2$ such that as the value of $\theta_b$ decreases, the two steady states collide at $\theta_b^{{{SN}^1}}=2.154867833473525$ and denoted as $S_{{SN}^1}^*=(2.12792,2.20183,11.7241,
	39.6329).$ The Jacobian matrix $J=Df(S_{{SN}^1}^*,\theta_b^{{{SN}^1}})$ has a simple eigenvalue $\lambda=0$ and  $v=(-0.0748,0.0504, -0.4121,0.9067)^T,$\\$w=(-0.9192,0.0771,-0.3861,0.0110)^T$, are the eigenvectors of $J$ and $J^T,$ respectively. Both the tranversality conditions $w^Tf_{\theta_b}(S_{{SN}^1}^*,\theta_b^{{{SN}^1}})=-0.0709\neq 0$ and \\ $	w^T[D^2f(S_{{SN}^1}^*,\theta_b^{{{SN}^1}})(v,v)]=0.0015\neq 0$ are satisfied. Hence, the system (\ref{1_OR2}) experiences saddle-node bifurcation at $\theta_b=\theta_b^{{{SN}^1}}.$ Also, the system (\ref{1_OR2}) has two steady states $S_1^*,\,S_2^*$ in $\theta_b^2$ such that as the value of $\theta_b$ increases, the two steady states collide at $\theta_b^{{{SN}^2}}=4.804927166491674$ and denoted as $S_{{SN}^2}^*=(0.7238,3.8654,3.9879,69.5764).$ The Jacobian matrix $J=Df(S_{{SN}^2}^*,\theta_b^{{{SN}^2}})$ has a simple eigenvalue $\lambda=0$ and  $v=(0.0263,-0.0549,0.1450,-0.9876)^T,\,w=(0.9206,-0.0544,0.3866,-0.0078)^T$, are the eigenvectors of $J$ and $J^T,$ respectively. Both the tranversality conditions $w^Tf_{\theta_b}(S_{{SN}^2}^*,\theta_b^{{{SN}^2}})=0.0394\neq 0$ and  $	w^T[D^2f(S_{{SN}^2}^*,\theta_b^{{{SN}^2}})(v,v)]=0.0004194\neq 0$ are satisfied. Hence, the system (\ref{1_OR2}) experiences saddle-node bifurcation at $\theta_b=\theta_b^{{{SN}^2}}.$\\
The system (\ref{1_OR2}) exhibits a hysteresis effect in $\theta_b^2$ where multiple steady states coexist, as
shown in Fig \ref{F3_OR2} (a). The two outer steady states are stable, while the interior steady state (red) is unstable.\\
	Again, in this case, the parameters are chosen as $\delta_a = 0.09,\,\delta_b = 0.9,\, \gamma_a = 1,\, \gamma_b = 0.9,\, k_a = 0.25,\, k_b = 0.9,\, m_b = 0.1,\,  m_a =2,\,\theta_a = 0.4,\, A_1 =0.0035,\, B_1 =7.85,\, \theta_{aa} =20,\,\theta_b = (0.00001,50).$ For this set of parameters, we obtain a unique stable steady state as discussed in Case 2 (OR2). The corresponding bifurcation plot is shown in Fig. \ref{F3_OR2} (e). To ensure the stability of $S_1^*,$  we choose $\theta_b=10$ from the range of $\theta_b.$  The coefficients of characteristic equation at $S_1^*$ are  calculated as $\epsilon_1=2.89>0,\,\epsilon_2=2.8546>0,\,\epsilon_3=1.0316>0,\,\epsilon_4=0.0669>0,$ and $\epsilon_1\epsilon_2\epsilon_3-\epsilon_1^2\epsilon_4 - \epsilon_3^2=6.88757>0.$ Therefore, from Theorem \ref{existence_OR2}, we conclude that $S_1^*$ is locally asymptotically stable. We draw the phase portrait for the considered $\theta_b$, shown in Fig. \ref{F3_OR2} (f). We can see that all trajectories starting from any initial points are converging to $S_1^*=(1.11812, 8.7349,  3.1059, 8.7349).$
 \begin{figure}[H]
			\subfigure[]{\includegraphics[scale=0.27]{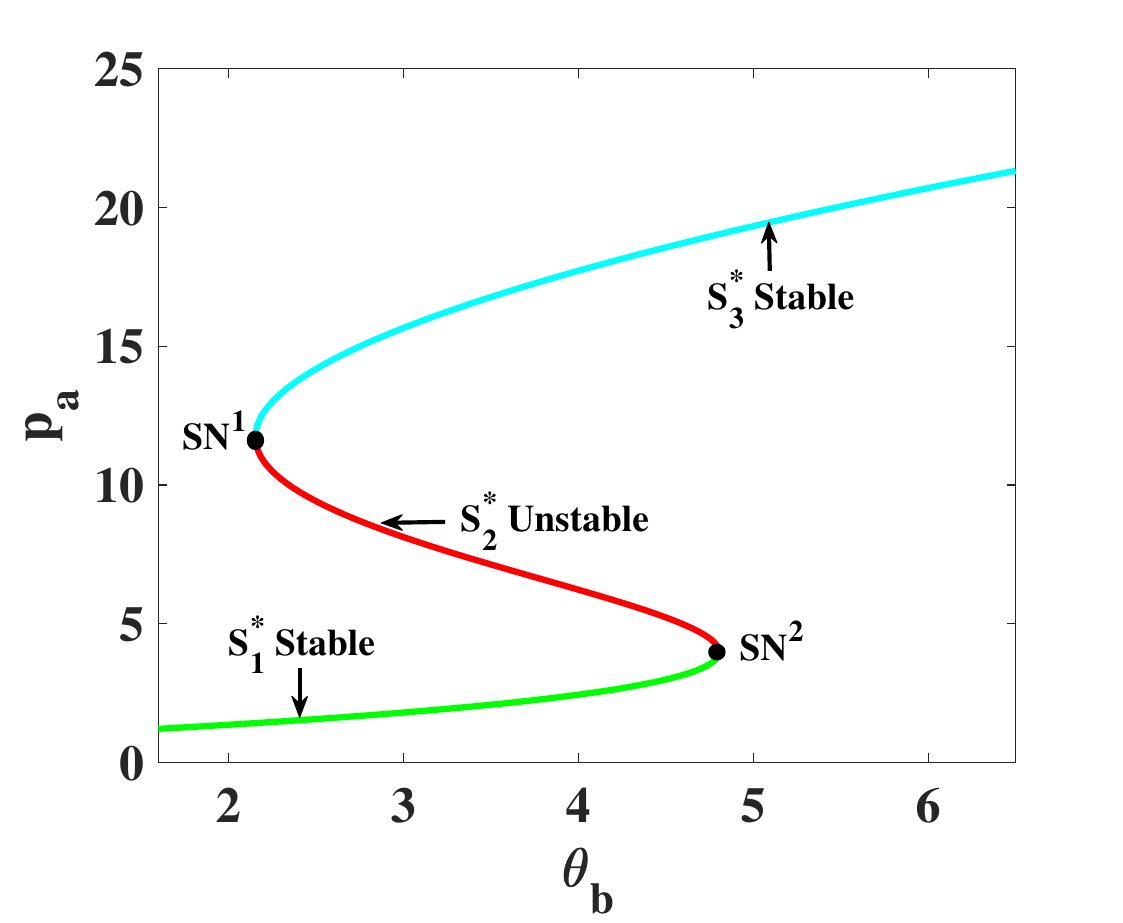}}
			\hfill
			\subfigure[]{\includegraphics[scale=0.27]{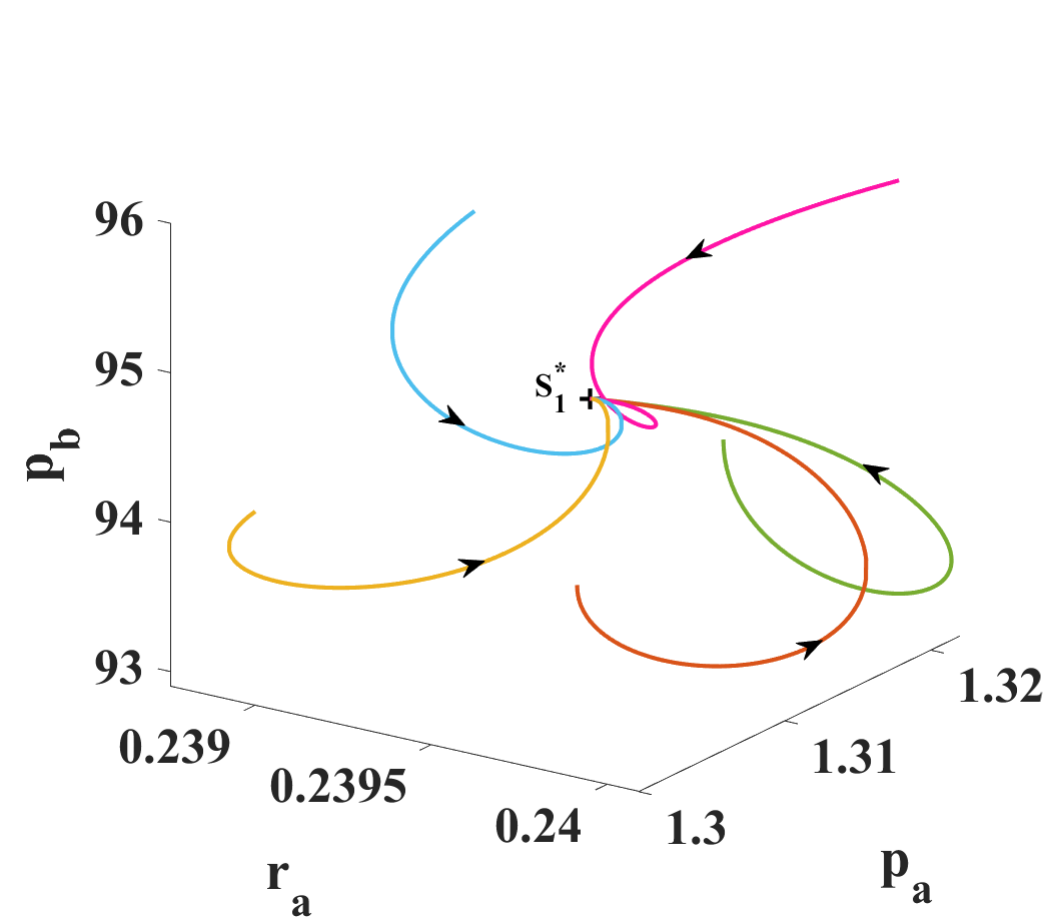}}
			\hfill
			\subfigure[]{\includegraphics[scale=0.27]{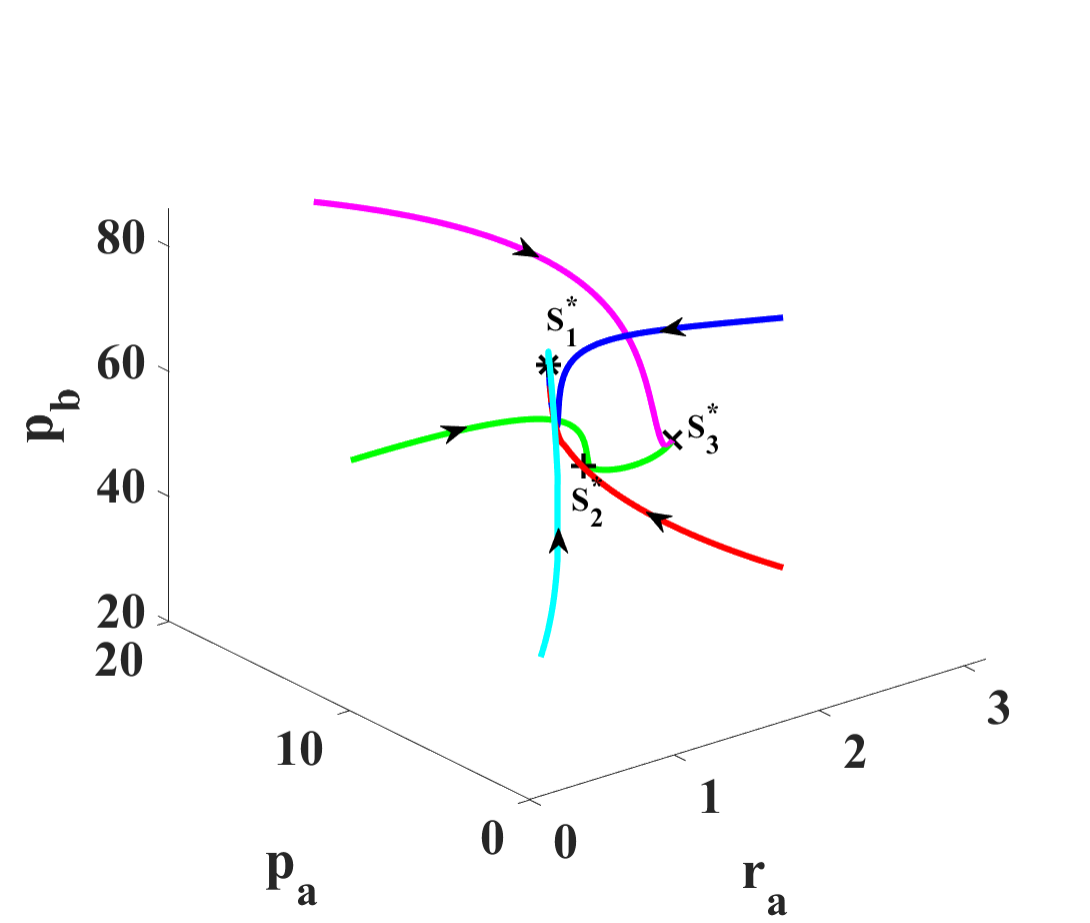}}
			\hfill
			\subfigure[]{\includegraphics[scale=0.27]{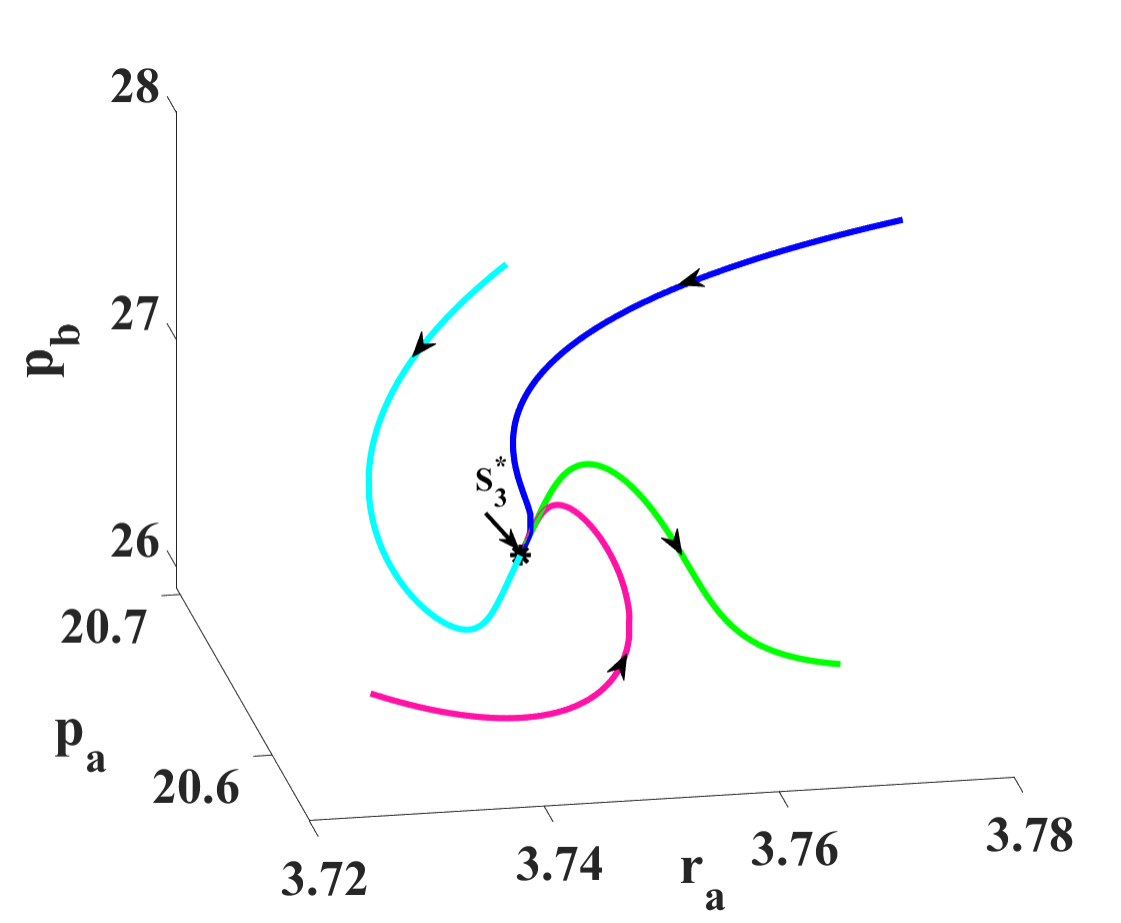}}
			\hfill
   \subfigure[]{\includegraphics[scale=0.27]{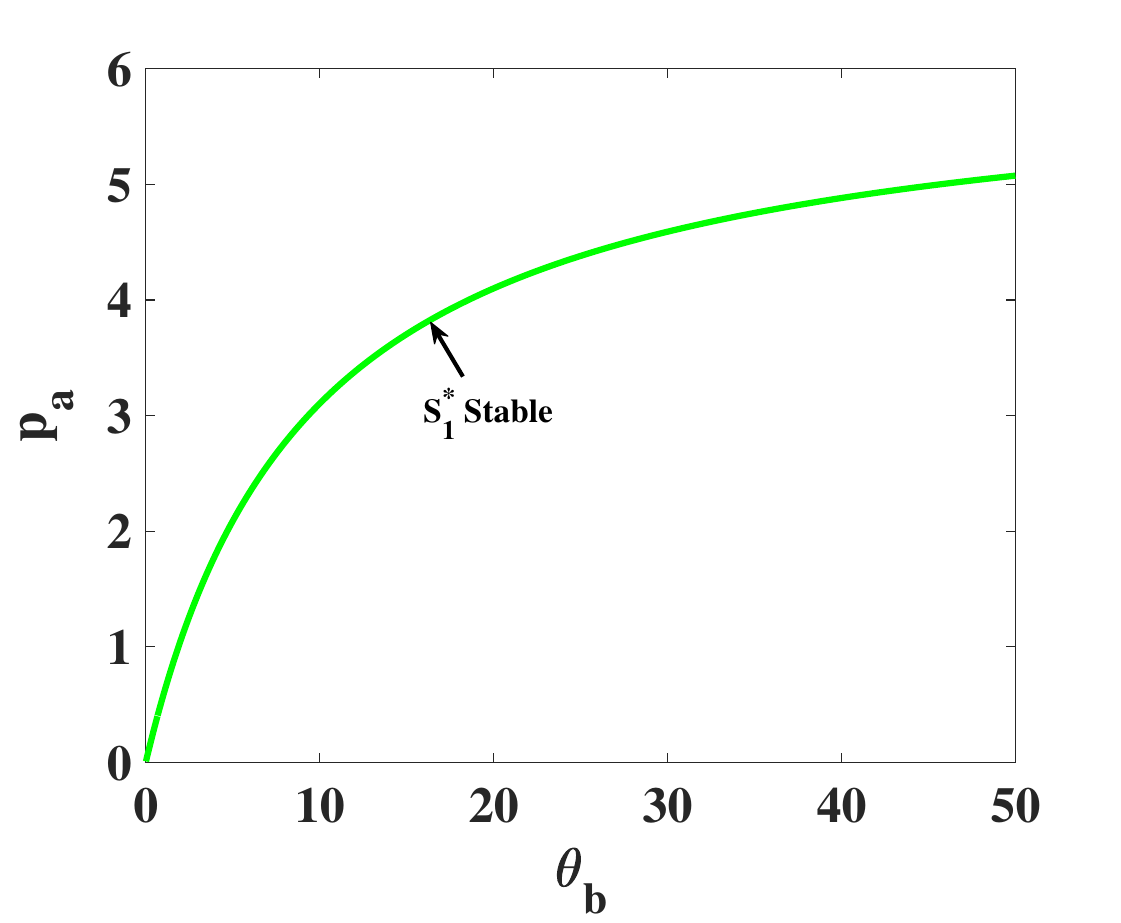}}
	\hfill
	\subfigure[]{\includegraphics[scale=0.27]{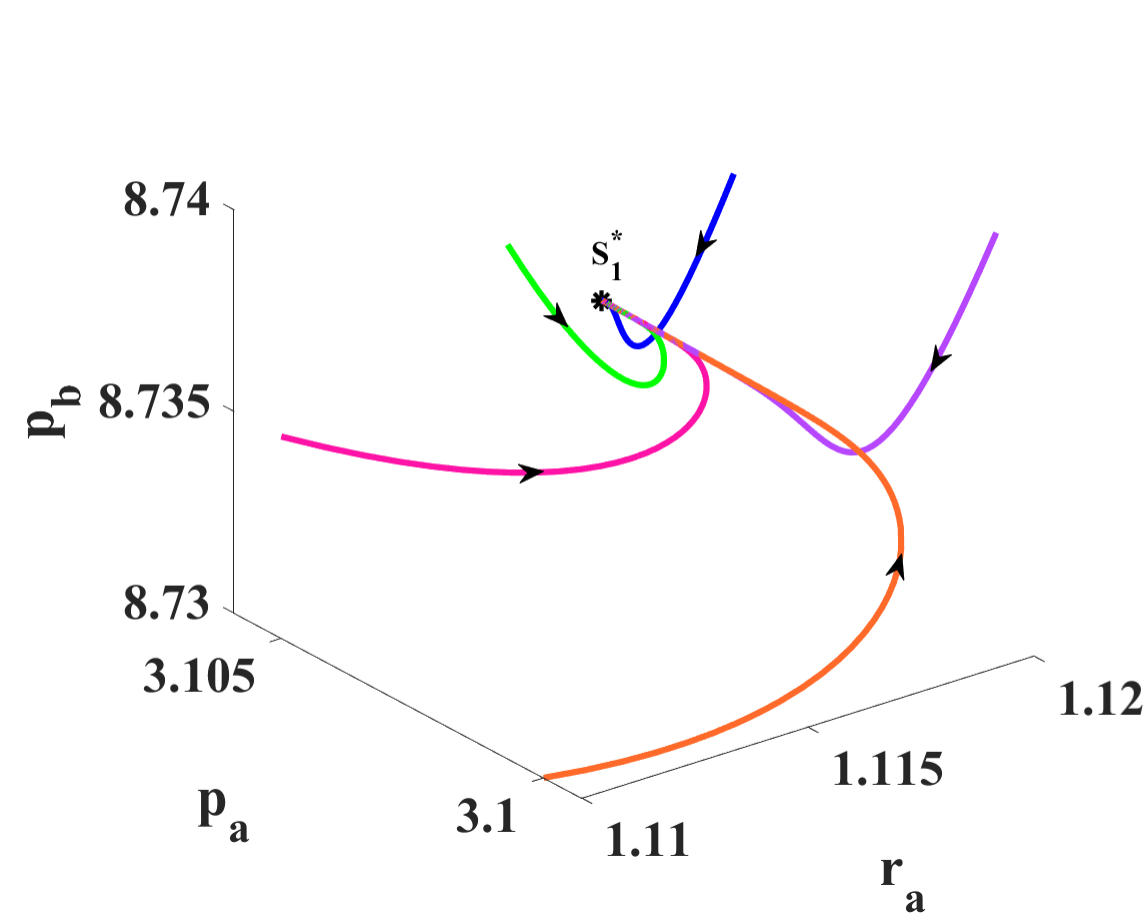}}
	\hfill
			\caption{For $n_a=1,\,n_b=1,\,n_{aa}=2$; (a) Bifurcation plot exhibiting hysteresis effect with three steady states (b) Phase portrait at $\theta_b=1.89$ (c) Phase portrait at $\theta_b=3.81$ depicting bistability (d)  Phase portrait at $\theta_b=5.94$ (e) Bifurcation plot with the unique stable steady state (f) Phase portrait at $\theta_b=10.$}
			\label{F3_OR2}
		\end{figure}
		\end{example}
  \begin{example}
  $n_a=2,\,n_b=3,\,n_{aa}=2$ (case 3) and the other parameters are chosen as $\delta_a = 0.32,\, \delta_b = 0.35,\, \gamma_a = 0.25,\,\gamma_b = 1,\,k_a = 0.5,\, k_b = 0.91,\, m_b = 15.5,\, m_a = 3.4,\, \theta_a = 9.4,\, A_1 = 0.2,\, B_1 = 0.01,\, \theta_{aa} = 11.67,\,\theta_b\in(3,11).$
  In this range of $\theta_b,$ there is three ($S_1^*,\,S_2^*,\,S_3^*$) steady states in $\theta_b^1$, five ($S_1^*,\,S_2^*,\,S_3^*,\,S_4^*,\,S_5*$) in $\theta_b^2,$ three ($S_1^*,\,S_2^*,\,S_5*)$ in $\theta_b^3,$ and unique ($S_5^*$) in $\theta_b^4,$ such that $\theta_b=\theta_b^1\cup\theta_b^2\cup\theta_b^2\cup\theta_b^4$ where $\theta_b^1=(3,3.51),\,\theta_b^2=(3.52,3.95),\,,\,\theta_b^3=(3.96,10.21)$ and $\theta_b^4=(10.22,11).$ The steady states $S_1^*,\,S_3^*,\,S_5^*$ are stable and $S_2^*,\,S_4^*$ are unstable. The corresponding bifurcation diagram is shown in Fig. \ref{tri_OR2} (a).\\	
To ensure the tristability in $\theta_b^2,$ we chosen $\theta_b=3.7\in \theta_b^2$ and the coefficients of characteristic equation at $S_1^*=(1.10066,15.008,1.71979,39.0207)$ are  calculated  as $\epsilon_1=1.92>0,\,\epsilon_2=1.1584>0,\,\epsilon_3=0.2520>0,\,\epsilon_4=0.0135>0,$ and $\epsilon_1\epsilon_2\epsilon_3-\epsilon_1^2\epsilon_4 - \epsilon_3^2=0.4470>0.$ 
The coefficients of characteristic equation at $S_2^*=(4.69214,9.64743,7.33148,25.0833)$ are  calculated  as $\epsilon_1=1.92>0,\,\epsilon_2=1.1054>0,\,\epsilon_3=0.1805>0,\,\epsilon_4=-0.0055<0,$ and $\epsilon_1\epsilon_2\epsilon_3-\epsilon_1^2\epsilon_4 - \epsilon_3^2=0.3708>0.$ 
The coefficients of characteristic equation at $S_3^*=(10.4286,3.88023,16.2947,10.0886)$ are  calculated  as $\epsilon_1=1.92>0,\,\epsilon_2=1.1527>0,\,\epsilon_3=0.2444>0,\,\epsilon_4=0.0043>0,$ and $\epsilon_1\epsilon_2\epsilon_3-\epsilon_1^2\epsilon_4 - \epsilon_3^2=0.4653>0.$ 
The coefficients of characteristic equation at $S_4^*=(16.8118,1.7695,26.2684,4.60071)$ are  calculated  as $\epsilon_1=1.92>0,\,\epsilon_2=1.1817>0,\,\epsilon_3=0.2834>0,\,\epsilon_4=-0.0052<0,$ and $\epsilon_1\epsilon_2\epsilon_3-\epsilon_1^2\epsilon_4 - \epsilon_3^2=0.5819>0.$ 
The coefficients of characteristic equation at $S_5^*=(23.5993,0.95581,36.874,2.48511)$ are  calculated  as $\epsilon_1=1.92>0,\,\epsilon_2=1.1919>0,\,\epsilon_3=0.2972>0,\,\epsilon_4=0.0093>0,$ and $\epsilon_1\epsilon_2\epsilon_3-\epsilon_1^2\epsilon_4 - \epsilon_3^2=0.5576>0.$ 
Therefore, from Theorem \ref{existence_OR2}, we conclude that $S_1^*,\,S_3^*,\,S_5^*$ are locally asymptotically stable and $S_2^*,\,S_4^*$ are unstable. We draw the phase portrait for the considered parameters, shown in Fig. \ref{tri_OR2} (b). We can see that all trajectories starting from any initial points are converging to $S_1^*,\,S_3^*,\,S_5^*$ and nearby trajectories of $S_2^*,\,S_4^*$ are moving away from $S_2^*,\,S_4^*.$ This shows the tristable behavior of the model system (\ref{1_OR2}).\\
		Further, to ensure the bistability in $\theta_b^3,$ we chosen $\theta_b=7\in \theta_b^3,$  and the coefficients of characteristic equation at $S_1^*=(1.25204, 14.8665,1.95631, 38.653)$ are  calculated  as $\epsilon_1=1.92>0,\,\epsilon_2=1.1533>0,\,\epsilon_3=0.2451>0,\,\epsilon_4=0.0114>0,$ and $\epsilon_1\epsilon_2\epsilon_3-\epsilon_1^2\epsilon_4 - \epsilon_3^2=0.4407>0.$
The coefficients of characteristic equation at $S_2^*=(3.84503, 11.0147, 6.00785,\\ 28.6382)$ are  calculated  as $\epsilon_1=1.92>0,\,\epsilon_2=1.1058>0,\,\epsilon_3=0.1810<0,\,\epsilon_4=-0.0072<0,$ and $\epsilon_1\epsilon_2\epsilon_3-\epsilon_1^2\epsilon_4 - \epsilon_3^2=0.3782>0.$ 
The coefficients of the characteristic equation at $S_5^*=(26.7174, 0.757962,  41.7459,  1.9707)$ are  calculated  as $\epsilon_1=1.92>0,\,\epsilon_2=1.194>0,\,\epsilon_3= 0.3001>0,\,\epsilon_4=0.0244>0,$ and $\epsilon_1\epsilon_2\epsilon_3-\epsilon_1^2\epsilon_4 - \epsilon_3^2=0.5081>0.$  Therefore, from Theorem \ref{existence_OR2}, we conclude that $S_1^*,\,S_5^*$ are locally asymptotically stable and $S_2^*$ is unstable.  We draw the phase portrait for the considered parameters, shown in Fig. \ref{tri_OR2} (c). We can see that all trajectories are converging to $S_1^*,\, S_5^*,$ and nearby trajectories of $S_2^*$ are moving away from $S_2^*.$ This shows the bistable behavior of the model system (\ref{1_OR2}).\\
To ensure the mono-stability in $\theta_b^4,$ we chosen  $\theta_b=10.5\in\theta_b^4,$ and 
		the coefficients of characteristic equation at $S_5^*=(26.9444, 0.746008,  42.1006, 1.93962)$ are  calculated  as $\epsilon_1=1.92>0,\,\epsilon_2=1.1941>0,\,\epsilon_3= 0.3003>0,\,\epsilon_4=0.0256>0,$ and $\epsilon_1\epsilon_2\epsilon_3-\epsilon_1^2\epsilon_4 - \epsilon_3^2=0.5038>0.$ Therefore, from Theorem \ref{existence_OR2}, we conclude that $S_5^*$ is locally asymptotically stable. We draw the phase portrait for the chosen parameters, shown in Fig. \ref{tri_OR2} (d). We can see that all trajectories starting from any initial points are converging to $S_5^*.$
 \begin{figure}[H]
 \centering
			\subfigure[]{\includegraphics[scale=0.33]{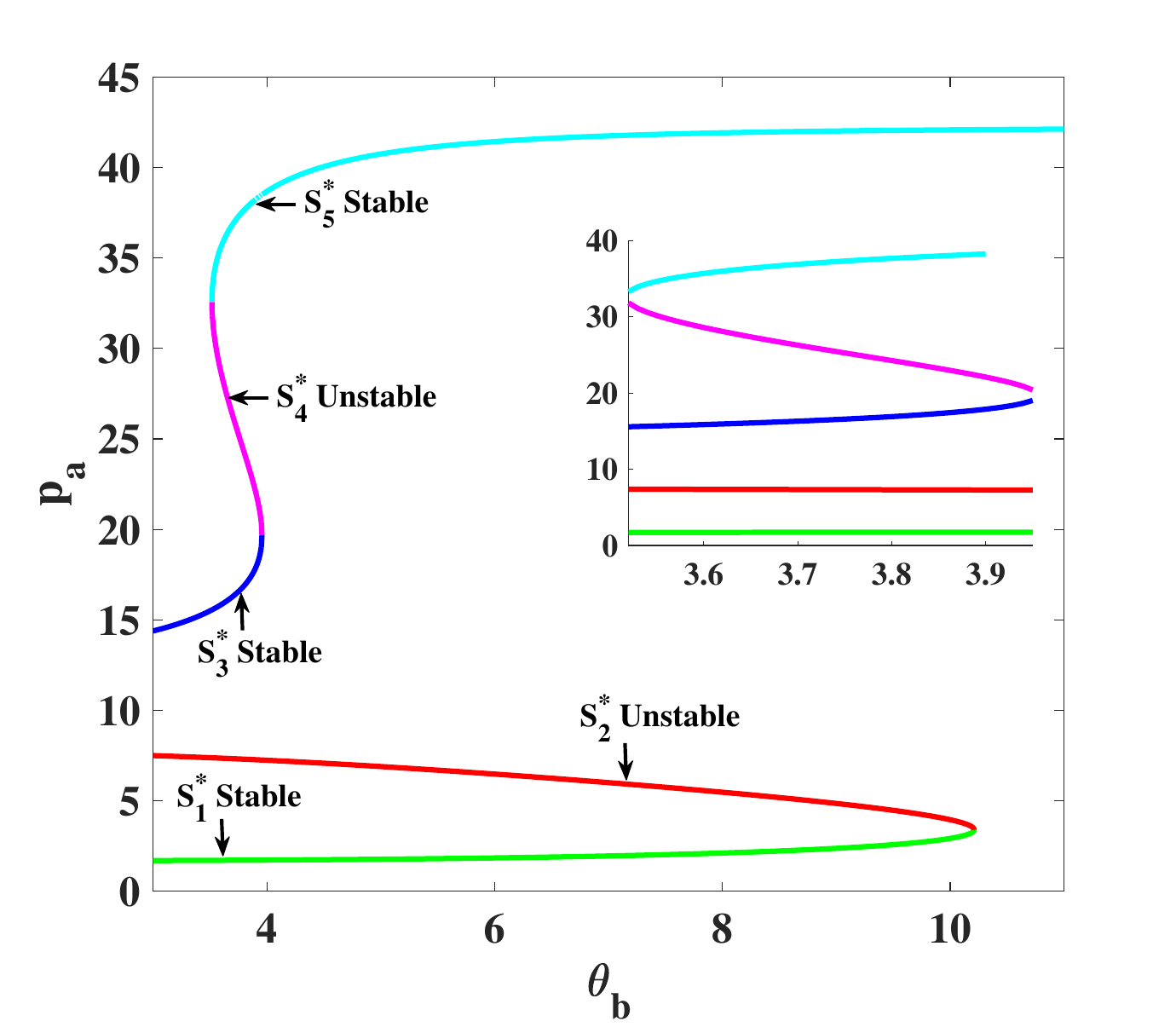}}
			\hfil
			\subfigure[]{\includegraphics[scale=0.33]{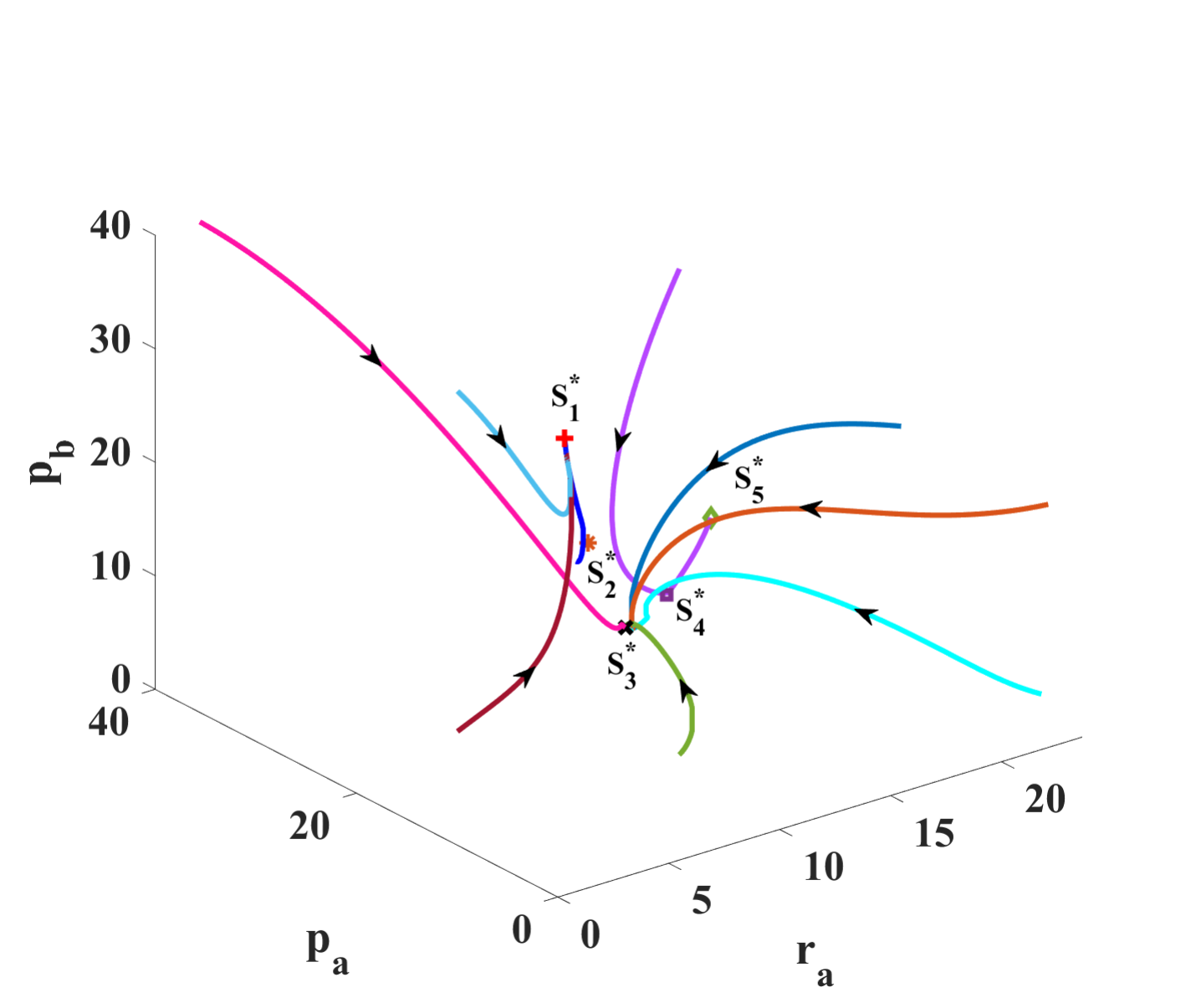}}
			\end{figure}
    \begin{figure}[H]
 \centering
			\subfigure[]{\includegraphics[scale=0.3]{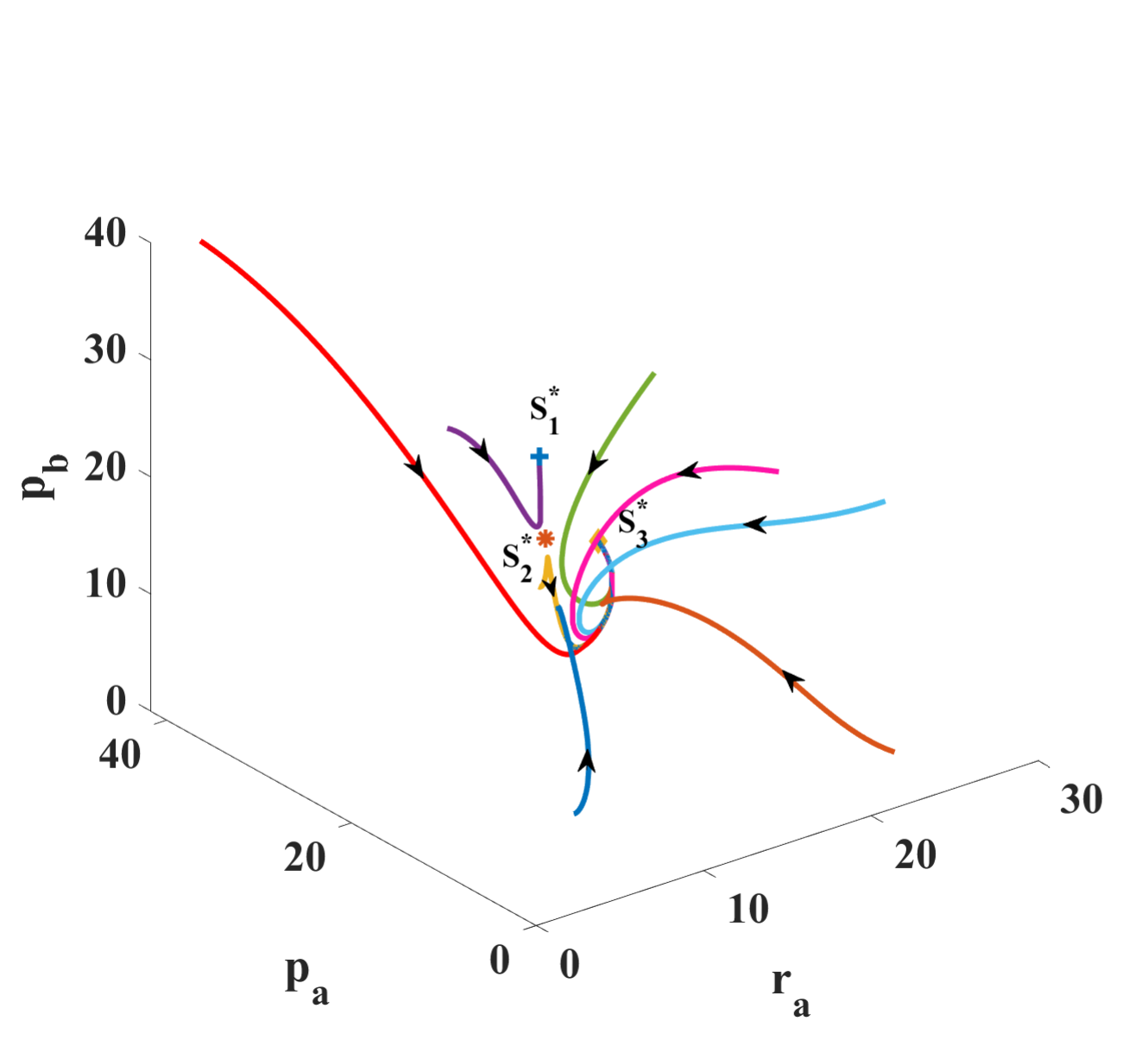}}
			\hfil
			\subfigure[]{\includegraphics[scale=0.3]{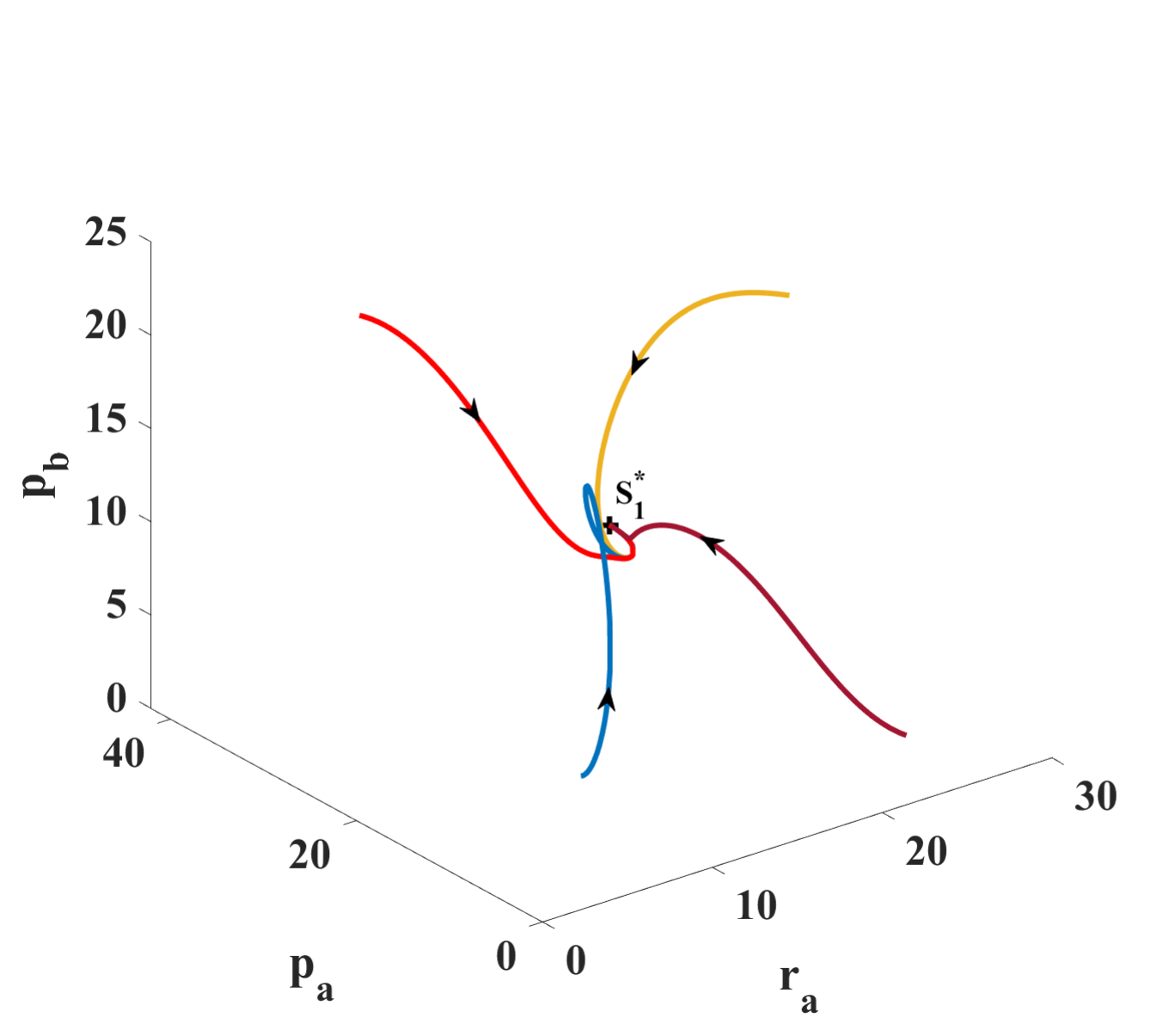}}
			 \caption{For $n_a=2,\,n_b=3,\,n_{aa}=2$; (a) Bifurcation plot with five steady states (b) Phase portrait at $\theta_b=3.7$ depicting tristability (c) Phase portrait at $\theta_b=7$ depicting bistability (d)  Phase portrait at $\theta_b=10.5.$}
			\label{tri_OR2}
		\end{figure}
  \end{example}

\subsection{Conclusion and biological interpretation for this section}
In this section, we consider a non-competitive type-II $\mathcal{OR}$ logic as an alternate to competitive OR logic. Here we assume that the two genes $a$ and $b$ don't compete for binding and any binding events can elicit gene expression. Our results overlap with that of AND logic by showing that bistability requires dimeric (Hill coefficients of degree 2) autoloop and cell state transition can occur only through abrupt transitions (hysteresis). Nevertheless, the main result of this section is tristability, which we show is possible only in higher degrees of multimerization (e.g., ($n_a, \ n_{aa}, \ n_b) \ = \ (2,\ 2, \ 3$)). One manifestation of tristability is phenotypic heterogeneity (concurrent occurrence of three phenotypes) and we also observe that cell transition occurs abruptly after a certain threshold of $\theta_b$ is crossed (Fig. \ref{tri_OR2}a). We thus conclude that an autoregulated two-node positive feedback circuit with non-competitive OR logic (type-II $\mathcal{OR}$) is the "minimal" circuit to exhibit multistability under higher degrees of multimerization.\\

\section{Impact of basal production rate on model dynamics}
Basal production rate or leakage defines the production of mRNA in the absence of any activatory or inhibitory interacting transcription factors. For instance, the basal production of $r_a$ will be the rate of production of $r_a$ in "normal" conditions, i.e., when the self-activation due to $a$ itself or inhibition by $b$ is absent. While modeling gene expression (transcription) processes, leakage is sometimes neglected at the cost of simplifying modeling efforts and analytical calculations. However, in all of our model systems, (\ref{1}), (\ref{and}) and (\ref{1_OR2}), we have considered the leakage rates of mRNA's, $r_a$ and $r_b$ as $A_1$ and $B_1$, respectively. 

We wanted to investigate whether the absence of the leakage rate of $r_a$ will have any impact on the model dynamics. To understand that, we compared the mRNA concentration $r_a$ for the model systems (\ref{1}), (\ref{and}) and (\ref{1_OR2}) for the zero and non-zero values of $A_1,$ as shown in Fig. \ref{orandOR2} (a) and (b), respectively. From Fig. \ref{orandOR2}, we observed that the level of mRNA concentration $r_a$ is high for the model with type-II $\mathcal{OR}$ logic (model system \ref{1_OR2}) (blue color curve) and low for the model with type-I OR logic (model system \ref{1}) (orange color curve). Also, it is evident that the level of mRNA concentration $r_a$ is high for all the three model systems with non-zero $A_1$ (Fig. \ref{orandOR2}b) compared to $A_1=0$ case (Fig. \ref{orandOR2}a). However, the structure of the time-dynamics curves of mRNA concentration remains largely unaltered. Our results thus show that leakage only plays the role of "scaling" in the model and its incorporation in the model systems won't have any impact on the overall qualitative dynamics of the network. We thus conclude that our results hold true even when the leakage is completely ignored. \\

   \begin{figure}[H]
   \centering
			\subfigure[]{\includegraphics[scale=0.3]{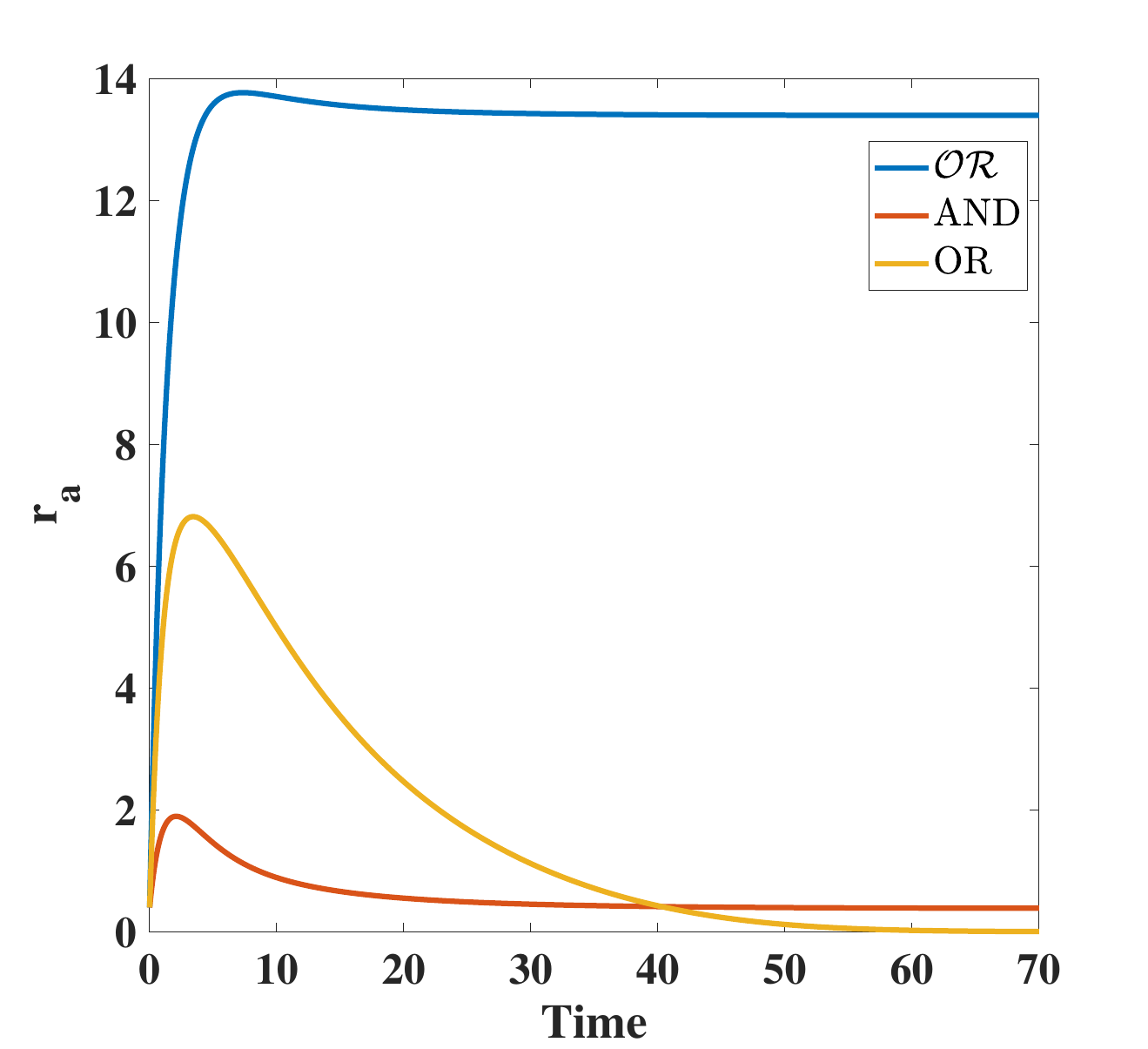}}
   \hfil
			\subfigure[]{\includegraphics[scale=0.3]{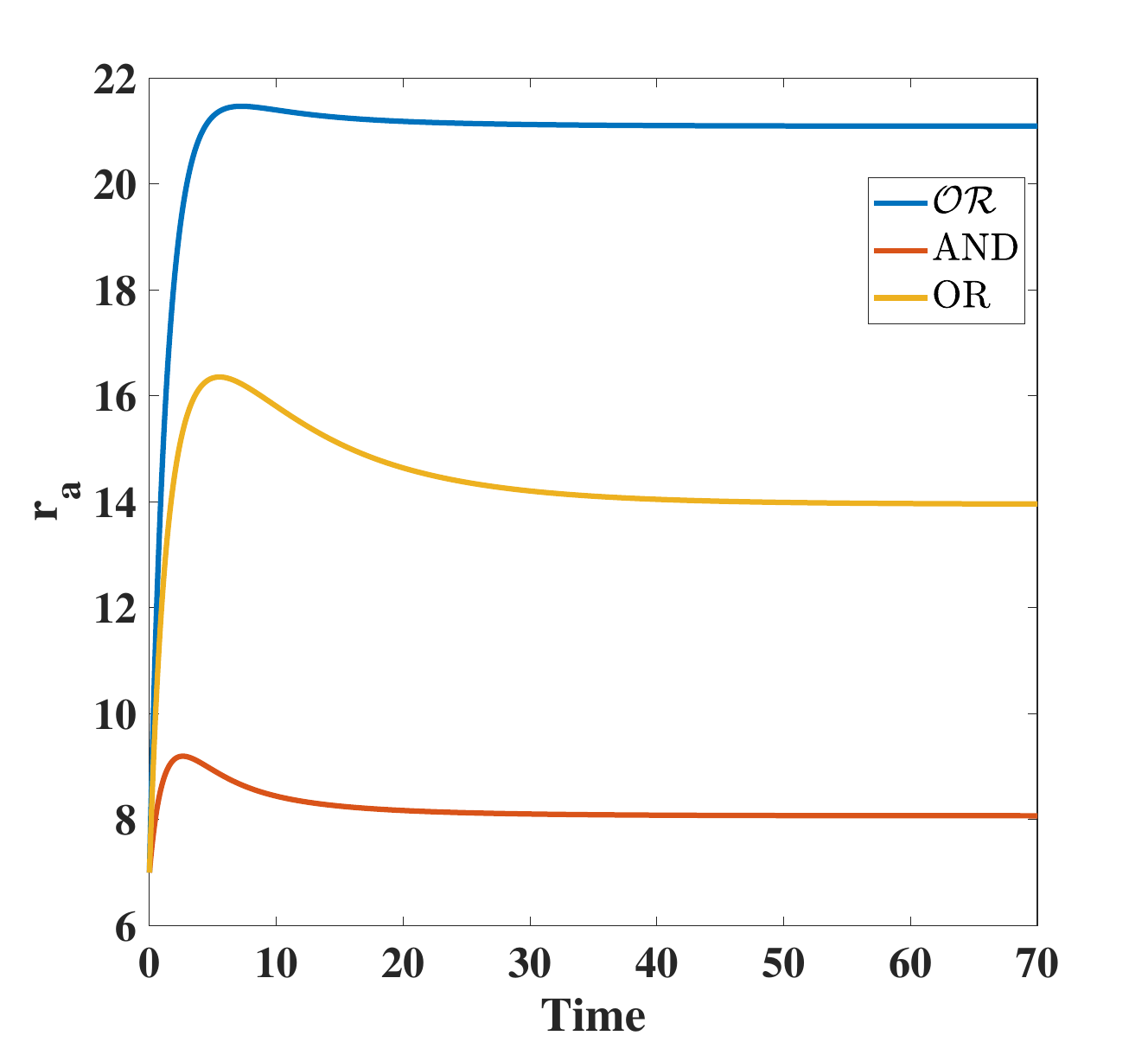}}
	 \caption{Comparative study of temporal dynamics of mRNA concentration $r_a$ of gene $a$ using three regulatory logics: Type-I OR logic, AND logic, and Type-II $\mathcal{OR}$ logic.(a) when $A_1=0$ and (b) when $A_1\neq0$. The parameter values for these plots are same as that for Fig. \ref{case1traj_new} (f)}.
			\label{orandOR2}
		\end{figure}

\section{Discussion and Final Conclusion}
Cell fate switching is a dynamic phenomenon often tied to regulatory network motifs that at the cellular level define the computational machinery of life. Most of these network motifs define molecular switches exhibiting diverse qualitative behaviors such as bistability, catastrophes, and hysteresis \cite{cell-fate-roy-soc-inter}. The most prominent examples of molecular switches involving minimal circuitry are the two-component positive feedback network motifs resulting from mutual repression or mutual activation of two genes. Traditional studies have shown that these motifs can exhibit at most bistable dynamics (two stable states or attractors) allowing the system to alternately switch between two states. Further, the underlying cause of bistable dynamics is attributed to multimeric regulation (higher order Hill coefficient), in contrast to monomeric regulation that gives rise to only monostable dynamics. The biological programs underlying cell fate decision-making are however not just restricted to mono-and bistable regimes of dynamics. In fact, higher-order dynamics such as tristability (three stable states/attractors/phenotypes) is now also prominently observed in biological mechanisms. This happens in differentiation programs such as differentiation of naive CD4+ T cells \cite{CD4-cell-differ-tristability} as well as in diseases-inducing processes such as epithelial-mesenchymal transition (EMT) in which non-motile epithelial cells switch to mesenchymal and hybrid epithelial-mesenchymal cell fates which have migratory and invasive traits that often causes metastasis of carcinomas \cite{microRNA-based}. Nevertheless, the minimal network motif that can exhibit tristability, or for that matter, mon, bi-and tri stability, remains largely unexplored. 

In our previous study on networks underlying and driving EMT, we used a “numerical” approach to report that a positive feedback loop formed by mutual repression of two genes, miR200 ( microRNA) and ZEB (mRNA), with one gene, ZEB, self-activated can exhibit tristability. However, this being a microRNA-mRNA circuit involves translational repression and active degradation of ZEB due to micro-RNA, making it a complicated mechanism. We asked whether a generic positive feedback loop formed by mutual repression of two transcription factors (rather than microRNAs and mRNA’s) with an autoloop, can be the minimal circuit to exhibit lower as well as higher-order stability regimes such as mon, bi-and tri stability. While small feedback circuits with and without delays have been studied previously \cite{wang2020dynamic,xi2015parameter, hogan-time-delay, oper-princp-tristable_switches, kishore-coupled-bistable-switches, delaydiffeq-natcomm, ruiqi-iden-crit-inter}, we still lack answers with rigorous analytical foundations to the following crucial questions: (a) How is multimerization associated with the order of stability? (b) Does the “logic” of regulation have any functional role in multistability? In this work, we address these questions using a rigorous analytical (and numerical) approach that allows us to reach the following conclusions summarised in Table \ref{summary}. 

We show that a positive feedback loop with an autoloop can exhibit diverse dynamical behaviors that depend on the degree of multimerization as well as the “logic” of regulation. The autoloop can be integrated into the core circuit through three types of logic gates: OR (Type-I and Type-II), and AND. With Type-I OR logic, the monomeric circuit with monomeric autoloop can show biphasic kinetics as well as hysteresis (bistability). However, with AND and Type-II logics the circuit can show hyperbolic saturation kinetics - a manifestation of monostability. We show that the necessary geometrical constraint for hysteresis (bistability) with AND and Type-II logics imposes the existence of (at least) dimeric autoloop. We also analyzed higher degrees of multimerization and reported the prevalence of biphasic as well as mono-and bistable dynamics. Finally, and strikingly, we show the possibility of tristability in this circuit. The two requirements for tistable dynamics are: (a) the degrees of multimerization for the two regulations in the core circuit should be (at least) two and three, (b) the degree of multimerization of autoloop should be (at least) two (i. e., dimeric). Taken together, our analysis shows how different multimerization and logical constraints in this minimal circuit can work together to produce numerous types of dynamical scenarios including monostability,  hysteresis (bistability), biphasic kinetics, and importantly, tristability. While hysteresis translates to forced and abrupt cell state transitions, biphasic kinetics explains $smooth \ cell \ state \ swap$ without a step-like switch. The biological network (i.e., miR200-ZEB circuit) corresponding to this generic circuit underlies EMT, and drives cell fate transition and metastasis in carcinomas \cite{genetic-dissect}. Besides, it has been shown (experimentally) to exhibit hysteric as well non-hysteric dynamics, with only the hysteric EMT enabling lung metastasis \cite{hysteric-control}, our results can thus have crucial implications in furthering our understanding of EMT mechanism due to miR200-ZEB feedback loop and can provide inputs to control cell-fate switching during metastasis of carcinomas. 
\begin{table}[H]
\centering
\scalebox{0.55}{
\begin{tabular}{|l|lll|lll|lll|}
\hline
\multicolumn{1}{|c|}{\multirow{2}{*}{\textbf{$(n_a,n_b,n_{aa})$}}} & \multicolumn{3}{c|}{\textbf{Type-I OR logic}}                                                                                                                                      & \multicolumn{3}{c|}{\textbf{AND logic}}                                                                                                                                            & \multicolumn{3}{c|}{\textbf{Type-II $\mathcal{OR}$ logic}}                                                                                                                                     \\ \cline{2-10} 
\multicolumn{1}{|c|}{}                                             & \multicolumn{1}{c|}{\textbf{Monostable}} & \multicolumn{1}{c|}{\textbf{\begin{tabular}[c]{@{}c@{}}Bistable\\ (Hysteresis)\end{tabular}}} & \multicolumn{1}{c|}{\textbf{Tristable}} & \multicolumn{1}{c|}{\textbf{Monostable}} & \multicolumn{1}{c|}{\textbf{\begin{tabular}[c]{@{}c@{}}Bistable\\ (Hysteresis)\end{tabular}}} & \multicolumn{1}{c|}{\textbf{Tristable}} & \multicolumn{1}{c|}{\textbf{Monostable}} & \multicolumn{1}{c|}{\textbf{\begin{tabular}[c]{@{}c@{}}Bistable\\ (Hysteresis)\end{tabular}}} & \multicolumn{1}{c|}{\textbf{Tristable}} \\ \hline
\multicolumn{1}{|c|}{\textbf{(1,1,1)}}                                              & \multicolumn{1}{c|}{$\checkmark$}        & \multicolumn{1}{c|}{$\checkmark$}                                                             & \multicolumn{1}{c|}{-}       & \multicolumn{1}{c|}{$\checkmark$}        & \multicolumn{1}{c|}{-}                                                             & \multicolumn{1}{c|}{-}       & \multicolumn{1}{c|}{$\checkmark$}        & \multicolumn{1}{c|}{-}                                                             & \multicolumn{1}{c|}{-}       \\ \hline
\multicolumn{1}{|c|}{\textbf{(1,1,2)}   }                                  & \multicolumn{1}{c|}{$\checkmark$}        & \multicolumn{1}{c|}{$\checkmark$}                                                             & \multicolumn{1}{c|}{-}     & \multicolumn{1}{c|}{$\checkmark$}        & \multicolumn{1}{c|}{$\checkmark$}                                                             & \multicolumn{1}{c|}{-}       & \multicolumn{1}{c|}{$\checkmark$}        & \multicolumn{1}{c|}{$\checkmark$}                                                             & \multicolumn{1}{c|}{-}      \\ \hline
\multicolumn{1}{|c|}{\textbf{(2,1,2)}   }                                      & \multicolumn{1}{c|}{$\checkmark$}        & \multicolumn{1}{c|}{$\checkmark$}                                                             & \multicolumn{1}{c|}{-}     & \multicolumn{1}{c|}{$\checkmark$}        & \multicolumn{1}{c|}{$\checkmark$}                                                             & \multicolumn{1}{c|}{-}       & \multicolumn{1}{c|}{$\checkmark$}        & \multicolumn{1}{c|}{$\checkmark$}                                                             & \multicolumn{1}{c|}{-}        \\ \hline
\multicolumn{1}{|c|}{\textbf{(2,2,1)}   }                                       & \multicolumn{1}{c|}{$\checkmark$}        & \multicolumn{1}{c|}{$\checkmark$}                                                             & \multicolumn{1}{c|}{-}     & \multicolumn{1}{c|}{$\checkmark$}        & \multicolumn{1}{c|}{$\checkmark$}                                                             & \multicolumn{1}{c|}{-}       & \multicolumn{1}{c|}{$\checkmark$}        & \multicolumn{1}{c|}{$\checkmark$}                                                             & \multicolumn{1}{c|}{-}        \\ \hline
\multicolumn{1}{|c|}{\textbf{(2,2,2)}   }                                         & \multicolumn{1}{c|}{$\checkmark$}        & \multicolumn{1}{c|}{$\checkmark$}                                                             & \multicolumn{1}{c|}{-}     & \multicolumn{1}{c|}{$\checkmark$}        & \multicolumn{1}{c|}{$\checkmark$}                                                             & \multicolumn{1}{c|}{-}       & \multicolumn{1}{c|}{$\checkmark$}        & \multicolumn{1}{c|}{$\checkmark$}                                                             & \multicolumn{1}{c|}{-}      \\ \hline
\multicolumn{1}{|c|}{\textbf{(2,3,2)}   }                                         & \multicolumn{1}{c|}{$\checkmark$}        & \multicolumn{1}{c|}{$\checkmark$}                                                             & \multicolumn{1}{c|}{-}     & \multicolumn{1}{c|}{$\checkmark$}        & \multicolumn{1}{c|}{$\checkmark$}                                                             & \multicolumn{1}{c|}{-}       & \multicolumn{1}{c|}{$\checkmark$}        & \multicolumn{1}{c|}{$\checkmark$}                                                             & \multicolumn{1}{c|}{$\checkmark$}        \\ \hline
\multicolumn{1}{|c|}{\textbf{(3,3,3)}   }                                       & \multicolumn{1}{c|}{$\checkmark$}        & \multicolumn{1}{c|}{$\checkmark$}                                                             & \multicolumn{1}{c|}{-}       & \multicolumn{1}{c|}{$\checkmark$}        & \multicolumn{1}{c|}{$\checkmark$}                                                             & \multicolumn{1}{c|}{-}       & \multicolumn{1}{c|}{$\checkmark$}        & \multicolumn{1}{c|}{$\checkmark$}                                                             & \multicolumn{1}{c|}{$\checkmark$}         \\ \hline

\end{tabular}}
\caption{Summary of different dynamical scenarios exhibited by the network.}
    \label{summary}
\end{table}

\vspace*{1.0cm}

\noindent {\bf Author Contributions} \\
\noindent {\bf MR}: Conceptualization, supervision, funding acquisition. {\bf AS}: Analytical calculations and numerical simulations. Both authors analysed the data, discussed the results, and wrote the manuscript. \\ 

\noindent {\bf Acknowledgments} \\
This work is supported by the Department of Science and
Technology, India [Grant No. DST/INSPIRE/04/2020/001492] and the Science and Engineering Research Board, India [Grant No. CRG/2023/006432] to Mubasher Rashid.

\vspace*{1.0cm}

\bibliographystyle{unsrt}
\bibliography{ref}

\begin{thebibliography}{10}

\bibitem{cell-fate-roy-soc-inter}
M~S{\'a}ez, J~Briscoe, and David~A Rand.
\newblock Dynamical landscapes of cell fate decisions.
\newblock {\em Interface focus}, 12(4):20220002, 2022.

\bibitem{und-gene_cir_cell-fate_branch_pts}
Joseph~X Zhou and Sui Huang.
\newblock Understanding gene circuits at cell-fate branch points for rational
  cell reprogramming.
\newblock {\em Trends in genetics}, 27(2):55--62, 2011.

\bibitem{Syn-gencir_cellDecMak}
Laura Prochazka, Yaakov Benenson, and Peter~W Zandstra.
\newblock Synthetic gene circuits and cellular decision-making in human
  pluripotent stem cells.
\newblock {\em Current Opinion in Systems Biology}, 5:93--103, 2017.

\bibitem{CellDecMak_JJ-collins}
G{\'a}bor Bal{\'a}zsi, Alexander Van~Oudenaarden, and James~J Collins.
\newblock Cellular decision making and biological noise: from microbes to
  mammals.
\newblock {\em Cell}, 144(6):910--925, 2011.

\bibitem{multiStab_dec-mak_in_diff}
Ra{\'u}l Guantes and Juan~F Poyatos.
\newblock Multistable decision switches for flexible control of epigenetic
  differentiation.
\newblock {\em PLoS computational biology}, 4(11):e1000235, 2008.

\bibitem{differ-ruiqi}
Dasong Huang and Ruiqi Wang.
\newblock Exploring the mechanisms of cell reprogramming and
  transdifferentiation via intercellular communication.
\newblock {\em Physical Review E}, 102(1):012406, 2020.

\bibitem{Phys-of-cell-dec-mak-in-EMT}
Shubham Tripathi, Herbert Levine, and Mohit~Kumar Jolly.
\newblock The physics of cellular decision making during epithelial-mesenchymal
  transition.
\newblock {\em Annual Review of Biophysics}, 49:1--18, 2020.

\bibitem{EMT-heterog_tripathi}
Shubham Tripathi, Priyanka Chakraborty, Herbert Levine, and Mohit~Kumar Jolly.
\newblock A mechanism for epithelial-mesenchymal heterogeneity in a population
  of cancer cells.
\newblock {\em PLoS computational biology}, 16(2):e1007619, 2020.

\bibitem{Heter&Plast}
Dandan Li, Lingyun Xia, Pan Huang, Zidi Wang, Qiwei Guo, Congcong Huang,
  Weidong Leng, and Shanshan Qin.
\newblock Heterogeneity and plasticity of epithelial--mesenchymal transition
  (emt) in cancer metastasis: Focusing on partial emt and regulatory
  mechanisms.
\newblock {\em Cell proliferation}, 56(6):e13423, 2023.

\bibitem{phenotypicHetero_Science_adv}
Meredith~S Brown, Behnaz Abdollahi, Owen~M Wilkins, Hanxu Lu, Priyanka
  Chakraborty, Nevena~B Ognjenovic, Kristen~E Muller, Mohit~Kumar Jolly,
  Brock~C Christensen, Saeed Hassanpour, et~al.
\newblock Phenotypic heterogeneity driven by plasticity of the intermediate emt
  state governs disease progression and metastasis in breast cancer.
\newblock {\em Science advances}, 8(31):eabj8002, 2022.

\bibitem{myPlosCompaper}
Mubasher Rashid, Kishore Hari, John Thampi, Nived~Krishnan Santhosh, and
  Mohit~Kumar Jolly.
\newblock Network topology metrics explaining enrichment of hybrid
  epithelial/mesenchymal phenotypes in metastasis.
\newblock {\em PLOS Computational Biology}, 18(11):e1010687, 2022.

\bibitem{why-are-cell-switch-boolean}
Javier Mac{\'\i}a, Stefanie Widder, and Ricard Sol{\'e}.
\newblock Why are cellular switches boolean? general conditions for multistable
  genetic circuits.
\newblock {\em Journal of theoretical biology}, 261(1):126--135, 2009.

\bibitem{necess-cond-kauffman}
Marcelline Kaufman, Christophe Soule, and Ren{\'e} Thomas.
\newblock A new necessary condition on interaction graphs for
  multistationarity.
\newblock {\em Journal of theoretical biology}, 248(4):675--685, 2007.

\bibitem{Angeli&Sontag2004}
David Angeli, James~E Ferrell~Jr, and Eduardo~D Sontag.
\newblock Detection of multistability, bifurcations, and hysteresis in a large
  class of biological positive-feedback systems.
\newblock {\em Proceedings of the National Academy of Sciences},
  101(7):1822--1827, 2004.

\bibitem{sontag-1}
German Enciso and Eduardo~D Sontag.
\newblock Monotone systems under positive feedback: multistability and a
  reduction theorem.
\newblock {\em Systems \& control letters}, 54(2):159--168, 2005.

\bibitem{sontag-2}
David Angeli and Eduardo~D Sontag.
\newblock Multi-stability in monotone input/output systems.
\newblock {\em Systems \& Control Letters}, 51(3-4):185--202, 2004.

\bibitem{graphic-req-multi}
Christophe Soul{\'e}.
\newblock Graphic requirements for multistationarity.
\newblock {\em ComPlexUs}, 1(3):123--133, 2003.

\bibitem{Toggle-Switch-JJCollins}
Timothy~S Gardner, Charles~R Cantor, and James~J Collins.
\newblock Construction of a genetic toggle switch in escherichia coli.
\newblock {\em Nature}, 403(6767):339--342, 2000.

\bibitem{develop}
Qing Zhou, Hiram Chipperfield, Douglas~A Melton, and Wing~Hung Wong.
\newblock A gene regulatory network in mouse embryonic stem cells.
\newblock {\em Proceedings of the National Academy of Sciences},
  104(42):16438--16443, 2007.

\bibitem{multistable-switches-differen}
Ahmadreza Ghaffarizadeh, Nicholas~S Flann, and Gregory~J Podgorski.
\newblock Multistable switches and their role in cellular differentiation
  networks.
\newblock {\em BMC bioinformatics}, 15:1--13, 2014.

\bibitem{diff-1}
Jinfang Zhu, Hidehiro Yamane, and William~E Paul.
\newblock Differentiation of effector cd4 t cell populations.
\newblock {\em Annual review of immunology}, 28:445--489, 2009.

\bibitem{diff-2}
Atchuta~Srinivas Duddu, Sarthak Sahoo, Souvadra Hati, Siddharth Jhunjhunwala,
  and Mohit~Kumar Jolly.
\newblock Multi-stability in cellular differentiation enabled by a network of
  three mutually repressing master regulators.
\newblock {\em Journal of the Royal Society Interface}, 17(170):20200631, 2020.

\bibitem{tristability-mouse}
Laurane De~Mot, Didier Gonze, Sylvain Bessonnard, Claire Chazaud, Albert
  Goldbeter, and Genevi{\`e}ve Dupont.
\newblock Cell fate specification based on tristability in the inner cell mass
  of mouse blastocysts.
\newblock {\em Biophysical journal}, 110(3):710--722, 2016.

\bibitem{epigene}
En~Li.
\newblock Chromatin modification and epigenetic reprogramming in mammalian
  development.
\newblock {\em Nature Reviews Genetics}, 3(9):662--673, 2002.

\bibitem{Tristability-in-miRNA-TF-Switch}
Mingyang Lu, Mohit~Kumar Jolly, Ryan Gomoto, Bin Huang, Jose Onuchic, and Eshel
  Ben-Jacob.
\newblock Tristability in cancer-associated microrna-tf chimera toggle switch.
\newblock {\em The journal of physical chemistry B}, 117(42):13164--13174,
  2013.

\bibitem{microRNA-based}
Mingyang Lu, Mohit~Kumar Jolly, Herbert Levine, Jos{\'e}~N Onuchic, and Eshel
  Ben-Jacob.
\newblock Microrna-based regulation of epithelial--hybrid--mesenchymal fate
  determination.
\newblock {\em Proceedings of the National Academy of Sciences},
  110(45):18144--18149, 2013.

\bibitem{CombCoop-ov}
Mubasher Rashid, Brasanna~M Devi, and Malay Banerjee.
\newblock Combinatorial cooperativity in mir200-zeb feedback network can
  control epithelial-mesenchymal transition.
\newblock {\em Bulletin of Mathematical Biology}, 86(5):1--23, 2024.

\bibitem{chick}
Vijay Chickarmane, Carl Troein, Ulrike~A Nuber, Herbert~M Sauro, and Carsten
  Peterson.
\newblock Transcriptional dynamics of the embryonic stem cell switch.
\newblock {\em PLoS computational biology}, 2(9):e123, 2006.

\bibitem{genetic-bistable-switch}
Daniel Huang, William~J Holtz, and Michel~M Maharbiz.
\newblock A genetic bistable switch utilizing nonlinear protein degradation.
\newblock {\em Journal of biological engineering}, 6:1--13, 2012.

\bibitem{Cherry&Adler2000}
Joshua~L Cherry and Frederick~R Adler.
\newblock How to make a biological switch.
\newblock {\em Journal of theoretical biology}, 203(2):117--133, 2000.

\bibitem{Jules_BMB}
Jules Guilberteau, Camille Pouchol, and Nastassia Pouradier~Duteil.
\newblock Monostability and bistability of biological switches.
\newblock {\em Journal of Mathematical Biology}, 83(6):65, 2021.

\bibitem{Zhang_mono&BiTrans-specific-case}
Jie Li and Weinian Zhang.
\newblock Transition between monostability and bistability of a genetic toggle
  switch in escherichia coli.
\newblock {\em Discrete \& Continuous Dynamical Systems-Series B}, 25(5), 2020.

\bibitem{mir200-Zeb_motor}
Simone Brabletz and Thomas Brabletz.
\newblock The zeb/mir-200 feedback loop—a motor of cellular plasticity in
  development and cancer?
\newblock {\em EMBO reports}, 11(9):670--677, 2010.

\bibitem{mir200-zeb_genetic-dissection}
Alexandra~C Title, Sue-Jean Hong, Nuno~D Pires, Lynn Hasen{\"o}hrl, Svenja
  Godbersen, Nadine Stokar-Regenscheit, David~P Bartel, and Markus Stoffel.
\newblock Genetic dissection of the mir-200--zeb1 axis reveals its importance
  in tumor differentiation and invasion.
\newblock {\em Nature communications}, 9(1):4671, 2018.

\bibitem{logics}
Menghan Chen and Ruiqi Wang.
\newblock Computational analysis of synergism in small networks with different
  logic.
\newblock {\em Journal of Biological Physics}, 49(1):1--27, 2023.

\bibitem{perko_2013}
Lawrence Perko.
\newblock Differential equations and dynamical systems.
\newblock 2013.

\bibitem{CD4-cell-differ-tristability}
Jinfang Zhu, Hidehiro Yamane, and William~E Paul.
\newblock Differentiation of effector cd4 t cell populations.
\newblock {\em Annual review of immunology}, 28:445--489, 2009.

\bibitem{wang2020dynamic}
Guiyuan Wang, Zhuoqin Yang, and Marc Turcotte.
\newblock Dynamic analysis of the time-delayed genetic regulatory network
  between two auto-regulated and mutually inhibitory genes.
\newblock {\em Bulletin of Mathematical Biology}, 82:1--30, 2020.

\bibitem{xi2015parameter}
Hongguang Xi and Marc Turcotte.
\newblock Parameter asymmetry and time-scale separation in core genetic
  commitment circuits.
\newblock {\em Quantitative Biology}, 3:19--45, 2015.

\bibitem{hogan-time-delay}
Kiresh Parmar, Konstantin~B Blyuss, Yuliya~N Kyrychko, and Stephen~John Hogan.
\newblock Time-delayed models of gene regulatory networks.
\newblock {\em Computational and mathematical methods in medicine}, 2015, 2015.

\bibitem{oper-princp-tristable_switches}
Dongya Jia, Mohit~Kumar Jolly, William Harrison, Marcelo Boareto, Eshel
  Ben-Jacob, and Herbert Levine.
\newblock Operating principles of tristable circuits regulating cellular
  differentiation.
\newblock {\em Physical biology}, 14(3):035007, 2017.

\bibitem{kishore-coupled-bistable-switches}
Kishore Hari, Pradyumna Harlapur, Aditi Gopalan, Varun Ullanat,
  Atchuta~Srinivas Duddu, and Mohit~Kumar Jolly.
\newblock Emergent properties of coupled bistable switches.
\newblock {\em Journal of Biosciences}, 47(4):81, 2022.

\bibitem{delaydiffeq-natcomm}
David~S Glass, Xiaofan Jin, and Ingmar~H Riedel-Kruse.
\newblock Nonlinear delay differential equations and their application to
  modeling biological network motifs.
\newblock {\em Nature communications}, 12(1):1788, 2021.

\bibitem{ruiqi-iden-crit-inter}
Qing Hu, Min Luo, and Ruiqi Wang.
\newblock Identifying critical regulatory interactions in cell fate decision
  and transition by systematic perturbation analysis.
\newblock {\em Journal of Theoretical Biology}, 577:111673, 2024.

\bibitem{genetic-dissect}
Alexandra~C Title, Sue-Jean Hong, Nuno~D Pires, Lynn Hasen{\"o}hrl, Svenja
  Godbersen, Nadine Stokar-Regenscheit, David~P Bartel, and Markus Stoffel.
\newblock Genetic dissection of the mir-200--zeb1 axis reveals its importance
  in tumor differentiation and invasion.
\newblock {\em Nature communications}, 9(1):4671, 2018.

\bibitem{hysteric-control}
T~Celi{\`a}-Terrassa, C~Bastian, DD~Liu, B~Ell, NM~Aiello, Y~Wei, J~Zamalloa,
  AM~Blanco, X~Hang, D~Kunisky, et~al.
\newblock Hysteresis control of epithelial-mesenchymal transition dynamics
  conveys a distinct program with enhanced metastatic ability. nat commun.
  2018; 9: 5005.

\end{thebibliography}

	\end{document}